\documentclass{article}

\usepackage{microtype}
\usepackage{graphicx}
\usepackage{subfigure}
\usepackage{booktabs} %

\usepackage{hyperref}

\usepackage[accepted]{icml2024}

\usepackage{amsmath}
\usepackage{amssymb}
\usepackage{mathtools}
\usepackage{amsthm}
\usepackage{xspace}
\usepackage{nicefrac}

\usepackage[capitalize,noabbrev]{cleveref}
\usepackage{wrapfig}

\theoremstyle{plain}
\newtheorem{theorem}{Theorem}[section]
\newtheorem{proposition}[theorem]{Proposition}
\newtheorem{lemma}[theorem]{Lemma}
\newtheorem{corollary}[theorem]{Corollary}
\theoremstyle{definition}
\newtheorem{definition}[theorem]{Definition}

\theoremstyle{remark}

\DeclareMathOperator*{\argmax}{arg\,max}
\DeclareMathOperator*{\argmin}{arg\,min}

\newcommand{\CFC}{\ensuremath{\mathsf{CFC}}\xspace}
\newcommand{\OWT}{\ensuremath{\mathsf{OWT}}\xspace}
\newcommand{\MTS}{\ensuremath{\mathsf{MTS}}\xspace}
\newcommand{\CFL}{\ensuremath{\mathsf{CFL}}\xspace}
\newcommand{\CFLno}{convex function chasing with a long-term constraint\xspace}
\newcommand{\MAP}{\ensuremath{\mathsf{MAP}}\xspace}
\newcommand{\MAL}{\ensuremath{\mathsf{MAL}}\xspace}
\newcommand{\MALno}{online metric allocation with a long-term constraint\xspace}
\newcommand{\ALG}{\ensuremath{\mathtt{ALG}}\xspace}
\newcommand{\LALG}{\ensuremath{\mathtt{LALG}}\xspace}
\newcommand{\ALGone}{\ensuremath{\mathtt{ALG1}}\xspace}
\newcommand{\ADV}{\ensuremath{\mathtt{ADV}}\xspace}
\newcommand{\ROB}{\ensuremath{\mathtt{ROB}}\xspace}
\newcommand{\BL}{\ensuremath{\mathtt{Baseline}}\xspace}
\newcommand{\OPT}{\ensuremath{\mathtt{OPT}}\xspace}
\newcommand{\CLIP}{\ensuremath{\mathtt{CLIP}}\xspace}
\newcommand{\OFF}{\ensuremath{\mathtt{OFF}}\xspace}
\newcommand{\ON}{\ensuremath{\mathtt{ON}}\xspace}
\newcommand{\mbf}[1]{{\mathbf{#1}}}

\newcommand{\sref}[2]{\hyperref[#2]{#1 \ref{#2}}}

\definecolor{ao(english)}{rgb}{0.0, 0.5, 0.0}

\usepackage[textsize=tiny]{todonotes}

\usepackage{etoolbox}
\makeatletter
\patchcmd{\hyper@makecurrent}{%
    \ifx\Hy@param\Hy@chapterstring
        \let\Hy@param\Hy@chapapp
    \fi
}{%
    \iftoggle{inappendix}{%
        \@checkappendixparam{chapter}%
        \@checkappendixparam{section}%
        \@checkappendixparam{subsection}%
        \@checkappendixparam{subsubsection}%
        \@checkappendixparam{paragraph}%
        \@checkappendixparam{subparagraph}%
    }{}%
}{}{\errmessage{failed to patch}}

\newcommand*{\@checkappendixparam}[1]{%
    \def\@checkappendixparamtmp{#1}%
    \ifx\Hy@param\@checkappendixparamtmp
        \let\Hy@param\Hy@appendixstring
    \fi
}
\makeatletter

\newtoggle{inappendix}
\togglefalse{inappendix}

\apptocmd{\appendix}{\toggletrue{inappendix}}{}{\errmessage{failed to patch}}

\icmltitlerunning{Chasing Convex Functions with Long-term Constraints}

\begin{document}

\twocolumn[
\icmltitle{Chasing Convex Functions with Long-term Constraints}

\icmlsetsymbol{equal}{*}

\begin{icmlauthorlist}
\icmlauthor{Adam Lechowicz}{umass}
\icmlauthor{Nicolas Christianson}{caltech}
\icmlauthor{Bo Sun}{waterloo}
\icmlauthor{Noman Bashir}{mit}
\icmlauthor{Mohammad Hajiesmaili}{umass}
\icmlauthor{Adam Wierman}{caltech}
\icmlauthor{Prashant Shenoy}{umass}
\end{icmlauthorlist}

\icmlaffiliation{umass}{Manning College of Information and Computer Sciences, \textbf{University of Massachusetts Amherst}, USA.}
\icmlaffiliation{caltech}{Computing \& Mathematical Sciences, \textbf{California Institute of Technology}, USA.}
\icmlaffiliation{waterloo}{Cheriton School of Computer Science, \textbf{University of Waterloo}, Ontario, Canada.}
\icmlaffiliation{mit}{Computer Science \& Artificial Intelligence Laboratory, \textbf{Massachusetts Institute of Technology}, USA.}

\icmlcorrespondingauthor{Adam Lechowicz}{alechowicz@cs.umass.edu}

\icmlkeywords{Machine Learning, ICML}

\vskip 0.3in
]

\printAffiliationsAndNotice{} %

\begin{abstract}
We introduce and study a family of online metric problems with long-term constraints.  In these problems, an online player makes decisions $\mathbf{x}_t$ in a metric space $(X,d)$ to simultaneously minimize their hitting cost $f_t(\mathbf{x}_t)$ and switching cost as determined by the metric.  Over the time horizon $T$, the player must satisfy a long-term demand constraint $\sum_{t} c(\mathbf{x}_t) \geq 1$, where $c(\mathbf{x}_t)$ denotes the fraction of demand satisfied at time $t$. Such problems can find a wide array of applications to online resource allocation in sustainable energy/computing systems. We devise optimal competitive and learning-augmented algorithms for the case of bounded hitting cost gradients and weighted $\ell_1$ metrics, and further show that our proposed algorithms perform well in numerical experiments.    
\end{abstract}

\section{Introduction}
\label{sec:intro}
This paper introduces and studies a novel class of online metric problems with \textit{long-term demand constraints} motivated by emerging applications in the design of sustainable systems.
In \emph{convex function chasing with a long-term constraint}, an online player aims to satisfy a demand by making decisions in a normed vector space, paying a hitting cost based on time-varying convex cost functions which are revealed online, and switching cost defined by the norm.  The player is constrained to ensure that the entire demand is satisfied at or before the time horizon $T$ ends, and their objective is to minimize their total cost.  The generality of this problem makes it applicable to a wide variety of online resource allocation problems; in this paper, we consider one such special case, discussing its connections to other online settings and suggestions towards broad new areas of inquiry in \textit{online optimization with long-term constraints}.

Our motivation to introduce these problems is rooted in an emerging class of \textit{carbon-aware} control problems for sustainable systems.  A shared objective involves minimizing carbon emissions by shifting flexible workloads temporally and/or spatially to better leverage low-carbon electricity generation (e.g., renewables such as solar and wind).  Examples which have recently seen significant interest include carbon-aware electric vehicle (EV) charging~\cite{Cheng:22} and carbon-aware compute shifting~\cite{Wiesner:21,bashir2021enabling,radovanovic2022carbon,acun2022holistic,Hanafy:23:CarbonScaler}.

The problems we introduce in this paper build on a long line of related work in online algorithms.  Most existing work can be roughly classified into two types: \textit{online metric problems}, where many works consider multidimensional decision spaces and switching costs but do not consider long-term constraints~\cite{Borodin:92, Koutsoupias:09, chenSmoothedOnlineConvex2018a, bubeckChasingNestedConvex2019, bubeckMetricalTaskSystems2021, Bansal:22, bubeckRandomizedServerConjecture2022}, and \textit{online search problems}, which feature long-term demand constraints but do not consider multidimensional decision spaces or switching costs~\cite{ElYaniv:01, Lorenz:08, Mohr:14, SunZeynali:20}.  

We briefly review the direct precursors of our work below.
In the online metric literature, the problem we study is an extension of \textit{convex function chasing} (\CFC) introduced by~\citet{FriedmanLinial:93}, where an online player makes online decisions $\mbf{x}_t$ in a normed vector space $(X, \lVert \cdot \rVert)$ over a sequence of time-varying cost functions in order to minimize their total hitting and switching cost.  In the online search literature, the problem we study is a generalization of \textit{one-way trading} (\OWT) introduced by~\citet{ElYaniv:01}, in which an online player must sell an entire asset in fractional shares over a sequence of time-varying prices while maximizing their profit.

Despite extensive existing work in the online metric and online search tracks, few works simultaneously consider long-term demand constraints (as in \OWT) and movement/switching costs (as in \CFC).  The existing prior works~\cite{Lechowicz:23, Lechowicz:24} that consider both components are restricted to unidimensional decision spaces, as is typical in the online search literature.  However, generalizing from the unidimensional case is highly non-trivial; e.g., in convex function chasing with a long-term constraint, the problem cannot simply be decomposed over dimensions due to the shared constraint function and multidimensional switching cost.
Thus, in this work we tackle the following question: \textit{Is it possible to design algorithms for the studied problems that operate in multidimensional decision spaces while simultaneously considering long-term constraints, hitting costs, and switching costs?}

Although the aforementioned literature focuses on competitive algorithms in adversarial settings, there has recently been significant interest in moving beyond worst-case analysis, which can result in overly pessimistic algorithms.
The field of \textit{learning-augmented algorithms}~\cite{Lykouris:18, Purohit:18} has emerged as a paradigm for designing and analyzing algorithms that incorporate untrusted machine-learned advice to improve average-case performance without sacrificing worst-case performance bounds.  Such algorithms are evaluated through the metrics of \textit{consistency and robustness} (see Def.~\ref{def:constrob}). Recent studies have proposed learning-augmented algorithms for related problems, including convex function chasing~\cite{Christianson:22}, one-way trading~\cite{SunLee:21}, metrical task systems~\cite{Christianson:23MTS}, and online search~\cite{Lee:24}.
While the literature in each of these tracks considers a spectrum of different advice models%
, their results prompt a natural open question: \textit{Can we design algorithms for online metric problems with long-term constraints that effectively utilize untrusted advice (such as machine-learned predictions) to improve performance while preserving worst-case competitive guarantees?}
\vspace{-1em}

\paragraph{Contributions. }
Despite extensive prior literature on adjacent problems, the problems we propose in this paper are the first online settings to combine long-term demand constraints with multidimensional decision spaces and switching costs.  
We introduce convex function chasing with a long-term constraint, and a special case called \textit{online metric allocation with a long-term constraint}.  The general forms of both are independently interesting for further study.

We obtain positive results for both of the questions posed above under problem instantiations %
that are especially relevant for motivating applications (namely, weighted $\ell_1$ norms, and cost functions with bounded gradients).  We provide the first competitive results for online problems of this form in \autoref{sec:comp}, and show that our proposed algorithm (\sref{Algorithm}{alg:roro}) achieves the best possible competitive ratio.  In \autoref{sec:clip}, we propose a learning-augmented algorithm, \CLIP (\sref{Algorithm}{alg:clip}), and show it achieves the provably optimal trade-off between consistency and robustness.

To achieve these results, the proposed algorithms must tackle technical challenges distinct from prior work on adjacent problems.  
Motivated by difficulties in directly applying algorithms for unidimensional decision spaces in the online search literature, we build on a generalization of threshold-based design called \textit{pseudo-cost minimization}.  While this framework is well-known in the online search literature, it is a new idea in the context of online metric problems and multidimensional decision spaces (see \autoref{sec:comp}).  %

Our learning-augmented algorithm \CLIP uses a novel adaptive optimization-based approach to achieve specific target consistency and robustness bounds.  In recent years, there has been interest in understanding fundamental trade-offs between consistency and robustness and designing algorithms that exactly match those trade-offs, which has proven to be non-trivial~\cite{Wei:20}.  \CLIP's approach directly incorporates the lower bound and exactly matches it, distinguishing it from algorithms that achieve, e.g., an asymptotically optimal trade-off.  To achieve this, \CLIP introduces a \textit{projected consistency constraint} designed to guarantee consistency against the advice \ADV by continuously comparing solutions in terms of cost incurred so far, switching cost trajectories, and the projected worst-case cost required to complete the long-term constraint. We believe \CLIP's high-level approach is applicable to other problems and may improve current results in the broader field of learning-augmented algorithms.

\vspace{-1em}

\section{Problem Formulation and Preliminaries} 
\label{sec:prob}
This section formalizes convex function chasing and online metric allocation with long-term constraints, motivating them with a sustainability application. We also provide preliminaries used throughout the paper, and give initial results to build algorithmic connections between both problems. 
\vspace{-0.5em}

\paragraph{Convex function chasing with a long-term constraint.}
A player chooses decisions $\mbf{x}_t \in X \subseteq \mathbb{R}^d$ %
online from a normed vector space $(X, \lVert \cdot \rVert)$ in order to minimize their total cost $\sum_{t=1}^T f_t(\mbf{x}_t) + \sum_{t=1}^{T+1} \lVert \mbf{x}_t - \mbf{x}_{t-1} \rVert$, where $f_t(\cdot) : X \rightarrow \mathbb{R}$ %
is a convex ``hitting'' cost that is revealed just before the player chooses $\mbf{x}_t$, and $\lVert \mbf{x}_t - \mbf{x}_{t-1} \rVert$ is a switching cost associated with changing decisions between rounds.
Additionally, the player must satisfy a long term constraint of the form $\sum_{t=1}^T c(\mbf{x}_t) = 1$, where $c(\mbf{x}) : X \to [0,1]$ 
gives the fraction of the constraint satisfied by a decision $\mbf{x}$.  
The offline version of the problem is formalized as follows:
\begingroup
\allowdisplaybreaks
\begin{align} 
\min_{\{\mbf{x}_t\}_{ t \in [T] }} &\underbrace{ \ \sum\nolimits_{t=1}^T f_t( \mbf{x}_t ) }_{\text{Convex hitting cost}} + \underbrace{ \sum\nolimits_{t=1}^{T+1} \lVert \mbf{x}_t - \mbf{x}_{t-1} \rVert }_{\text{Switching cost}} \label{align:objMin} \\
\text{s.t.}\quad & \underbrace{\sum\nolimits_{t=1}^T c(\mbf{x}_t) \geq 1,}_{\text{Long-term constraint}}\\
&\mbf{x}_t^i \in [0, 1]  \ \forall i \in [d], \ \forall t \in [T]. \label{align:pos}
\end{align}
\endgroup
We denote the \textit{utilization} at time $t$ by $z^{(t)} = \sum_{\tau=1}^t c(\mbf{x}_{\tau})$, which gives the total fraction of the long-term constraint satisfied up to and including time $t$.
\paragraph{Assumptions.}

Here, we describe the precise variant of \CFLno for which we design algorithms in the remainder of the paper.  
{\it Let $\lVert \mbf{x} - \mbf{x}' \rVert \coloneqq \lVert \mbf{x} - \mbf{x}' \lVert_{\ell_1 (\mbf{w})}$, where $\lVert \cdot \lVert_{\ell_1 (\mbf{w})}$ denotes the weighted $\ell_1$ norm with weight vector $\mbf{w} \in \mathbb{R}^d$.  

We define the long-term constraint such that $c(\mbf{x}) \coloneqq \lVert \mbf{x} \rVert_{\ell_1( \mbf{c})}$, i.e., the weighted $\ell_1$ norm with weight vector $\mbf{c} \in \mathbb{R}^d$.  Then let the metric space $X$ be the $\ell_1$ ball defined by $X \coloneqq \{ \mbf{x} \in \mathbb{R}^d : c(\mbf{x}) \leq 1 \}$.  Note that $\mbf{x}$ is restricted to lie in the positive orthant by \eqref{align:pos}.

For all cost functions $f_t(\cdot) : X \rightarrow \mathbb{R}$, we assume bounded gradients such that $L \leq \nicefrac{ \left[ \nabla f_t \right]^i } {\mbf{c}^i} \leq U \ \forall i \in [d], t \in [T]$, where $i$ denotes the $i$th dimension of the corresponding vector, and $L, U$ are known positive constants.  This also gives as a corollary that $f_t(\mbf{x}) \geq 0$  for any valid $\mbf{x}$.

Letting $\mbf{0}$ denote the origin in $\mathbb{R}^d$ (w.l.o.g), we have the property $f_t(\mbf{0}) = 0$ for all $t \in [T]$, i.e., that ``satisfying none of the long-term constraint costs nothing'', since $c(\mbf{0}) = 0$.  We assume the player starts and ends at the origin, i.e., $\mbf{x}_0 = \mbf{0}$ and $\mbf{x}_{T+1} = \mbf{0}$, to enforce switching ``on'' and ``off.'' These assumptions are intuitive and reasonable in practice, e.g., in our example motivating application below.

For analysis, it will be useful to establish a shorthand for the magnitude of the switching cost. 
Let $\beta \coloneqq \max \left( \nicefrac{\mbf{w}^i}{\mbf{c}^i} \right)$, which gives the greatest magnitude of the switching cost coefficient when normalized by the constraint function.  We assume that $\beta$ is bounded on the interval $[0, \nicefrac{U-L}{2})$; if $\beta$ is ``large'' (i.e., $> \nicefrac{U-L}{2}$), we can show that the player should prioritize minimizing the switching cost.\footnote{As brief justification for the bounds on $\beta$, consider that a feasible solution may have objective value $L + 2\beta$.  If $\beta > \nicefrac{U-L}{2}$, $L + 2\beta > U$, and we argue that the incurred switching cost is more important than the cost functions accepted.
} 

Recall that the long-term constraint must be satisfied before the sequence ends.  If the player has satisfied $z^{(t)}$ of the constraint at time $t$, we assume a \textit{compulsory trade} begins at time $j$ as soon as $\left( T - (j+1) \right) \cdot \mbf{c}^i < \left( 1 - z^{(j)} \right) \ \forall i \in [d]$ (i.e., when the time steps after $j$ are not sufficient to satisfy the constraint by moving to a point at the boundary defined by \eqref{align:pos}).  During this compulsory trade, a cost-agnostic algorithm takes over, making maximal decisions to satisfy the constraint.  $T$ is unknown in advance, but the player is notified when this compulsory trade should begin.\footnote{Note that pre-notification of $T$ is necessary to ensure constraint satisfaction -- if $T$ is completely unknown, the algorithm has very little flexibility to do anything without risking a violation of the long-term constraint.}  For the problem to remain technically interesting, we assume that this compulsory trade is a small portion of the sequence.\footnote{We assume the first time $j'$ where $\left( T - (j'+1)\right) \mbf{c}^i < 1 \ \forall i$ satisfies $j' \gg 1$, implying that $T$ and $\mbf{c}$ are both appropriate for the constraint.  This is reasonable for our motivating applications, since short deadlines (small $T$) or low throughput (small $\mbf{c}^i \ \forall i$) imply that even offline solutions suffer a lack of flexibility in reducing the overall cost. }}

For brevity, we henceforth use \CFL to refer to the variant of convex function chasing with a long-term constraint under the assumptions outlined here.
While this setting is challenging theoretically and useful for real-world applications as outlined below, it is worth mentioning that our solution techniques in Sections~\ref{sec:comp} and \ref{sec:clip} do not heavily rely on the idiosyncrasies of, e.g., the $\ell_1$ norm.  Furthermore, as is common in the literature, the results in the rest of this paper can extend to other metrics by leveraging finite-dimension norm equivalencies~\cite{Johnson:20}.

\paragraph{An example motivating application.}
\CFL can model a variety of applications, including specific applications that motivate this study. Consider a \textit{carbon-aware temporal load shifting} application with heterogeneous servers.  Here, each of the $d$ dimensions corresponds to one of $d$ heterogeneous servers.  An algorithm makes decisions $\mbf{x}_t \in \mathbb{R}^d$, where $\mbf{x}_{t}^i \in [0,1]$ denotes the load of the $i^{\text{th}}$ server at time $t$.  The long-term constraint $\sum_{t=1}^T c(\mbf{x}_t) \geq 1$ enforces that an entire workload should be finished before time $T$, and each coefficient $\mbf{c}^i$ represents the throughput of the $i$th server.  Each cost function $f_t(\mbf{x}_t)$ represents the carbon emissions due to the electricity usage of the servers configured according to $\mbf{x}_t$, and the switching cost $\lVert \cdot \rVert_{\ell_1 (\mbf{w})}$ captures the carbon emissions overhead (e.g., extra latency) of pausing, resuming, scaling, and moving the workload between servers.

Our motivation for the assumptions placed on \CFL are deeply rooted in this and other similar applications, where the switching cost and constraint function are both typically best modeled as a linear (i.e., $\ell_1$) function, and bounds on the marginal hitting cost (i.e., the bounded gradient assumption) are reasonable to obtain.

\paragraph{Online metric allocation with a long-term constraint.}
\citet{Bansal:22} introduced the online metric allocation problem (\MAP), which connects several online metric problems.  \MAP on a star metric is equivalent to \CFC when cost functions are separable over dimensions and supported on the unit simplex $\Delta_n$.\footnote{Given metric space $X$, consider $\Delta(X)$, which represents the set of probability measures over the points of $X$. Since $X$ is finite, we have that $\lvert X \rvert = n$ and $\Delta(X)$ is denoted as $\Delta_n$.}  Furthermore, the randomized metrical task systems problem (\MTS) is a special case of \MAP when cost functions are linear and increasing.

We build on this formulation in our setting and introduce \textit{online metric allocation with a long-term constraint}, which captures a particularly interesting special case of \CFL.
The general version of the problem considers an $n$-point metric space $(X, d)$, and a unit resource which can be allocated in arbitrary fractions to the points of $X$.  At each time $t \in [T]$, convex cost functions $f_t^a(\cdot) : [0,1] \rightarrow \mathbb{R}$ arrive at each point $a$ in the metric space.  The online player chooses an allocation $x_t^a$ to each point $a$ in the metric space, such that $\sum_{a=1}^n x_t^a = 1$ for all $t \in [T]$. When changing this allocation between time steps, the player pays a switching cost defined by $d(a, b)$ for any distinct points $a, b \in X$.
As in \CFL, the long-term constraint enforces that $\sum_{t=1}^T c(\mbf{x}_t) \geq 1$, where $c(\mbf{x})$ is a linear and separable function of the form $c(\mbf{x}) = \sum_{a = 1}^n \mbf{c}^a x^a$.  As previously, the player's objective is to minimize the total cost (hitting plus switching costs) incurred while satisfying the long-term constraint.

\paragraph{Assumptions.}
In the rest of the paper, we consider an instantiation of \MALno on \textit{weighted star metrics} that is particularly relevant to a wide class of resource allocation problems.

{\it To ensure the long-term constraint is non-trivial, we denote at least one point $a'$ in the metric space as an ``\OFF state'', where $\mbf{c}^{a'} = 0$ and $f_t^{a'}(x) = 0 \ \ \forall t \in [T], \forall x \in [0,1]$.  We do not require that an \OFF state be situated at the center of the star, although this condition creates a useful special case of \MAL (see \sref{Lemma}{lem:simplexTransform} for details).  For all other cost functions, we carry forward the assumptions that $L \leq \left( \nicefrac{\frac{d f_t^a}{dx^a}}{\mbf{c}^a} \right) \leq U, f_t^a(0) = 0 \ \ \forall t \in [T]$.  We define $\beta \coloneqq \max_{a', a} \nicefrac{d( a', a)}{\mbf{c}^a}$, i.e., the maximum distance between any \OFF state and any \ON state in the weighted star, normalized by the value of $\mbf{c}^a$ at a specific \ON state.  We inherit the same assumption that $\beta \in [0, \nicefrac{U-L}{2})$.}  For brevity, we henceforth use \MAL to refer to the problem on weighted star metrics with the assumptions described above.

\paragraph{Competitive analysis. } 
{Our goal is to design an algorithm that guarantees a small \textit{competitive ratio}~\cite{Manasse:88, Borodin:92}, i.e., performs nearly as well as the offline optimal solution.  
Formally, let $\mathcal{I} \in \Omega$ denote a valid input sequence, where $\Omega$ is the set of all feasible inputs for the problem.  Let $\OPT(\mathcal{I})$ denote the cost of an optimal offline solution for instance $\mathcal{I}$, and let $\ALG(\mathcal{I})$ denote the cost incurred by running an online algorithm $\ALG$ over the same instance.  The competitive ratio is then defined as $
\textnormal{CR}(\ALG) \coloneqq \sup_{\mathcal{I} \in \Omega} \nicefrac{ \ALG(\mathcal{I}) }{ \OPT(\mathcal{I}) } = \eta
$, and \ALG is said to be $\eta$-\textbf{competitive}.
Note that $\textnormal{CR}(\ALG)$ is always $\geq 1$, and a \textit{lower} competitive ratio implies that the online algorithm is guaranteed to be \textit{closer} to the offline optimal solution.}

\paragraph{Learning-augmented consistency and robustness.}
In the emerging literature on learning-augmented algorithms, competitive ratio is interpreted via the notions of \textit{consistency} and \textit{robustness}, introduced by~\cite{Lykouris:18, Purohit:18}.  

\begin{definition} \label{def:constrob}
{\it Let $\LALG$ denote a learning-augmented online algorithm provided with advice denoted by $\ADV$.
Then $\LALG$ is said to be $b$-\textbf{consistent} if it is $b$-competitive with respect to $\ADV$.  Conversely, $\LALG$ is $r$-\textbf{robust} if it is $r$-competitive with respect to $\OPT$ when given any $\ADV$ (i.e., regardless of the performance of $\ADV$). }
\end{definition}
\vspace{-1em}

\paragraph{A connection between \CFL and \MAL.}
Below we state two useful results connecting the \CFL and \MAL settings that influence our algorithm design for each problem. 

\begin{lemma} \label{lem:simplexTransform}
    For any \MAL instance on a weighted star metric $(X, d)$, there is a corresponding \CFL instance on $(\mathbb{R}^{n-1}, \lVert \cdot \rVert_{\ell_1 (\mbf{w'})})$ that preserves $f_t^a(\cdot) \ \forall t, c(\cdot) \ \forall a \in X$. Furthermore, for any points $(a,b) \in X$, their distance is upper bounded by a weighted $\ell_1$ norm between the corresponding vectors $(\mbf{a}, \mbf{b}) \in \mathbb{R}^{n-1}$, i.e., $d(a,b) \leq \lVert \mbf{a} - \mbf{b} \rVert_{\ell_1 (\mbf{w'})}$.
    If the \MAL instance contains an \OFF state at the center of the weighted star $X$, distances are preserved exactly, i.e., $d(a,b) = \lVert \mbf{a} - \mbf{b} \rVert_{\ell_1 (\mbf{w'})}$.
\end{lemma}
Leveraging \sref{Lemma}{lem:simplexTransform}, the following result explicitly connects the competitive results of the \CFL and \MAL settings.
\begin{proposition} \label{prop:CFL-MAL}
Given an algorithm \ALG for \CFL, any competitive bound for \ALG gives an identical competitive bound for \MAL with parameters corresponding to the \CFL instance constructed in \sref{Lemma}{lem:simplexTransform}.
\end{proposition}
The proofs of both are deferred to \autoref{apx:alphaCompMAL}.  At a high-level, \sref{Proposition}{prop:CFL-MAL} shows that if $\ALG$ is $\eta$-competitive against $\OPT$ which pays no switching cost, \sref{Lemma}{lem:simplexTransform} implies it is also $\eta$-competitive on \MAL.
In the next section, our proposed algorithms will be presented using \CFL notation, but these results provide the necessary condition which allows them to solve \MAL as well.

\section{Designing Competitive Algorithms}
\label{sec:comp}
In this section, we present our robust algorithm design. %
We start by discussing some inherent challenges in the problem, highlighting reasons why existing algorithms fail.  Next, we %
introduce a generalization of techniques from online search called pseudo-cost minimization, which underpins our competitive algorithm, \ALGone (\sref{Algorithm}{alg:roro}).  Finally, we state (and prove in \autoref{apx:comp}) two bounds, which jointly imply that \ALGone achieves the optimal competitive ratio for \CFL and \MAL.
\paragraph{Challenges.}
Canonical algorithms for \CFC~\cite{chenSmoothedOnlineConvex2018a,sellkeChasingConvexBodies2020,Zhang:21} make decisions that attempt to minimize (or nearly minimize) the hitting cost of cost functions $f_t(\cdot)$ and switching cost across all time steps.  As discussed in the introduction, the structure of the problem with a long-term constraint means that such myopic cost-minimization algorithms will fail in general.  To illustrate this, consider the actions of a minimizer-driven algorithm on an arbitrary sequence with length $T$.  For each $t < T$, the algorithm chooses a point at or near $\mbf{0}$, since $\mbf{0}$ is the minimizer of each $f_t$.  However, since $c(\mbf{0}) = 0$, such an algorithm must subsequently satisfy all or almost all of the long-term constraint during the \textit{compulsory trade}, incurring an arbitrarily bad hitting cost.

This challenge motivates an algorithm design that balances between the two extremes of finishing the long-term constraint ``immediately'' (i.e., at early time steps), and finishing the long-term constraint ``when forced to'' (i.e., during the compulsory trade).  Both extremes result in a poor competitive ratio.  Algorithms in the online search literature (e.g., online knapsack, \OWT) leverage a \textit{threshold-based design} to address precisely this problem, as in~\cite{Zhou:08, SunZeynali:20, Lechowicz:24}.  However, such threshold-based algorithms are traditionally derived for unidimensional decision spaces with no switching costs.  

Generalizing beyond unidimensional decision spaces proves to be non-trivial -- consider the following example.  Towards the application of carbon-aware load shifting, it may reasonable to expect that cost functions are separable over dimensions.  It is thus reasonable to consider whether an existing unidimensional algorithm (i.e., as shown by~\citet{Lechowicz:24}) can solve the problem via \textit{decomposition}, e.g., by running $d$ instances of the unidimensional algorithm.  However, such a technique fails from the perspective of competitive analysis, because the necessary coupling between $d$ independently determined unidimensional decisions and the multidimensional long-term constraint is broken.  A multidimensional switching cost complicates such a decomposition even further, as decisions in one dimension increase cost in other dimensions, potentially beyond the competitive threshold used to determine a decision in the first place.

In what follows, we describe a \textit{pseudo-cost minimization approach}, which generalizes the threshold-based design to operate in the \CFL setting.  A key enabling result for this approach is the ability to simultaneously consider all dimensions for both the hitting costs and the long-term constraint, leveraging the definition of the constraint function $c(\mbf{x})$.

\begin{algorithm}[t]
   \caption{Pseudo-cost minimization algorithm (\ALGone)}
   \label{alg:roro}
{\small
\begin{algorithmic}
   \STATE {\bfseries input:} long-term constraint function $c(\cdot)$, distance metric $\lVert \cdot \rVert_{\ell_1 (\mbf{w})}$, pseudo-cost threshold function $\phi(z)$
   \STATE {\bfseries initialize:} $z^{(0)} = 0;$
   \WHILE {cost function $f_t (\cdot)$ is revealed and $z^{(t-1)} < 1$}
   \STATE solve pseudo-cost minimization problem: {\small \begin{align}
       \mbf{x}_t = \argmin_{\mbf{x} \in X : c(\mbf{x}) \leq 1-z^{(t-1)}} &f_t(\mbf{x}) + \lVert \mbf{x} - \mbf{x}_{t-1} \rVert_{\ell_1 (\mbf{w})} \nonumber\\
       &- \int_{z^{(t-1)}}^{z^{(t-1)} + c(\mbf{x})} \phi (u) du \label{eq:pseudocostRORO}
    \end{align}}
   \STATE update utilization $z^{(t)} = z^{(t-1)} + c(\mbf{x}_t)$
   \ENDWHILE
\end{algorithmic}}
\end{algorithm}
\setlength{\textfloatsep}{6pt}%
\setlength{\intextsep}{0pt}

\paragraph{Algorithm description. } 
Recall that $z^{(t)}$ gives the fraction of the long-term constraint satisfied at time $t$.  Building off of the intuition of threshold-based design, we define a function $\phi$, which will be used to compute a \textit{pseudo-cost minimization} problem central to our robust algorithm.

\begin{definition}[Pseudo-cost threshold function $\phi$ for \CFL]\label{dfn:phi-min}
{\it For any utilization $z \in [0,1]$, $\phi$ is defined as:\vspace{-0.3em}
\begin{equation}\label{eq:phi-min}
\phi(z) = U - \beta + (\nicefrac{U}{\alpha} - U + 2 \beta) \exp(\nicefrac{z}{\alpha}),
\end{equation}\vspace{-1em}
where $\alpha$ is the competitive ratio and is defined in \eqref{eq:alpha}.}
\end{definition}

Then our algorithm (\sref{Algorithm}{alg:roro}, referred to as \ALGone) solves the pseudo-cost minimization problem defined in \eqref{eq:pseudocostRORO} to obtain a decision $\mbf{x}_t$ at each time step. At a high level, the inclusion of $\phi$ in this pseudo-cost problem enforces that, upon arrival of a cost function, the algorithm satisfies ``just enough'' of the long-term constraint. Concretely, the structure of the $\phi$ function enforces that $\phi(z^{(t)}) - \beta$ corresponds to the ``best cost function seen so far''.  Then, if a good cost function arrives, the pseudo-cost minimization problem solves for the $\mbf{x}_t$ which guarantees a competitive ratio of $\alpha$ against the current estimate of $\OPT$.

At a glance, it is not obvious that the minimization problem in \eqref{eq:pseudocostRORO} is tractable; however, in \autoref{apx:pseudo-convex}, we show that the problem is \textit{convex}, implying that it can be solved efficiently.  In \autoref{thm:alphaCompCFL}, we state the competitive result for \ALGone.  We discuss the significance of the result below, and relegate the full proof to \autoref{apx:alphaCompCFL}.

\begin{theorem} \label{thm:alphaCompCFL}
\ALGone is $\alpha$-competitive for \CFL, where $\alpha$ is the solution to $\frac{U - L - 2 \beta}{U - \nicefrac{U}{\alpha} - 2 \beta} = \exp(1/\alpha)$, given by \vspace{-0.5em}
\begin{equation}
    \alpha \coloneqq \left[ W \left( \left( \frac{2\beta}{U} + \frac{L}{U} -1 \right) e^{\frac{2\beta}{U}-1} \right) - \frac{2\beta}{U} + 1 \right]^{-1} \kern-1em \label{eq:alpha},
\end{equation}
where $W$ is the Lambert $W$ function~\cite{Corless:96LambertW}.
\end{theorem}
Intuitively, parameters of \CFL ($L$, $U$, and $\beta$) appear in the competitive bound.  While results for \OWT and \CFC are not directly comparable with \CFL results, we discuss a few connections here.

When $\beta \to 0$, $\alpha$ matches the optimal competitive ratio of $\left[ W \left( \left( \nicefrac{L}{U} - 1 \right) e^{-1} \right) + 1 \right]^{-1}$ for the minimization variant of \OWT~\cite{Lorenz:08, SunLee:21}.  
In the intermediate case (i.e., when $\beta \in (0, \nicefrac{U-L}{2})$), Taylor expanding the expression for $\alpha$ yields a leading order approximation of $O \left( \sqrt{\nicefrac{U}{L}} \right) + O \left( \beta \right)$ -- thus, \CFL adds a new \textit{linear} dependence on $\beta$ compared to \OWT.
Furthermore, as $\beta \to \nicefrac{U-L}{2}$, $\alpha$ approaches $\nicefrac{U}{L}$, which is the competitive ratio achievable by e.g., a myopic cost minimization algorithm.  In the Appendix, in \autoref{fig:plotting-alpha}, we plot $\alpha$ as a function of these dependencies (i.e., $\nicefrac{U}{L}$, $\beta$).

Since $\alpha$ does not feature a dependence on the dimension $d$ of the vector space, we note a connection with \CFC: 
it is known that ``dimension-free'' bounds are achievable in \CFC with necessary structural assumptions on the hitting cost that are evocative of our bounded gradient assumptions~\cite{argueDimensionFreeBoundsChasing2020}. In particular, \citet{chenSmoothedOnlineConvex2018a} give an algorithm for \CFC that achieves a dimension-independent competitive ratio on $\ell_2$ metrics when hitting cost functions increase away from the minimizer at a certain rate.  They extend this to a dimension-dependent bound for general metrics using norm equivalency results.

Via \sref{Proposition}{prop:CFL-MAL}, we obtain an immediate corollary to \autoref{thm:alphaCompCFL} which gives the following competitive bound when \ALGone is used to solve \MAL.  The full proof of \sref{Corollary}{cor:alphaCompMAL} can be found in \autoref{apx:alphaCompMAL}.
\begin{corollary} \label{cor:alphaCompMAL}
\ALGone is $\alpha$-competitive for \MAL.
\end{corollary}
\paragraph{On the tightness of competitive ratios. }  It is important to highlight that the bounds in \autoref{thm:alphaCompCFL} and \sref{Corollary}{cor:alphaCompMAL} are the first competitive bounds for any variant of convex function chasing or online metric allocation imbued with long-term constraints. %
A natural follow-up question concerns whether any online algorithm for \CFL (or \MAL) can achieve a better competitive bound.  In the following, we answer this question in the negative, showing that \ALGone's competitive ratio is the best that any deterministic online algorithm for \CFL and/or \MAL can achieve.  We state the result here, and defer the full proof to \autoref{apx:lowerboundCFL}.

\begin{theorem} \label{thm:lowerboundCFL}
For any $L, U,$ and $\beta \in [0, \nicefrac{U-L}{2})$, there exists a family of \CFL instances such that any deterministic online algorithm for \CFL is at least $\alpha$-competitive, where $\alpha$ is as defined in \eqref{eq:alpha}.
\end{theorem}\vspace{-0.5em}

Since \ALGone is $\alpha$-competitive by \autoref{thm:alphaCompCFL}, this implies that \ALGone achieves the optimal competitive ratio for \CFL.  Furthermore, by leveraging \sref{Lemma}{lem:simplexTransform}, this result gives an immediate corollary result in the \MAL setting by constructing a corresponding family of \MAL instances, which forces any algorithm to achieve a competitive ratio of $\alpha$.  We state the result here, deferring the full proof to \autoref{apx:lowerboundMAL}.

\begin{corollary} \label{cor:lowerboundMAL}
The \CFL instances in \autoref{thm:lowerboundCFL} correspond to instances of \MAL such that any deterministic online algorithm for \MAL is at least $\alpha$-competitive.
\end{corollary}\vspace{-0.5em}

As previously, since \ALGone is $\alpha$-competitive by \sref{Corollary}{cor:alphaCompMAL}, it achieves the optimal competitive ratio for \MAL.  We note that beyond the settings of \CFL and \MAL considered in this paper, \autoref{thm:lowerboundCFL} and \sref{Corollary}{cor:lowerboundMAL} are the first lower bound results for convex function chasing and online metric allocation with long-term constraints, and may thus give useful insight into the achievable competitive bounds for different or more general settings of these problems.

\section{Learning-augmented Algorithms}
\label{sec:clip}
In this section, we %
consider how \textit{untrusted black-box advice} can help improve the average-case performance of a learning-augmented algorithm for \CFL and \MAL while retaining worst-case guarantees. We first consider a sub-optimal ``baseline'' algorithm that directly combines advice with a robust algorithm such as \ALGone.  We then propose a unified algorithm called \CLIP, which integrates advice more efficiently and achieves the optimal trade-off between \textit{consistency} and \textit{robustness} (\sref{Definition}{def:constrob}).\vspace{-0.5em}

\paragraph{Advice model.}
For a \CFL or \MAL instance $\mathcal{I} \in \Omega$, let $\ADV$ denote untrusted black-box decision advice, i.e., $\ADV \coloneqq \{ \mbf{a}_t \in X : t \in [T] \}$.  If the advice is correct, it achieves the optimal objective value (i.e., $\ADV(\mathcal{I}) = \OPT(\mathcal{I})$). %

\paragraph{A simple baseline.}  \citet{Lechowicz:24} show that a straightforward ``fixed-ratio'' learning-augmented approach works well in practice for unidimensional online search with switching costs.  Here we show that a similar technique (playing a convex combination of the solutions chosen by the advice and a robust algorithm) achieves bounded but sub-optimal consistency and robustness for \CFL.  

Let $\ROB \coloneqq \{ \Tilde{\mbf{x}}_t : t \in [T] \}$ denote the actions of a robust algorithm for \CFL (e.g., \ALGone).  For any value $\epsilon \in (0, \alpha - 1]$, the fixed-ratio algorithm (denoted as \BL for brevity) sets a combination factor $\lambda \coloneqq \frac{\alpha - 1 - \epsilon}{\alpha - 1}$.  Then at each time step, \BL makes a decision according to $\mbf{x}_t = \lambda \mbf{a}_t + (1- \lambda) \Tilde{\mbf{x}}_t$.  We present consistency and robustness results for \BL below, deferring the full proof to \autoref{apx:baseline}.
\begin{lemma} \label{lem:baseline}
Letting \ROB denote the actions of \ALGone and setting a parameter $\epsilon \in (0, \alpha -1]$, \BL is $(1+\epsilon)$-consistent and $\left( \frac{\nicefrac{(U + 2 \beta)}{L} (\alpha - 1 - \epsilon) + \alpha \epsilon}{( \alpha - 1) } \right)$-robust for \CFL.
\end{lemma}

Although this fixed-ratio algorithm verifies that an algorithm for \CFL can utilize untrusted advice to improve performance, it remains an open question of whether the trade-off between consistency and robustness given in \sref{Lemma}{lem:baseline} is optimal.  Thus, we study whether a learning-augmented algorithm for \CFL can be designed which \textit{does} achieve the provably optimal consistency-robustness trade-off.  In the next section, we start by considering a more sophisticated method of incorporating advice into an algorithm design.%

\begin{algorithm*}[t!]
   \caption{\textbf{C}onsistency \textbf{Li}mited \textbf{P}seudo-cost minimization (\CLIP)}
   \label{alg:clip}
{\small
\begin{algorithmic}
   \STATE {\bfseries input:} consistency parameter $\epsilon$, long-term constraint function $c(\cdot)$, pseudo-cost threshold function $\phi^\epsilon(\cdot)$
   \STATE {\bfseries initialize:} $z^{(0)} = 0; \ p^{(0)} = 0; \ A^{(0)} = 0; \ \CLIP_0 = 0; \ \ADV_0 = 0$
   \WHILE {cost function $f_t (\cdot)$ is revealed, untrusted advice $\mbf{a}_t$ is revealed, and $z^{(t-1)} < 1$}
   \STATE update advice cost $\ADV_t = \ADV_{t-1} + f_t(\mbf{a}_t) + \lVert \mbf{a}_t - \mbf{a}_{t-1} \rVert_{\ell_1 (\mbf{w})}$ and advice utilization $A^{(t)} = A^{(t-1)} + c(\mbf{a}_t)$
   \STATE solve \textit{constrained} pseudo-cost minimization problem:\vspace{-1em} {\scriptsize \begin{align}
        & \mbf{x}_t = \argmin_{\mbf{x} \in X : c(\mbf{x}) \leq 1-z^{(t-1)}} f_t(\mbf{x}) + \lVert \mbf{x} - \mbf{x}_{t-1} \rVert_{\ell_1 (\mbf{w})} - \int_{p^{(t-1)}}^{p^{(t-1)} + c(\mbf{x})} \phi^\epsilon (u) du \label{eq:pseudocostCLIP} \\
       \begin{split}
       &\text{\small such that}\\
        &\CLIP_{t-1} + f_t(\mbf{x}) + \lVert \mbf{x} - \mbf{x}_{t-1} \rVert_{\ell_1 (\mbf{w})} + \lVert \mbf{x} - \mbf{a}_{t} \rVert_{\ell_1 (\mbf{w})} + \lVert \mbf{a}_{t} \rVert_{\ell_1 (\mbf{w})} + (1 - z^{(t-1)} - c(\mbf{x}))L + \max( (A^{(t)} - z^{(t-1)} - c(\mbf{x})), \ 0)(U-L)\\ 
       & \hspace{43.5em} \leq (1+\epsilon) [\ADV_{t} + \lVert \mbf{a}_t \rVert_{\ell_1 (\mbf{w})} + (1 - A^{(t)} ) L]
       \end{split} \label{eq:const-constraint}
    \end{align}}
   \STATE update cost $\CLIP_t = \CLIP_{t-1} + f_t(\mbf{x}_t) + \lVert \mbf{x}_t - \mbf{x}_{t-1} \rVert_{\ell_1 (\mbf{w})}$ and utilization $z^{(t)} = z^{(t-1)} + c(\mbf{x}_t)$
   \STATE solve \textit{unconstrained} pseudo-cost minimization problem:\vspace{-1em} {\scriptsize \begin{align}
       \bar{\mbf{x}}_t &= \argmin_{\mbf{x} \in X : c(\mbf{x}) \leq 1-z^{(t-1)}} f_t(\mbf{x}) + \lVert \mbf{x} - \mbf{x}_{t-1} \rVert_{\ell_1 (\mbf{w})} - \int_{p^{(t-1)}}^{p^{(t-1)} + c(\mbf{x})} \phi^\epsilon (u) du
    \end{align}}\vspace{-1em}
    \STATE update pseudo-utilization $p^{(t)} = p^{(t-1)} + \min ( c(\bar{\mbf{x}}_t), c(\mbf{x}_t))$
   \ENDWHILE
\end{algorithmic}}
\end{algorithm*}

\paragraph{An optimal learning-augmented algorithm.}
We present \CLIP (\textbf{C}onsistency-\textbf{Li}mited \textbf{P}seudo-cost minimization, \sref{Algorithm}{alg:clip}) which achieves the optimal trade-off between consistency and robustness for \CFL.

To start, for any $\epsilon \in (0, \alpha - 1]$, we define a corresponding \textit{target robustness factor} $\gamma^\epsilon$, the unique positive solution to:

\vspace{-2em}
\begin{equation}
    \gamma^\epsilon = \epsilon + \frac{U}{L} - \frac{\gamma^\epsilon}{L} (U-L) \ln \left( \frac{U - L - 2\beta}{U - \nicefrac{U}{\gamma^\epsilon} - 2\beta} \right). \label{eq:gamma}
\end{equation} 

Note that $\gamma^{\alpha - 1} = \alpha$, and $\gamma^{0} = \nicefrac{U}{L}$.  
We use $\gamma^\epsilon$ to define a pseudo-cost threshold function $\phi^\epsilon$ that will be used in a minimization problem to choose a decision at each step of the \CLIP algorithm.  

\begin{definition}[Pseudo-cost threshold function $\phi^\epsilon$]\label{dfn:clip-phi}
\hspace{3em}{\it
Given $\gamma^\epsilon$ from \eqref{eq:gamma}, $\phi^\epsilon(p)$ for $p \in [0,1]$ is defined as:}
\begin{align}\label{eq:clipphi}
        \phi^\epsilon(p) & = U - \beta + (\nicefrac{U}{\gamma^\epsilon} - U + 2\beta) \exp (\nicefrac{p}{\gamma^\epsilon}).
    \end{align}
\end{definition}
For each time step $t \in [T]$, we define a \textit{pseudo-utilization} $p^{(t)} \in [0,1]$, where $p^{(t)} \le z^{(t)} \ \forall t$, and $p^{(t)}$ describes the fraction of the long-term constraint which been satisfied ``robustly'' (as defined by the pseudo-cost) at time $t$.

Then \CLIP (see \sref{Algorithm}{alg:clip}) solves a \textit{constrained} pseudo-cost minimization problem (defined in \eqref{eq:pseudocostCLIP}) %
to obtain a decision $\mbf{x}_t$ at each time step.  The objective of this problem is mostly inherited from \ALGone, but the inclusion of a \textit{consistency constraint} allows the framework to accommodate untrusted advice for bounded consistency and robustness.

The high-level intuition behind this consistency constraint (defined in \eqref{eq:const-constraint}) is to directly compare the solutions of \CLIP and \ADV \textit{so far}, while \textit{``hedging''} against worst-case scenarios which may cause \CLIP to violate the desired $(1+\epsilon)$-consistency.
We introduce some notation to simplify the expression of the constraint.  We let $\CLIP_t$ denote the cost of $\CLIP$ up to time $t$, i.e., $\CLIP_t \coloneqq \sum_{\tau=1}^{t} f_{\tau}(\mbf{x}_\tau) + \lVert \mbf{x}_{\tau} - \mbf{x}_{\tau-1} \rVert_{\ell_1 (\mbf{w})}$.  Similarly, we let $\ADV_t \coloneqq \sum_{\tau=1}^{t} f_{\tau}(\mbf{a}_{\tau}) + \lVert \mbf{a}_\tau - \mbf{a}_{\tau-1} \rVert_{\ell_1 (\mbf{w})}$ denote the cost of $\ADV$ up to time $t$.  Additionally, we let $A^{(t)}$ denote the utilization of $\ADV$ at time $t$, i.e., $A^{(t)} \coloneqq \sum_{\tau=1}^{t} c(\mbf{a}_{\tau})$

The constraint defined in \eqref{eq:const-constraint} considers the cost of both \CLIP and \ADV so far, and the current hitting and switching cost $f_t(\mbf{x}) + \lVert \mbf{x} - \mbf{x}_{t-1} \rVert_{\ell_1 (\mbf{w})}$, ensuring that $(1+\epsilon)$-consistency is preserved.  Both sides of the constraint also include terms which consider the cost of potential future situations.  First, $\lVert \mbf{x} - \mbf{a}_t \rVert_{\ell_1 (\mbf{w})} + \lVert \mbf{a}_t \rVert_{\ell_1 (\mbf{w})}$ ensures that if \CLIP pays a switching cost to follow \ADV and/or pays a switching cost to ``switch off'' (move to $\mbf{0}$) in e.g., the next time step, that cost has been paid for ``in advance''.  As $\mbf{x}_{T+1} = \mbf{0}$, the constraint also charges $\ADV$ in advance for the mandatory switching cost at the end of the sequence $\left( \lVert \mbf{a}_t \rVert_{\ell_1 (\mbf{w})} \right)$; this ensures that there is always a feasible setting of $\mbf{x}_t$.

In the term $\left( 1 - A^{(t)} \right) L $, the consistency constraint assumes that $\ADV$ can satisfy the rest of the long-term constraint at the best marginal cost $L$.  Respectively, in the term $(1 - z^{(t-1)} - c(\mbf{x}))L + \max( (A^{(t)} - z^{(t-1)} - c(\mbf{x})), \ 0)(U-L)$, the constraint assumes $\CLIP$ can satisfy \textit{up to} $\left( 1 - A^{(t)} \right)$ of the remaining long-term constraint at the best cost $L$, but any excess (i.e., $(A^{(t)} - z^{(t)})$) must be satisfied at the worst cost $U$ (e.g., during the compulsory trade).  This worst-case assumption ensures that when actual hitting costs replace the above terms, the desired $(1+\epsilon)$-consistency holds.

At each step, \CLIP also solves an \textit{unconstrained} pseudo-cost minimization problem to obtain $\bar{\mbf{x}}_t$, which updates the pseudo-utilization $p^{(t)}$.  This ensures that when \ADV has accepted a cost function which would not be accepted by the unconstrained pseudo-cost minimization, the threshold function $\phi^\epsilon$ can ``start from zero'' in subsequent time steps.

At a high level, \CLIP's consistency constraint combined with the pseudo-cost minimization generates decisions which are \textit{as robust as possible} while preserving consistency. 
In \autoref{thm:constrobCLIP}, we state the consistency and robustness of \CLIP; we relegate the full proof to \autoref{apx:constrobCLIP}.

\begin{theorem} \label{thm:constrobCLIP}
For any $\epsilon \in [0, \alpha - 1]$, \CLIP is $(1+\epsilon)$-consistent and $\gamma^\epsilon$-robust for \CFL ($\gamma^\epsilon$ as defined in \eqref{eq:gamma}).
\end{theorem}

The previous result gives an immediate corollary when \CLIP is used to solve \MAL, which we state below.  The full proof of \sref{Corollary}{cor:constrobCLIPMAL} can be found in \autoref{apx:constrobCLIPMAL}.

\begin{corollary} \label{cor:constrobCLIPMAL}
For any $\epsilon \in [0, \alpha - 1]$, \CLIP is $(1+\epsilon)$-consistent and $\gamma^\epsilon$-robust for \MAL.
\end{corollary}

\paragraph{Optimal trade-offs between robustness and consistency. }
Although the trade-off given by \CLIP implies that achieving $1$-consistency requires a large robustness bound of $\nicefrac{U}{L}$ in the worst-case, in the following theorem we show that this is the best we can obtain from any consistent and robust algorithm.  We state the result and discuss its significance here, deferring the full proof to \autoref{apx:optimalconstrobCFL}.

\begin{theorem} \label{thm:optimalconstrobCFL}
Given untrusted advice \ADV and $\epsilon \in (0, \alpha-1]$, any $(1 + \epsilon)$-consistent learning-augmented algorithm for \CFL is at least $\gamma^\epsilon$-robust, where $\gamma^\epsilon$ is defined in \eqref{eq:gamma}.
\end{theorem}

This result implies that \CLIP achieves the \textit{optimal} trade-off between consistency and robustness for \CFL.  Furthermore, via \sref{Lemma}{lem:simplexTransform}, this result immediately gives \sref{Corollary}{cor:optimalconstrobMAL}, which we state here and prove in \autoref{apx:optimalconstrobMAL}.

\begin{corollary} \label{cor:optimalconstrobMAL}
Any $(1 + \epsilon)$-consistent learning-augmented algorithm for \MAL is at least $\gamma^\epsilon$-robust ($\gamma^\epsilon$ defined by \eqref{eq:gamma}).
\end{corollary} 

As previously, this implies \CLIP achieves the optimal consistency-robustness trade-off for \MAL.  Beyond the settings of \CFL and \MAL, these Pareto-optimality results may give useful insight into the achievable consistency-robustness trade-offs for more general settings.\vspace{-0.5em}

\section{Numerical Experiments}
\vspace{1em}
\begin{figure}[h]
    \small
    \minipage{0.225\textwidth}
	\includegraphics[width=\linewidth]{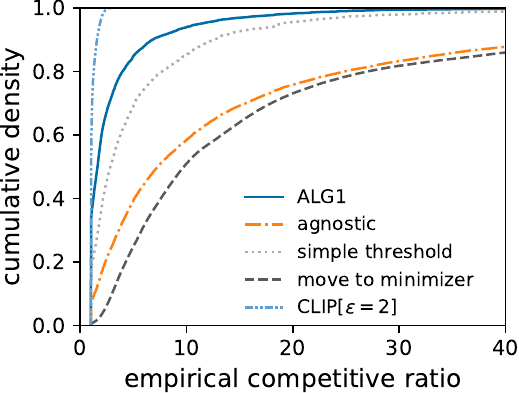}\vspace{-2em}
    \caption{CDFs of empirical competitive ratios for various algorithms.} \label{fig:cdf}
	\endminipage\hfill
    \minipage{0.235\textwidth}
	\includegraphics[width=\linewidth]{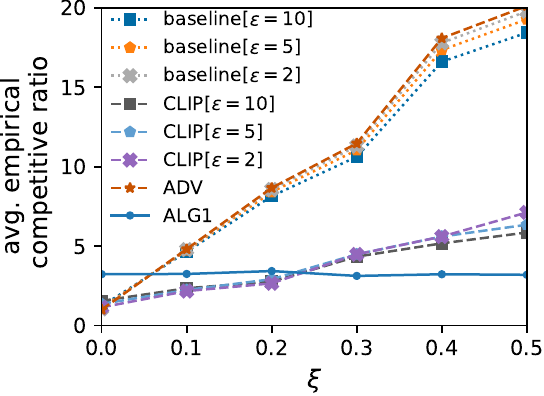}\vspace{-2em}
    \caption{Varying adversarial factor $\xi$, with $\nicefrac{U}{L}$ = $250, \beta$ = $50, d$ = $5$, and $\sigma$ = $50$.}\label{fig:xi}
	\endminipage
    \vspace{1em}
\end{figure}
\label{sec:exp}
In this section, we conduct numerical experiments on synthetic \CFL instances.  We evaluate \ALGone and \CLIP against the offline optimal solution, three heuristics adapted from related work, and the learning-augmented \BL.

\begin{figure*}[t!]
    \small
    \begin{center} 
    \minipage{\textwidth}
    \centering
    \includegraphics[width=0.66\linewidth]{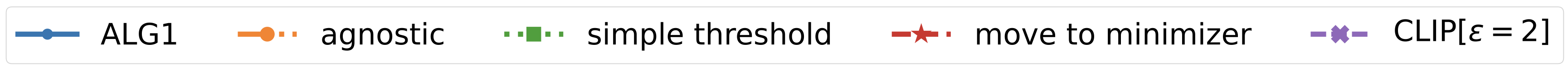}
    \vspace{-0.2cm}
    \endminipage\hfill
    \end{center}
	\minipage{0.24\textwidth}
	\includegraphics[width=\linewidth]{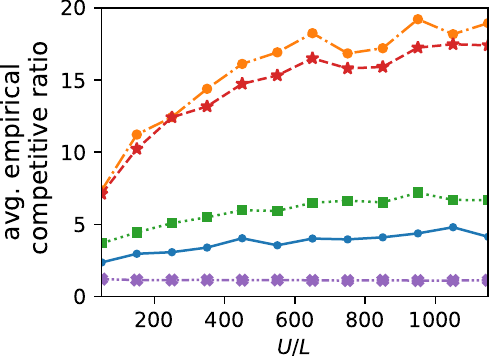}\vspace{-2em}
    \caption{Varying $\nicefrac{U}{L}$, with $\beta = \nicefrac{U}{5}, d = 5, \xi = 0$, and $\sigma = \nicefrac{U}{5}$.} \label{fig:ul}
	\endminipage\hfill
    \minipage{0.24\textwidth}
	\includegraphics[width=\linewidth]{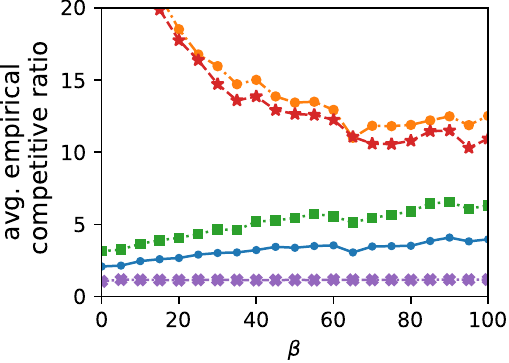}\vspace{-2em}
    \caption{Varying $\beta$, with $\nicefrac{U}{L} = 250, d = 5, \xi = 0$, and $\sigma = 50$.} \label{fig:beta}
	\endminipage \hfill
 	\minipage{0.24\textwidth}
	\includegraphics[width=\linewidth]{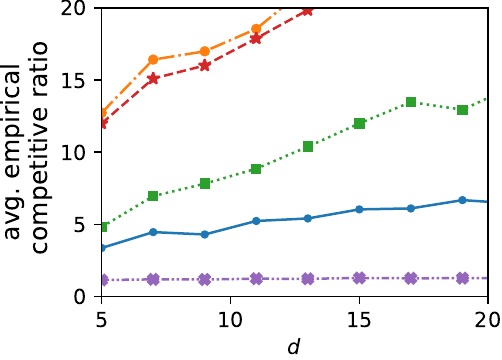}\vspace{-2em}
    \caption{Varying $d$ with $\beta$ = $50, \nicefrac{U}{L}$ = $250, \sigma$ = $50,$ and $\xi$=$0$.}\label{fig:dim}
	\endminipage\hfill
    \minipage{0.24\textwidth}
	\includegraphics[width=\linewidth]{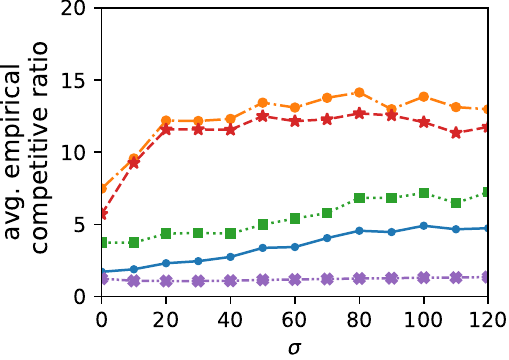}\vspace{-2em}
    \caption{Varying $\sigma$, with $\beta$ = $50, \nicefrac{U}{L}$ = $250, d$ = $5$, and $\xi$ = $0$.}\label{fig:var}
	\endminipage
\end{figure*}

\paragraph{Setup.} 
We construct a $d$-dimensional decision space, where $d$ is picked from the set $\{5, 7, \dots, 21\}$. The competitive ratio of our proposed algorithms depends on both $\nicefrac{U}{L}$ and $\beta = \max_i \mathbf{w^i}$, as the switching cost. Hence, we evaluate their performance over the range of these parameters. We set different cost fluctuation ratios $\nicefrac{U}{L} \in \{ 50, 150, \dots, 1250 \}$ by setting $L$ and $U$ accordingly, and $\beta$ is picked from the set $\beta \in \{ 0, 5, \dots, \nicefrac{U}{2.5}\}$.  For each experiment, $c(\mbf{x}) = \lVert \mbf{x} \rVert_1$.

For a given setting of $d$, $\nicefrac{U}{L}$, and $\beta$, we generate 1,000 random instances as follows.  First, each term of the weight vector $\mbf{w}$ for the weighted $\ell_1$ norm is drawn randomly from the uniform distribution on $[0, \beta]$.
Next, the time horizon $T$ is generated randomly from a uniform distribution on $[6, 24]$.
For each time $t \in [T]$, a cost function is generated as follows:  Let $f_t(\mbf{x}) = \mbf{f}_t^\intercal \mbf{x}$, where $\mbf{f}_t$ is a $d$-dimensional \textit{cost vector}.  To generate $\mbf{f}_t$, we first draw $\mu_t$ from the uniform distribution on $[L, U]$, and then draw each term of $\mbf{f}_t$ from a normal distribution centered at $\mu_t$ with standard deviation $\sigma$ (i.e., $\mbf{f}^i_t \thicksim \mathcal{N}(\mu_t, \sigma)$).  Any terms which are outside the assumed interval $[L,U]$ (i.e. $\mbf{f}^i_t < L$ or $\mbf{f}^i_t > U$) are truncated appropriately.
For each instance, we report the empirical competitive ratios as the evaluation metric, comparing the tested algorithms against an offline optimal benchmark.

In the setting with advice, we obtain simulated advice as follows:  Let $\xi \in [0,1]$ denote an \textit{adversarial factor}.  When $\xi = 0$, \ADV gives the optimal solution, and when $\xi = 1$, \ADV is fully adversarial.  Formally, letting $\{ \mbf{x}_t^\star : t \in [T]\}$ denote the decisions made by an optimal solution, and letting ${\{ \breve{\mbf{x}_t} : t \in [T]\}}$ represent the decisions made by a solution which maximizes the objective (rather than minimizing it), we have that $\ADV = \{ (1-\xi) \mbf{x}_t^\star + \xi \breve{\mbf{x}_t} : t \in [T]\}$.  We note that although $\{ \breve{\mbf{x}_t} : t \in [T]\}$ is adversarial from the perspective of the objective, it is still a feasible solution for the problem (i.e., it satisfies the long-term constraint). 

\paragraph{Comparison algorithms. } 
We use CVXPY~\cite{CVXPY} to compute the offline optimal solution for each instance using a convex optimization solver with access to all cost functions in advance.  
Since we are the first to study \CFL, there are no directly comparable algorithms with competitive guarantees -- thus, we consider three heuristic techniques based on literature for adjacent problems. 

The first technique is termed ``\textbf{agnostic}'', which chooses the minimum dimension of the cost function in the first time step $t=1$ (i.e., $k = \argmin_{i \in [d]} \mbf{c}_1^i$), sets $\mbf{x}_1^k = 1$, and $\mbf{x}_t = \mbf{0} \ \forall t > 1$.
The second technique is termed ``\textbf{move to minimizer}'', which takes inspiration from algorithms for \CFC~\cite{Zhang:21} and satisfies $\nicefrac{1}{T}$ fraction of the long-term constraint at each time step by moving to the minimum dimension of each cost function.  Formally, at each time step $t$, letting $k_t = \argmin_{i \in [d]} \mbf{c}_t^i$, ``move to minimizer'' sets $\mbf{x}_t^{k_t} = \nicefrac{1}{T}$. 
Finally, the third technique is termed ``\textbf{simple threshold}'', which takes inspiration from algorithms for online search~\cite{ElYaniv:01}.  This algorithm sets a fixed threshold $\psi = \sqrt{UL}$, and completes the long-term constraint at the first time step and dimension where the hitting cost does not exceed $\psi$.  Formally, at the first time step $\tau$ satisfying $\exists \ k \in [d] : \mbf{f}_{\tau}^k \le \psi$, ``simple threshold'' sets $\mbf{x}_{\tau}^{k} = 1$.
In the setting with advice, we compare our proposed \CLIP learning-augmented algorithm against the \BL learning-augmented algorithm described in \autoref{sec:clip} (e.g., \sref{Lemma}{lem:baseline}).

\paragraph{Experimental results. } 
\autoref{fig:cdf} summarizes the main results for \ALGone, the comparison heuristic algorithms, and one setting of \CLIP ($\epsilon = 2$) in a CDF plot of the empirical competitive ratios across several experiments.  In these experiments, we fix $\nicefrac{U}{L} = 250, \xi = 0, \sigma = 50$, while varying $\beta$ and $d$.  Notably, \ALGone outperforms in both average-case and worst-case performance, improving on the closest ``simple threshold'' by an average of 18.2\%, and outperforming ``agnostic'' and ``move to minimizer'' by averages of 56.1\% and 71.5\%, respectively.  With correct advice, \CLIP sees significant performance gains across the board.

In \autoref{fig:ul}-\ref{fig:var}, we investigate the impact of parameters on the average empirical competitive ratio for each algorithm.  In \autoref{apx:95th}, we give corresponding plots for the $95$th percentile (``worst-case'') results.  
\autoref{fig:ul} plots competitive ratios for different values of $\nicefrac{U}{L}$. We fix $\beta = \nicefrac{U}{5}, d = 5, \xi = 0, \sigma = \nicefrac{U}{5}$, while varying $\nicefrac{U}{L}$.  Since there is a dependence on $\nicefrac{U}{L}$ in our competitive results, the performance of \ALGone degrades as $\nicefrac{U}{L}$ grows, albeit at a favorable pace compared to the heuristics.  \autoref{fig:beta} plots competitive ratios for different values of $\beta$.  We fix $\nicefrac{U}{L} = 250, d = 5, \xi = 0, \sigma = 50$.  As $\beta$ grows, the ``agnostic'' and ``move to minimizer'' heuristics improve because the switching cost paid by $\OPT$ grows.  %

In \autoref{fig:var}, we plot competitive ratios for different values of $\sigma$. We fix $\nicefrac{U}{L} = 250, \beta = 50, d = 5,\xi = 0$, while varying $\sigma$.  As cost functions become more variable, the performance of all algorithms degrades, with the exception of \CLIP.  There is a plateau as $\sigma$ grows, because a large $\sigma$ implies that more terms in each $\mbf{f}_t$ must be truncated to the interval $[L,U]$.
Finally, \autoref{fig:dim} plots competitive ratios for different values of $d$. We fix $\nicefrac{U}{L} = 250, \beta = 50, \xi = 0,\sigma = 50$, while varying $d$.  As $d$ grows, \ALGone and \CLIP's performance degrades slower compared to the heuristics, as predicted by their \textit{dimension-free} theoretical bounds.\footnote{While the plot suggests a correlation between $d$ and \ALGone's performance, this seems to be a side effect of the random cost functions, which indirectly introduce more dimension-wise variability as $d$ increases.  There is a large gap between empirical performance and the theoretical upper bounds -- the asymptotic performance of \ALGone would be constant for large $d$.}

\autoref{fig:xi} plots the effect of prediction error on the learning-augmented algorithms \CLIP and \BL.  We test several values of $\xi \in [0,\nicefrac{1}{2}]$ (recall that $\xi = 0$ recovers correct advice), while fixing $\nicefrac{U}{L} = 250, \beta = 50, d = 5$, and $\sigma = 50$.  We also test \BL and \CLIP for several values of $\epsilon \in \{2, 5, 10\}$ (note that $\ADV$ corresponds to \BL and \CLIP with $\epsilon = 0$).  Notably, we find that \CLIP significantly outperforms the \BL algorithm as $\xi$ grows, showing an average improvement of 60.8\% when $\xi > 0.1$.  This result implies that \CLIP is more empirically robust to prediction errors than the simple fixed ratio technique of \BL.

\section{Conclusion}
\label{sec:conclusion}
We study online metric problems with long-term constraints, motivated by emerging problems in sustainability.  These are the first such problems to concurrently incorporate multidimensional decision spaces, switching costs, and long-term demand constraints.  Our main results instantiate the \CFL and \MAL problems towards a motivating application.  We design competitive and learning-augmented algorithms, show that their performance bounds are tight, and validate them in numerical experiments.  Several interesting open questions are prompted by our work.  Specifically, (i) what is achievable in non-$\ell_{1}$ vector spaces e.g., the Euclidean setting, and (ii) can our results for \MAL inform algorithm designs for e.g., tree metrics, and by extension, arbitrary metric spaces?
\clearpage

\section*{Impact Statement}

This paper presents work whose goal is to advance the field of online optimization 
and study problems relevant to sustainability applications.
There are many potential societal consequences of our work, none which we feel must be explicitly highlighted here. 

\section*{Acknowledgements}
This research is supported by National Science Foundation grants CAREER-2045641, CNS-2102963, CNS-2106299, CNS-2146814, CNS-1518941, CNS-2106403, CPS-2136197, CPS-2136199, NGSDI-2105494, NGSDI-2105648, 1908298, 2020888, 2021693, 2045641, 2213636, 2211888, and an NSF Graduate Research Fellowship (DGE-2139433).

This material is based upon work supported by the U.S. Department of Energy, Office of Science, Office of Advanced Scientific Computing Research, Department of Energy Computational Science Graduate Fellowship under Award Number DE-SC0024386.

\section*{Disclaimers}
This report was prepared as an account of work sponsored by an agency of the United States Government. Neither the United States Government nor any agency thereof, nor any of their employees, makes any warranty, express or implied, or assumes any legal liability or responsibility for the accuracy, completeness, or usefulness of any information, apparatus, product, or process disclosed, or represents that its use would not infringe privately owned rights. Reference herein to any specific commercial product, process, or service by trade name, trademark, manufacturer, or otherwise does not necessarily constitute or imply its endorsement, recommendation, or favoring by the United States Government or any agency thereof. The views and opinions of authors expressed herein do not necessarily state or reflect those of the United States Government or any agency thereof.

\bibliography{main}

\begin{thebibliography}{36}
\providecommand{\natexlab}[1]{#1}
\providecommand{\url}[1]{\texttt{#1}}
\expandafter\ifx\csname urlstyle\endcsname\relax
  \providecommand{\doi}[1]{doi: #1}\else
  \providecommand{\doi}{doi: \begingroup \urlstyle{rm}\Url}\fi

\bibitem[Acun et~al.(2023)Acun, Lee, Kazhamiaka, Maeng, Gupta, Chakkaravarthy,
  Brooks, and Wu]{acun2022holistic}
Acun, B., Lee, B., Kazhamiaka, F., Maeng, K., Gupta, U., Chakkaravarthy, M.,
  Brooks, D., and Wu, C.-J.
\newblock {Carbon Explorer: A Holistic Framework for Designing Carbon Aware
  Datacenters}.
\newblock In \emph{Proceedings of the 28th ACM International Conference on
  Architectural Support for Programming Languages and Operating Systems, Volume
  2}, ASPLOS 2023, pp.\  118–132, New York, NY, USA, 2023. Association for
  Computing Machinery.
\newblock ISBN 9781450399166.
\newblock \doi{10.1145/3575693.3575754}.
\newblock URL \url{https://doi.org/10.1145/3575693.3575754}.

\bibitem[Argue et~al.(2020)Argue, Gupta, and
  Guruganesh]{argueDimensionFreeBoundsChasing2020}
Argue, C.~J., Gupta, A., and Guruganesh, G.
\newblock {Dimension-Free Bounds for Chasing Convex Functions}.
\newblock In \emph{Proceedings of {{Thirty Third Conference}} on {{Learning
  Theory}}}, pp.\  219--241. {PMLR}, July 2020.

\bibitem[Bansal \& Coester(2022)Bansal and Coester]{Bansal:22}
Bansal, N. and Coester, C.
\newblock {Online Metric Allocation and Time-Varying Regularization}.
\newblock In Chechik, S., Navarro, G., Rotenberg, E., and Herman, G. (eds.),
  \emph{30th Annual European Symposium on Algorithms (ESA 2022)}, volume 244 of
  \emph{Leibniz International Proceedings in Informatics (LIPIcs)}, pp.\
  13:1--13:13, Dagstuhl, Germany, 2022. Schloss Dagstuhl -- Leibniz-Zentrum
  f{\"u}r Informatik.
\newblock ISBN 978-3-95977-247-1.
\newblock \doi{10.4230/LIPIcs.ESA.2022.13}.
\newblock URL
  \url{https://drops.dagstuhl.de/entities/document/10.4230/LIPIcs.ESA.2022.13}.

\bibitem[Bashir et~al.(2021)Bashir, Guo, Hajiesmaili, Irwin, Shenoy, Sitaraman,
  Souza, and Wierman]{bashir2021enabling}
Bashir, N., Guo, T., Hajiesmaili, M., Irwin, D., Shenoy, P., Sitaraman, R.,
  Souza, A., and Wierman, A.
\newblock {Enabling Sustainable Clouds: The Case for Virtualizing the Energy
  System}.
\newblock In \emph{Proceedings of the ACM Symposium on Cloud Computing}, SoCC
  '21, pp.\  350–358, New York, NY, USA, 2021. Association for Computing
  Machinery.
\newblock ISBN 9781450386388.
\newblock \doi{10.1145/3472883.3487009}.
\newblock URL \url{https://doi.org/10.1145/3472883.3487009}.

\bibitem[Borodin et~al.(1992)Borodin, Linial, and Saks]{Borodin:92}
Borodin, A., Linial, N., and Saks, M.~E.
\newblock {An Optimal On-Line Algorithm for Metrical Task System}.
\newblock \emph{J. ACM}, 39\penalty0 (4):\penalty0 745–763, Oct 1992.
\newblock ISSN 0004-5411.
\newblock \doi{10.1145/146585.146588}.
\newblock URL \url{https://doi.org/10.1145/146585.146588}.

\bibitem[Bubeck et~al.(2019)Bubeck, Klartag, Lee, Li, and
  Sellke]{bubeckChasingNestedConvex2019}
Bubeck, S., Klartag, B., Lee, Y.~T., Li, Y., and Sellke, M.
\newblock {Chasing Nested Convex Bodies Nearly Optimally}.
\newblock In \emph{Proceedings of the 2020 {{ACM-SIAM Symposium}} on {{Discrete
  Algorithms}} ({{SODA}})}, Proceedings, pp.\  1496--1508. {Society for
  Industrial and Applied Mathematics}, December 2019.
\newblock \doi{10.1137/1.9781611975994.91}.

\bibitem[Bubeck et~al.(2021)Bubeck, Cohen, Lee, and
  Lee]{bubeckMetricalTaskSystems2021}
Bubeck, S., Cohen, M.~B., Lee, J.~R., and Lee, Y.~T.
\newblock {Metrical Task Systems on Trees via Mirror Descent and Unfair
  Gluing}.
\newblock \emph{SIAM Journal on Computing}, 50\penalty0 (3):\penalty0 909--923,
  January 2021.
\newblock ISSN 0097-5397, 1095-7111.
\newblock \doi{10.1137/19M1237879}.

\bibitem[Bubeck et~al.(2023)Bubeck, Coester, and
  Rabani]{bubeckRandomizedServerConjecture2022}
Bubeck, S., Coester, C., and Rabani, Y.
\newblock {The Randomized \$k\$-{{Server Conjecture}} Is False!}
\newblock In \emph{Proceedings of the 55th Annual ACM Symposium on Theory of
  Computing (STOC 2023)}, STOC 2023, pp.\  581–594, New York, NY, USA, 2023.
  Association for Computing Machinery.
\newblock ISBN 9781450399135.
\newblock \doi{10.1145/3564246.3585132}.
\newblock URL \url{https://doi.org/10.1145/3564246.3585132}.

\bibitem[Chen et~al.(2018)Chen, Goel, and
  Wierman]{chenSmoothedOnlineConvex2018a}
Chen, N., Goel, G., and Wierman, A.
\newblock {Smoothed Online Convex Optimization in High Dimensions via Online
  Balanced Descent}.
\newblock In \emph{Proceedings of the 31st {{Conference On Learning Theory}}},
  pp.\  1574--1594. {PMLR}, July 2018.

\bibitem[Cheng et~al.(2022)Cheng, Bian, Shi, and Chen]{Cheng:22}
Cheng, K.-W., Bian, Y., Shi, Y., and Chen, Y.
\newblock {Carbon-Aware EV Charging}.
\newblock In \emph{2022 IEEE International Conference on Communications,
  Control, and Computing Technologies for Smart Grids (SmartGridComm)}, pp.\
  186--192, 2022.
\newblock \doi{10.1109/SmartGridComm52983.2022.9960988}.

\bibitem[Christianson et~al.(2022)Christianson, Handina, and
  Wierman]{Christianson:22}
Christianson, N., Handina, T., and Wierman, A.
\newblock {Chasing Convex Bodies and Functions with Black-Box Advice}.
\newblock In \emph{Proceedings of the 35th Conference on Learning Theory},
  volume 178, pp.\  867--908. PMLR, 02--05 Jul 2022.

\bibitem[Christianson et~al.(2023)Christianson, Shen, and
  Wierman]{Christianson:23MTS}
Christianson, N., Shen, J., and Wierman, A.
\newblock {Optimal robustness-consistency tradeoffs for learning-augmented
  metrical task systems}.
\newblock In \emph{International Conference on Artificial Intelligence and
  Statistics}, 2023.

\bibitem[Corless et~al.(1996)Corless, Gonnet, Hare, Jeffrey, and
  Knuth]{Corless:96LambertW}
Corless, R.~M., Gonnet, G.~H., Hare, D.~E., Jeffrey, D.~J., and Knuth, D.~E.
\newblock {On the Lambert W function}.
\newblock \emph{Advances in Computational mathematics}, 5:\penalty0 329--359,
  1996.

\bibitem[Diamond \& Boyd(2016)Diamond and Boyd]{CVXPY}
Diamond, S. and Boyd, S.
\newblock {CVXPY: A Python-embedded modeling language for convex optimization}.
\newblock \emph{Journal of Machine Learning Research}, 17\penalty0
  (83):\penalty0 1--5, 2016.

\bibitem[El-Yaniv et~al.(2001)El-Yaniv, Fiat, Karp, and Turpin]{ElYaniv:01}
El-Yaniv, R., Fiat, A., Karp, R.~M., and Turpin, G.
\newblock {Optimal Search and One-Way Trading Online Algorithms}.
\newblock \emph{Algorithmica}, 30\penalty0 (1):\penalty0 101--139, May 2001.
\newblock \doi{10.1007/s00453-001-0003-0}.
\newblock URL \url{https://doi.org/10.1007/s00453-001-0003-0}.

\bibitem[Friedman \& Linial(1993)Friedman and Linial]{FriedmanLinial:93}
Friedman, J. and Linial, N.
\newblock {On convex body chasing}.
\newblock \emph{Discrete \& Computational Geometry}, 9\penalty0 (3):\penalty0
  293--321, March 1993.
\newblock \doi{10.1007/bf02189324}.
\newblock URL \url{https://doi.org/10.1007/bf02189324}.

\bibitem[Hanafy et~al.(2023)Hanafy, Liang, Bashir, Irwin, and
  Shenoy]{Hanafy:23:CarbonScaler}
Hanafy, W.~A., Liang, Q., Bashir, N., Irwin, D., and Shenoy, P.
\newblock {CarbonScaler: Leveraging Cloud Workload Elasticity for Optimizing
  Carbon-Efficiency}.
\newblock \emph{Proceedings of the ACM on Measurement and Analysis of Computing
  Systems}, 7\penalty0 (3), Dec 2023.

\bibitem[Johnson(2020)]{Johnson:20}
Johnson, S.~G.
\newblock {N}otes on the equivalence of norms.
\newblock \url{https://math.mit.edu/~stevenj/18.335/norm-equivalence.pdf},
  2020.

\bibitem[Koutsoupias(2009)]{Koutsoupias:09}
Koutsoupias, E.
\newblock {The k-server problem}.
\newblock \emph{Computer Science Review}, 3\penalty0 (2):\penalty0 105--118,
  May 2009.
\newblock \doi{10.1016/j.cosrev.2009.04.002}.
\newblock URL \url{https://doi.org/10.1016/j.cosrev.2009.04.002}.

\bibitem[Kumar et~al.(2018)Kumar, Purohit, and Svitkina]{Purohit:18}
Kumar, R., Purohit, M., and Svitkina, Z.
\newblock {Improving Online Algorithms via ML Predictions}.
\newblock In Bengio, S., Wallach, H., Larochelle, H., Grauman, K.,
  Cesa-Bianchi, N., and Garnett, R. (eds.), \emph{Advances in Neural
  Information Processing Systems}, volume~31. Curran Associates, Inc., 2018.

\bibitem[Lechowicz et~al.(2023)Lechowicz, Christianson, Zuo, Bashir,
  Hajiesmaili, Wierman, and Shenoy]{Lechowicz:23}
Lechowicz, A., Christianson, N., Zuo, J., Bashir, N., Hajiesmaili, M., Wierman,
  A., and Shenoy, P.
\newblock {The Online Pause and Resume Problem: Optimal Algorithms and An
  Application to Carbon-Aware Load Shifting}.
\newblock \emph{Proceedings of the ACM on Measurement and Analysis of Computing
  Systems}, 7\penalty0 (3), Dec 2023.

\bibitem[Lechowicz et~al.(2024)Lechowicz, Christianson, Sun, Bashir,
  Hajiesmaili, Wierman, and Shenoy]{Lechowicz:24}
Lechowicz, A., Christianson, N., Sun, B., Bashir, N., Hajiesmaili, M., Wierman,
  A., and Shenoy, P.
\newblock {Online Conversion with Switching Costs: Robust and
  Learning-augmented Algorithms}, 2024.
\newblock URL \url{https://arxiv.org/abs/2310.20598}.

\bibitem[Lee et~al.(2024)Lee, Sun, Hajiesmaili, and Lui]{Lee:24}
Lee, R., Sun, B., Hajiesmaili, M., and Lui, J. C.~S.
\newblock {Online Search with Predictions: Pareto-optimal Algorithm and its
  Applications in Energy Markets}.
\newblock In \emph{Proceedings of the 15th ACM International Conference on
  Future Energy Systems}, e-Energy '24, New York, NY, USA, June 2024.
  Association for Computing Machinery.

\bibitem[Lorenz et~al.(2008)Lorenz, Panagiotou, and Steger]{Lorenz:08}
Lorenz, J., Panagiotou, K., and Steger, A.
\newblock {Optimal Algorithms for k-Search with Application in~Option Pricing}.
\newblock \emph{Algorithmica}, 55\penalty0 (2):\penalty0 311--328, August 2008.
\newblock \doi{10.1007/s00453-008-9217-8}.
\newblock URL \url{https://doi.org/10.1007/s00453-008-9217-8}.

\bibitem[Lykouris \& Vassilvtiskii(2018)Lykouris and
  Vassilvtiskii]{Lykouris:18}
Lykouris, T. and Vassilvtiskii, S.
\newblock {Competitive Caching with Machine Learned Advice}.
\newblock In Dy, J. and Krause, A. (eds.), \emph{Proceedings of the 35th
  International Conference on Machine Learning}, volume~80 of \emph{Proceedings
  of Machine Learning Research}, pp.\  3296--3305. PMLR, 10--15 Jul 2018.
\newblock URL \url{https://proceedings.mlr.press/v80/lykouris18a.html}.

\bibitem[Manasse et~al.(1988)Manasse, McGeoch, and Sleator]{Manasse:88}
Manasse, M., McGeoch, L., and Sleator, D.
\newblock {Competitive Algorithms for On-Line Problems}.
\newblock In \emph{Proceedings of the Twentieth Annual ACM Symposium on Theory
  of Computing}, STOC '88, pp.\  322–333, New York, NY, USA, 1988.
  Association for Computing Machinery.
\newblock ISBN 0897912640.
\newblock \doi{10.1145/62212.62243}.

\bibitem[Mitrinovic et~al.(1991)Mitrinovic, Pe{\v{c}}ari{\'c}, and
  Fink]{Mitrinovic:91}
Mitrinovic, D.~S., Pe{\v{c}}ari{\'c}, J.~E., and Fink, A.~M.
\newblock \emph{{Inequalities Involving Functions and Their Integrals and
  Derivatives}}, volume~53.
\newblock Springer Science \& Business Media, 1991.

\bibitem[Mohr et~al.(2014)Mohr, Ahmad, and Schmidt]{Mohr:14}
Mohr, E., Ahmad, I., and Schmidt, G.
\newblock {Online algorithms for conversion problems: A survey}.
\newblock \emph{Surveys in Operations Research and Management Science},
  19\penalty0 (2):\penalty0 87--104, July 2014.
\newblock \doi{10.1016/j.sorms.2014.08.001}.
\newblock URL \url{https://doi.org/10.1016/j.sorms.2014.08.001}.

\bibitem[Radovanovic et~al.(2022)Radovanovic, Koningstein, Schneider, Chen,
  Duarte, Roy, Xiao, Haridasan, Hung, Care, et~al.]{radovanovic2022carbon}
Radovanovic, A., Koningstein, R., Schneider, I., Chen, B., Duarte, A., Roy, B.,
  Xiao, D., Haridasan, M., Hung, P., Care, N., et~al.
\newblock {Carbon-Aware Computing for Datacenters}.
\newblock \emph{IEEE Transactions on Power Systems}, 2022.

\bibitem[Sellke(2020)]{sellkeChasingConvexBodies2020}
Sellke, M.
\newblock {Chasing Convex Bodies Optimally}.
\newblock In \emph{Proceedings of the {{Thirty-First Annual ACM-SIAM
  Symposium}} on {{Discrete Algorithms}}}, {{SODA}} '20, pp.\  1509--1518,
  {USA}, January 2020. {Society for Industrial and Applied Mathematics}.

\bibitem[Sun et~al.(2021{\natexlab{a}})Sun, Lee, Hajiesmaili, Wierman, and
  Tsang]{SunLee:21}
Sun, B., Lee, R., Hajiesmaili, M., Wierman, A., and Tsang, D.
\newblock {Pareto-Optimal Learning-Augmented Algorithms for Online Conversion
  Problems}.
\newblock In Ranzato, M., Beygelzimer, A., Dauphin, Y., Liang, P., and Vaughan,
  J.~W. (eds.), \emph{Advances in Neural Information Processing Systems},
  volume~34, pp.\  10339--10350. Curran Associates, Inc., 2021{\natexlab{a}}.

\bibitem[Sun et~al.(2021{\natexlab{b}})Sun, Zeynali, Li, Hajiesmaili, Wierman,
  and Tsang]{SunZeynali:20}
Sun, B., Zeynali, A., Li, T., Hajiesmaili, M., Wierman, A., and Tsang, D.~H.
\newblock {Competitive Algorithms for the Online Multiple Knapsack Problem with
  Application to Electric Vehicle Charging}.
\newblock \emph{Proceedings of the ACM on Measurement and Analysis of Computing
  Systems}, 4\penalty0 (3), June 2021{\natexlab{b}}.
\newblock \doi{10.1145/3428336}.
\newblock URL \url{https://doi.org/10.1145/3428336}.

\bibitem[Wei \& Zhang(2020)Wei and Zhang]{Wei:20}
Wei, A. and Zhang, F.
\newblock Optimal robustness-consistency trade-offs for learning-augmented
  online algorithms.
\newblock In \emph{Proceedings of the 34th International Conference on Neural
  Information Processing Systems}, NeurIPS '20, Red Hook, NY, USA, 2020. Curran
  Associates Inc.
\newblock ISBN 9781713829546.

\bibitem[Wiesner et~al.(2021)Wiesner, Behnke, Scheinert, Gontarska, and
  Thamsen]{Wiesner:21}
Wiesner, P., Behnke, I., Scheinert, D., Gontarska, K., and Thamsen, L.
\newblock {Let's Wait AWhile: How Temporal Workload Shifting Can Reduce Carbon
  Emissions in the Cloud}.
\newblock In \emph{Proceedings of the 22nd International Middleware
  Conference}, pp.\  260--272, New York, NY, USA, 2021. Association for
  Computing Machinery.
\newblock \doi{10.1145/3464298.3493399}.

\bibitem[Zhang et~al.(2021)Zhang, Jiang, Lu, and Yang]{Zhang:21}
Zhang, L., Jiang, W., Lu, S., and Yang, T.
\newblock {Revisiting Smoothed Online Learning}, 2021.
\newblock URL \url{https://arxiv.org/abs/2102.06933}.

\bibitem[Zhou et~al.(2008)Zhou, Chakrabarty, and Lukose]{Zhou:08}
Zhou, Y., Chakrabarty, D., and Lukose, R.
\newblock {Budget Constrained Bidding in Keyword Auctions and Online Knapsack
  Problems}.
\newblock In \emph{Lecture Notes in Computer Science}, pp.\  566--576. Springer
  Berlin Heidelberg, 2008.

\end{thebibliography}
\bibliographystyle{icml2024}

\newpage
\appendix
\onecolumn
\section*{\centering Appendix}

\section{Numerical Experiments (continued)} \label{apx:exp}
In this section, we give supplemental results examining the 95th percentile (``worst-case'') empirical competitive ratio results, following the same general structure as in the main body.

\subsection{Supplemental Results} \label{apx:95th}

To complement the results for the average empirical competitive ratio shown in \autoref{sec:exp}, in this section we plot the 95th percentile empirical competitive ratios for each tested algorithm, which primarily serve to show that the improved performance of our proposed algorithm holds in both average-case and tail (``worst-case'') scenarios.

In \autoref{fig:ul-95}-\ref{fig:var-95}, we investigate the impact of different parameters on the performance of each algorithm.
In \autoref{fig:ul-95}, we plot $95$th percentile empirical competitiveness for different values of $\nicefrac{U}{L}$ -- in this experiment, we fix $\beta = \nicefrac{U}{5}, d = 5, \xi = 0$, and $\sigma = \nicefrac{U}{5}$, while varying $\nicefrac{U}{L} \in \{ 50, \dots, 1250 \}$.  As observed in the average competitive ratio plot (\autoref{fig:ul}), the performance of \ALGone degrades as $\nicefrac{U}{L}$ grows, albeit at a favorable pace compared to the comparison algorithms.  \autoref{fig:beta-95} plots the $95$th percentile empirical competitiveness for different values of $\beta$ -- in this experiment, we fix $\nicefrac{U}{L} = 250, d = 5, \xi = 0$, and $\sigma = 50$.  As previously in the average competitive results (\autoref{fig:beta}), ``agnostic'' and ``move to minimizer'' heuristics perform better when $\beta$ grows, because the switching cost paid by the optimal solution grows as well.  %

\begin{wrapfigure}{r}{0.32\textwidth}
    \minipage{0.32\textwidth}
	\includegraphics[width=\linewidth]{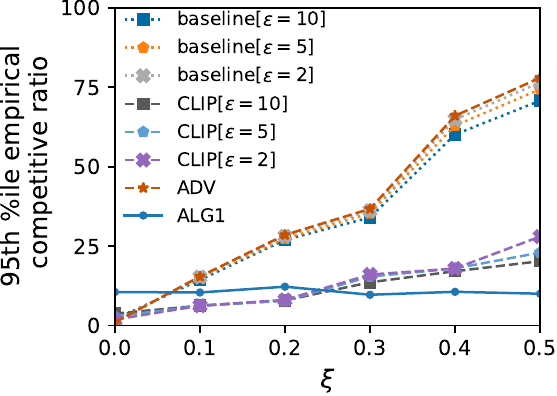}\vspace{-2em}
    \caption{Varying adversarial factor $\xi$, with $\nicefrac{U}{L}$ = $250, \beta$ = $50, d$ = $5$, $\sigma$ = $50$.}\label{fig:xi-95}
    \endminipage
    \vspace{1em}
\end{wrapfigure}

In \autoref{fig:dim-95}, we plot the $95$th percentile empirical competitiveness for different values of $d$ -- in this experiment, we fix $\nicefrac{U}{L} = 250, \beta = 50, \xi = 0$, and $\sigma = 50$, while varying $d$.  Mirroring the previous results (\autoref{fig:dim}), \ALGone and \CLIP's competitive performance degrades slower as $d$ grows compared to the comparison heuristics, as predicted by their \textit{dimension-free} theoretical bounds.  Finally, \autoref{fig:var-95} plots the $95$th percentile empirical competitiveness for different values of $\sigma$, which is the dimension-wise variability of each cost function. Here we fix $\nicefrac{U}{L} = 250, \beta = 50, d = 5,$ and $\xi = 0$, while varying $\sigma \in \{ 0, \dots, \nicefrac{U}{2}\}$.  Intuitively, as cost functions become more variable, the competitive ratios of all tested algorithms degrade, with the exception of our learning-augmented algorithm \CLIP.  This degradation plateaus as $\sigma$ grows, as a large standard deviation forces more of the terms of each cost vector $\mbf{c}_t$ to be truncated to the interval $[L,U]$.

In \autoref{fig:xi-95}, we plot the $95$th percentile empirical competitive ratio companion to \autoref{fig:xi}, which measures the effect of prediction error on the learning-augmented algorithms \CLIP and \BL.  We test several values of $\xi \in [0,1]$, the adversarial factor (recall that $\xi = 0$ implies the advice is correct), while fixing $\nicefrac{U}{L} = 250, \beta = 50, d = 5, \sigma = 50$.  We  test \BL and \CLIP for several values of $\epsilon \in \{2, 5, 10\}$ (note that $\ADV$ corresponds to running either \BL or \CLIP with $\epsilon = 0$).  Notably, in these $95$th percentile ``worst-case'' results, we find that \CLIP continues to significantly outperforms the \BL algorithm as $\xi$ grows, further validating that \CLIP is more empirically robust to prediction errors than the simple fixed ratio technique of \BL.

\begin{figure*}[h]
    \small
    \begin{center} 
    \vspace{1em}
    \minipage{\textwidth}
    \centering
    \includegraphics[width=0.66\linewidth]{figs/legend.png}
    \endminipage\hfill
    \end{center}
	\minipage{0.24\textwidth}
	\includegraphics[width=\linewidth]{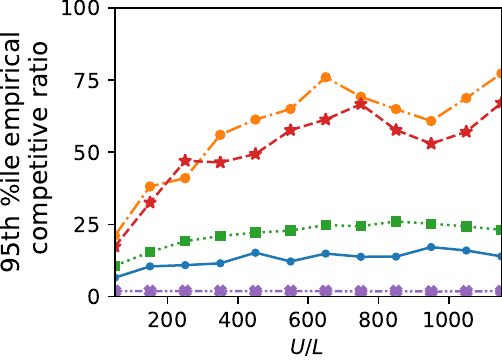}\vspace{-2em}
    \caption{Varying $\nicefrac{U}{L}$, with $\beta = \nicefrac{U}{5}, d = 5, \xi = 0$, and $\sigma = \nicefrac{U}{5}$.} \label{fig:ul-95}
	\endminipage\hfill
    \minipage{0.24\textwidth}
	\includegraphics[width=\linewidth]{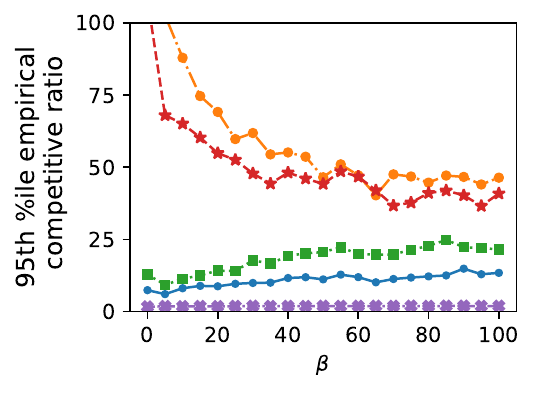}\vspace{-2em}
    \caption{Varying $\beta$, with $\nicefrac{U}{L} = 250, d = 5, \xi = 0$, and $\sigma = 50$.} \label{fig:beta-95}
	\endminipage \hfill
 	\minipage{0.24\textwidth}
	\includegraphics[width=\linewidth]{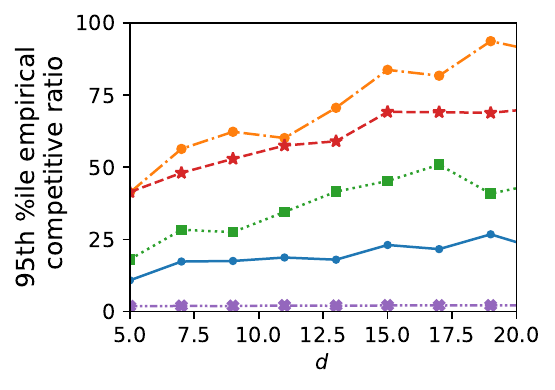}\vspace{-2em}
    \caption{Varying $d$ with $\beta$ = $50, \nicefrac{U}{L}$ = $250, \sigma$ = $50,$ and $\xi$=$0$.}\label{fig:dim-95}
	\endminipage\hfill
    \minipage{0.24\textwidth}
	\includegraphics[width=\linewidth]{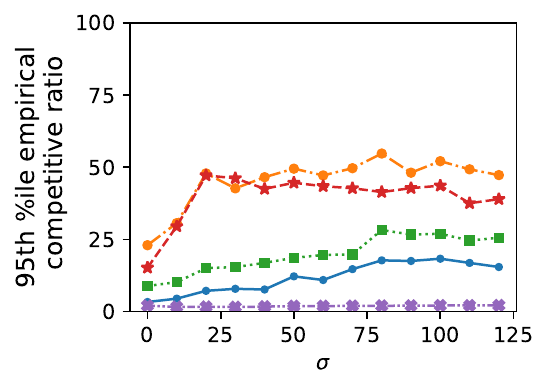}\vspace{-2em}
    \caption{Varying $\sigma$, with $\beta$ = $50, \nicefrac{U}{L}$ = $250, d$ = $5$, and $\xi$ = $0$.}\label{fig:var-95}
	\endminipage
\end{figure*}

\section{Proofs for \autoref{sec:comp} (Competitive Algorithms)} \label{apx:comp}
\begin{figure*}[h!]
    \small
    \minipage{0.66\textwidth}
    \centering
	\includegraphics[width=0.9\linewidth]{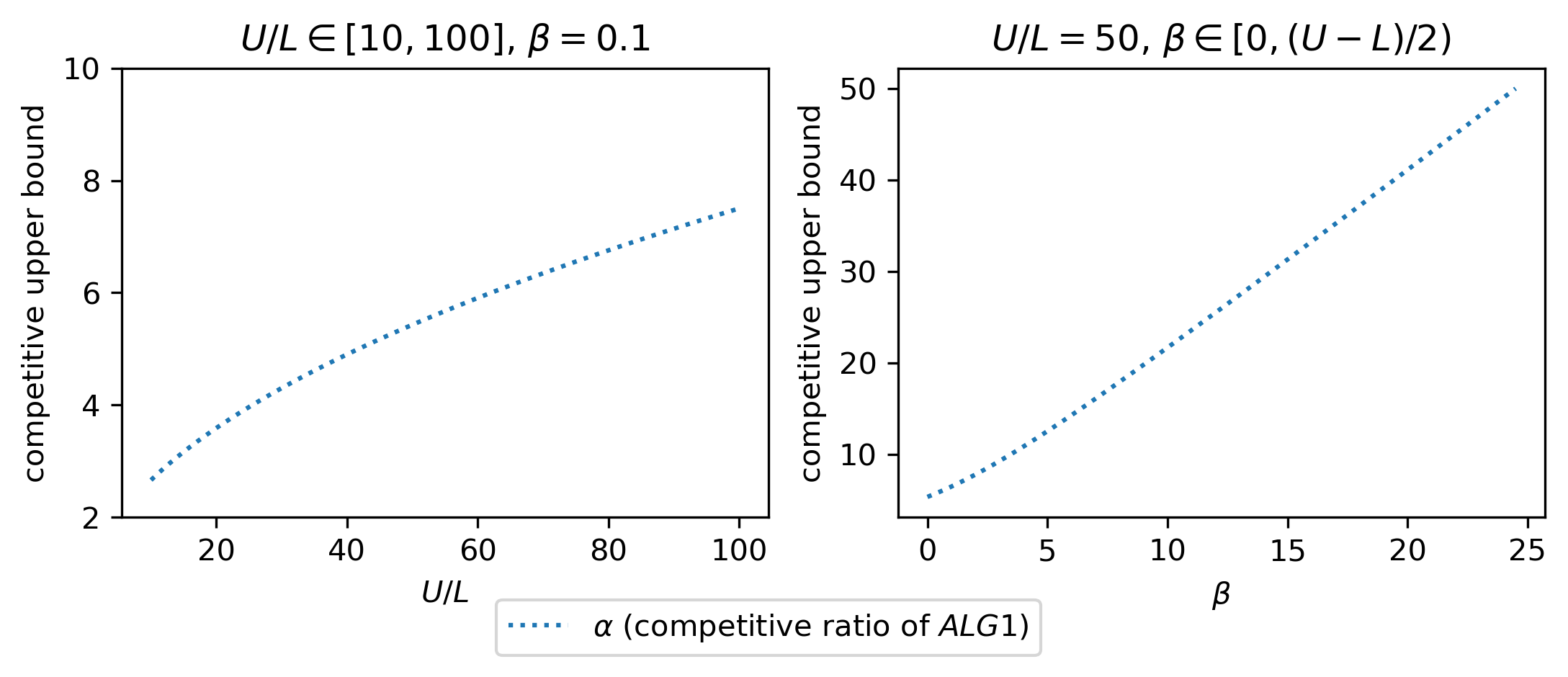}\vspace{-1em}
	\endminipage \hfill
    \minipage{0.33\textwidth}
	\includegraphics[width=\linewidth]{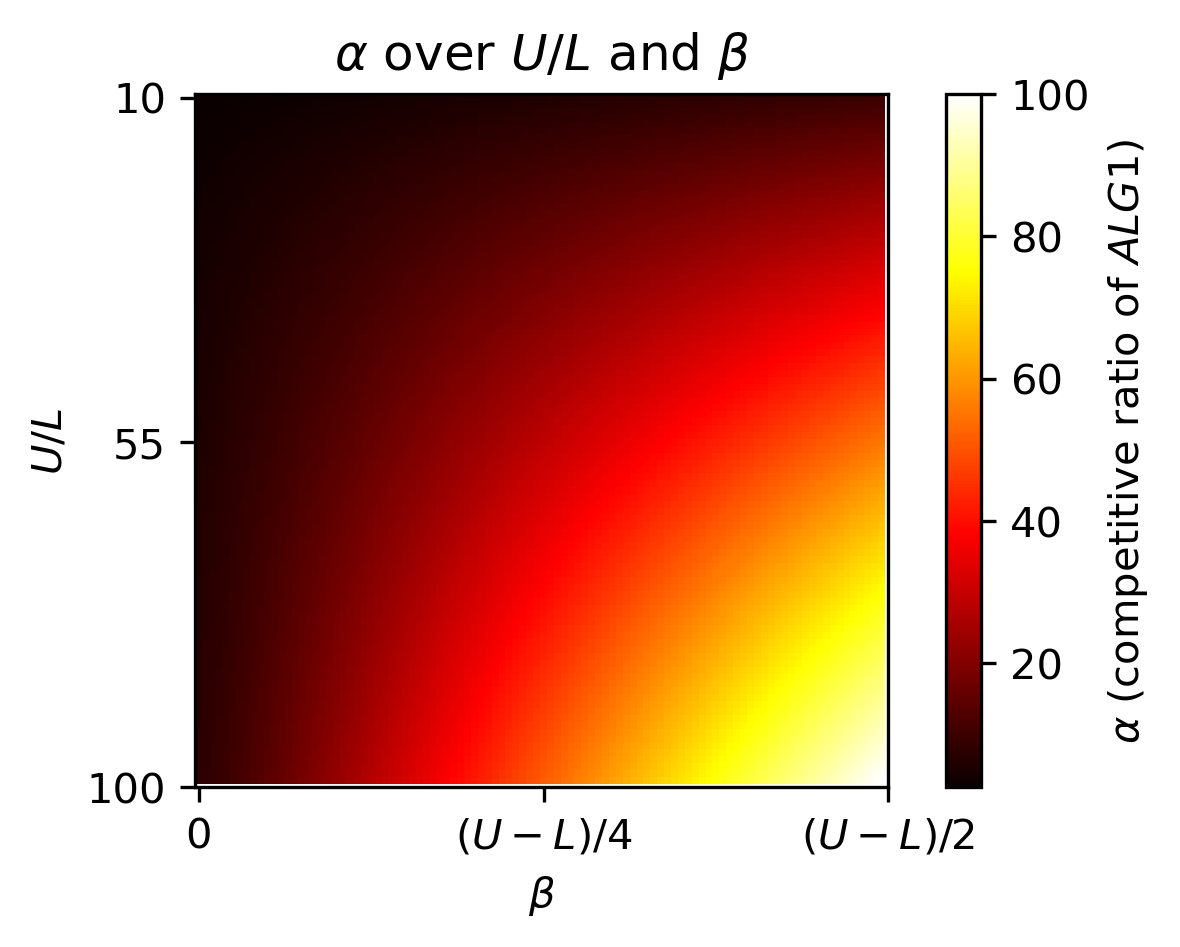}\vspace{-1em}
	\endminipage
     \caption{Plotting $\alpha$, the competitive upper bound of \ALGone (see \autoref{eq:alpha}), as a function of the \CFL problem's parameters (i.e., $\nicefrac{U}{L}$, $\beta$).}\label{fig:plotting-alpha}
    \vspace{1em}
\end{figure*}

\subsection{Convexity of the pseudo-cost minimization problem in \ALGone} \label{apx:pseudo-convex}

In this section, we show that the pseudo-cost minimization problem central to the design of \ALGone is a convex minimization problem, implying that it can be solved efficiently.

Define $h_t(\mbf{x}) : t \in [T]$ to represent the pseudo-cost minimization problem for a single arbitrary time step:
\begin{align}
    h_t(\cdot) = f_t(\mbf{x}) + d(\mbf{x}, \mbf{x}_{t-1}) - \int_{z^{(t-1)}}^{z^{(t-1)} + c( \mbf{x} ) }\phi(u) du.
\end{align}
 
\begin{theorem}
    Under the assumptions of the \CFL and \MAL problem settings, $h_t( \cdot )$ is always convex.
\end{theorem}
\begin{proof}

We prove the above statement by contradiction.

By definition, we know that the sum of two convex functions gives a convex function.  Since we have that $d(\mbf{x}, \mbf{x}')$ is defined as some norm, by definition and by observing that $\mbf{x}'$ is fixed, $d(\mbf{x}, \mbf{x}')$ is convex.  We have also assumed as part of the problem setting that each $f_t( \mbf{x} )$ is convex.  Thus, $f_t(\mbf{x}) + d(\mbf{x}, \mbf{x}')$ must be convex.

We turn our attention to the term $- \int_{z^{(t-1)}}^{z^{(t-1)} + c( \mbf{x} ) }\phi(u) du$.  Let $k(c(\mbf{x})) = \int_{z^{(t-1)}}^{z^{(t-1)} + c( \mbf{x} ) }\phi(u) du$.  By the fundamental theorem of calculus, $\nabla k(c(\mbf{x})) = \phi( z^{(t-1)} + c( \mbf{x} ) ) \nabla c(\mbf{x})$

Let $g(c(\mbf{x})) = \phi( z^{(t-1)} + c( \mbf{x} ) )$.  Then $\nabla^2 k(c(\mbf{x})) = \nabla^2 c(\mbf{x}) k(c(\mbf{x})) + \nabla c(\mbf{x}) g'(c(\mbf{x})) \nabla c(\mbf{x})^\intercal$.  Since $c(\mbf{x})$ is piecewise linear (\CFL and \MAL both assume it is linear), we know that $\nabla^2 c(\mbf{x}) g(c(\mbf{x})) = 0$.  Since $\phi$ is monotonically decreasing on the interval $[0,1]$, we know that $ g'(c(\mbf{x})) < 0$, and thus $\nabla c(\mbf{x}) g'(c(\mbf{x})) \nabla c(\mbf{x})^\intercal$ is negative semidefinite.  This implies that $k(c(\mbf{x}))$ is concave in $\mbf{x}$.

Since the negation of a concave function is convex, this causes a contradiction, because the sum of two convex functions gives a convex function.

Thus, $h_t(\cdot) = f_t(\mbf{x}) + d(\mbf{x}, \mbf{x}_{t-1}) - \int_{z^{(t-1)}}^{z^{(t-1)} + c( \mbf{x} ) }\phi(u) du$ is always convex under the assumptions of \CFL and \MAL.
\end{proof}

By showing that $h_t( \cdot ) $ is convex, it follows that the pseudo-cost minimization \eqref{eq:pseudocostRORO} in \ALGone is a convex minimization problem (i.e., it can be solved efficiently using numerical methods).

\subsection{Proof of \autoref{thm:alphaCompCFL}} \label{apx:alphaCompCFL}

In this section, we prove \autoref{thm:alphaCompCFL}, which shows that $\alpha$ as given by \eqref{eq:alpha} is an upper bound on the worst-case competitive ratio of \ALGone (given by Algorithm~\ref{alg:roro}) for the \CFL problem.

\begin{proof}[Proof of \autoref{thm:alphaCompCFL}]

Let $z^{(j)} = \sum_{t \in [T]} c( \mbf{x}_t )$ denote the fraction of the long-term constraint satisfied by \ALGone before the compulsory trade on an arbitrary \CFL instance $\mathcal{I} \in \Omega$. Also note that $z^{(t)} = \sum_{m\in[t]} c( \mbf{x}_m )$ is non-decreasing over $n$.

\begin{lemma}
\label{lem:opt-lb}
The offline optimal solution $\OPT(\mathcal{I})$ for any \CFL instance $\mathcal{I} \in \Omega$ is lower bounded by $\phi(z^{(j)}) - \beta$.
\end{lemma}
\paragraph{Proof of Lemma~\ref{lem:opt-lb}.} We prove this lemma by contradiction. Note that the offline optimum will stay at $\mbf{0}$ whenever possible, and satisfy the long-term constraint using the cost functions with the minimum gradient (i.e., the best marginal cost).  Assume that $\OPT(\mathcal{I}) < \phi(z^{(j)}) - \beta$, and that $z^{(j)} < 1$ (implying that $\OPT(\mathcal{I}) > L$).

Recall that any cost function $f_t(\cdot) : X \to \mathbb{R}$ is minimized exactly at $\mbf{0}$, since $f_t(\mbf{0}) = 0 \ \forall t \in [T]$.  By convexity of the cost functions, this implies that the gradient of some cost function $f_t$ is similarly minimized at the point $\mbf{0}$, and thus the \textit{best marginal cost} for $f_t$ can be obtained by taking an infinitesimally small step away from $\mbf{0}$ in at least one direction, which we denote (without loss of generality) as $i \in [d]$. For brevity, we denote this best marginal cost in $f_t$ by $\left[ \nabla f_t \right]^i$.

The assumption that $\OPT(\mathcal{I}) < \phi(z^{(j)}) - \beta$ implies that instance $\mathcal{I}$ must contain a cost function $f_m( \cdot )$ at some arbitrary time step $m$ ($m\in[T]$) which satisfies $\left[ \nabla f_m \right]^i < \phi(z^{(j)}) - \beta$ for any dimension $i \in [d]$.  

Prior work~\cite{Lorenz:08, SunZeynali:20} has shown that the worst-case for online search problems with long-term demand constraints occurs when cost functions arrive online in descending order, so we henceforth adopt this assumption.  Recall that at each time step, \ALGone solves the pseudo-cost minimization problem defined in \eqref{eq:pseudocostRORO}.
Without loss of generality, assume that $z^{(m-1)} = z^{(j)}$, i.e. the cost function $f_m( \cdot )$ arrives when $\ALGone$ has already reached its final utilization (before the compulsory trade).  This implies that $\mbf{x}_m = \mbf{0}$, and further that $c( \mbf{x}_m ) = 0$.  This implies that  
$
f_m (\mbf{x}) + \lVert \mbf{x} - \mbf{x}_{m-1} \rVert_{\ell_1 (\mbf{w})} > \int_{z^{(m-1)}}^{z^{(m-1)}+ c(\mbf{x})} \phi(u) du
$, since the pseudo-cost minimization problem should be minimized when \ALGone sets $\mbf{x}_m = \mbf{0}$.

The pseudo-cost minimization problem at time step $m$ can be expressed as follows:
\begin{align*}
    \mbf{x}_m = \argmin_{\mbf{x} \in \mathbb{R}^d : c(\mbf{x}) \leq 1 - z^{(m-1)}} f_m(\mbf{x}) + \lVert \mbf{x} - \mbf{x}_{m-1} \rVert_{\ell_1 (\mbf{w})} - \int_{z^{(m-1)}}^{z^{(m-1)}+ c(\mbf{x})} \phi(u) du.
\end{align*}

We note that $\lVert \mbf{x} - \mbf{x}_{m-1} \rVert_{\ell_1 (\mbf{w})}$ is upper bounded by $\beta( z^{(m-1)} + c(\mbf{x}) )$, since in the worst case, the previous online decision $\mbf{x}_{m-1}$ built up all of \ALGone's utilization ($z^{(m-1)}$) so far, and in the next step it will have to switch dimensions to ramp up to $\mbf{x}$.

Since the function $\phi$ is monotonically decreasing on $z \in [0,1]$, the $\mbf{x}_m$ solving the true pseudo-cost minimization problem is lower-bounded by the $\Breve{\mbf{x}}_m$ solving the following minimization problem (i.e., $c( \Breve{\mbf{x}}_m ) \leq c( \mbf{x}_m ) $):\\
\begin{align*}
    \Breve{\mbf{x}}_m = \argmin_{\mbf{x} \in \mathbb{R}^d : c( \mbf{x} ) \leq 1 - z^{(m-1)}} f_m(\mbf{x}) + \beta ( z^{(m-1)} + c(\mbf{x}) ) - \int_{z^{(m-1)}}^{z^{(m-1)}+ c(\mbf{x})} \phi(u) du.
\end{align*}

This further gives the following:
\begin{align*}
    & f_m(\mbf{x}) + \beta ( z^{(t)} + c( \mbf{x} ) ) - \int_{z^{(m-1)}}^{z^{(m-1)}+ c( \mbf{x} )} \phi(u) du\\
    & f_m(\mbf{x}) + \beta ( z^{(t)} + c( \mbf{x} ) ) - \int_{z^{(m-1)}}^{z^{(m-1)}+ c( \mbf{x} )} \left[ U - \beta + \left( \frac{U}{\alpha} - U + 2\beta \right) \exp(u/\alpha) \right] du\\
    & f_m(\mbf{x}) - (U - \beta) c( \mbf{x} ) + \beta ( z^{(t)} + c( \mbf{x} ) ) - \left[ U - U\alpha + 2\beta \alpha \right] \left( \exp \left( \frac{z^{(m-1)}+ c( \mbf{x} )}{\alpha} \right) - \exp \left( \frac{z^{(m-1)}}{\alpha} \right) \right)
\end{align*}

By assumption, since $f_m( \cdot)$ is convex and satisfies $\left[ \nabla f_m \right]^i < \phi(z^{(j)}) - \beta$ at $\mbf{x} = \mbf{0}$, there must exist a dimension $i$ in $f_m$ where an incremental step away from $\mbf{0}$ in direction $i$ satisfies the following inequality: $f_m(\mbf{x}) \lesssim \left[ \nabla f_m \right]^i \cdot c( \mbf{x} ) < [\phi(z^{(j)}) - \beta]  c( \mbf{x} )$ for some $\mbf{x}$ where $c(\mbf{x}) > 0$.  Thus, we have the following in the pseudo-cost minimization problem:
\begin{align*}
    & ( \left[ \nabla f_m \right]^i - U + \beta) c( \mbf{x} ) + \beta ( z^{(t)} + c( \mbf{x} ) ) - \left[ U - U\alpha + 2\beta \alpha \right] \left( \exp \left( \frac{z^{(m-1)}+ c( \mbf{x} )}{\alpha} \right) - \exp \left( \frac{z^{(m-1)}}{\alpha} \right) \right)
\end{align*}

Letting $c( x )$ be some scalar $y$ (which is valid since we assume there is at least one dimension in $f_t(\cdot)$ where the cost function growth rate is at most $\nabla f_m$), the pseudo-cost minimization problem finds the value $y$ which minimizes the following quantity:
\begin{align*}
    & (\left[ \nabla f_m \right]^i - U + \beta) y + \beta ( z^{(t)} + y ) - \left[ U - U\alpha + 2\beta \alpha \right] \left( \exp \left( \frac{z^{(m-1)}+ y}{\alpha} \right) - \exp \left( \frac{z^{(m-1)}}{\alpha} \right) \right)
\end{align*}

Taking the derivative of the above with respect to $y$ yields the following:
\begin{align*}
    & \frac{d}{dy} \left[ (\left[ \nabla f_m \right]^i - U + \beta) y + \beta ( z^{(t)} + y ) - \left[ U - U\alpha + 2\beta \alpha \right] \left( \exp \left( \frac{z^{(m-1)}+ y}{\alpha} \right) - \exp \left( \frac{z^{(m-1)}}{\alpha} \right) \right) \right] = \\
    & =  \left[ \nabla f_m \right]^i + 2 \beta - U + \frac{(U\alpha - 2\alpha \beta - U) \exp \left( \frac{z^{(m-1)}+ y}{\alpha} \right)}{\alpha}
\end{align*}

If $y = 0$, we have the following by assumption that $\left[ \nabla f_m \right]^i < \phi(z^{(j)}) - \beta$ and that $z^{(j)} = z^{(m-1)}$:
\begin{align*}
    \left[ \nabla f_m \right]^i + 2 \beta - U + (U - 2 \beta - U / \alpha) \exp \left( \frac{z^{(m-1)}}{\alpha} \right) &< \phi(z^{(j)}) + \beta - U + (U - 2 \beta - U / \alpha) \exp \left( \frac{z^{(m-1)}}{\alpha} \right)\\
    &< \phi(z^{(j)}) - \left( \phi(z^{(m-1)}) \right) = 0
\end{align*}

The above derivation implies that the derivative of the cost minimization problem at $c( \mbf{x} ) = 0$ (which corresponds to the case where $\mbf{x} = \mbf{0}$) is strictly less than $0$.  This further implies that $\Breve{\mbf{x}}_m$ must be \textit{non-zero}, since the minimizer must satisfy $c ( \Breve{\mbf{x}}_m ) > 0$.  Since $c( \Breve{\mbf{x}}_m )$ lower bounds the true $c( \mbf{x}_m )$, this causes a contradiction, as it was assumed that the utilization after time step $m$ would satisfy $z^{(m)} = z^{(m-1)} = z^{(j)}$, but if $c( \mbf{x}_m ) > 0$, $z^{(m)}$ must satisfy $z^{(m)} > z^{(m-1)}$.

It then follows by contradiction that $\OPT(\mathcal{I}) \geq \phi(z^{(j)}) - \beta$.

\begin{lemma}
\label{lem:alg-ub}
The cost of \ALGone on any valid \CFL instance $\mathcal{I} \in \Omega$ is upper bounded by
\begin{align}
    \emph{\ALGone}(\mathcal{I})\le \int_{0}^{z^{(j)}} \phi(u) du + \beta z^{(j)} + (1- z^{(j)}) U.
\end{align}
\end{lemma}
\paragraph{Proof of Lemma~\ref{lem:alg-ub}.}

First, recall that $z^{(t)} = \sum_{\tau \in[t]} c( \mbf{x}_\tau )$ is non-decreasing over $t \in [T]$.

Observe that whenever $c( \mbf{x}_t ) > 0$, we know that $f_t(\mbf{x}_t) + \lVert \mbf{x}_t - \mbf{x}_{t-1} \rVert_{\ell 1 (\mbf{w})} < \int_{z^{(t-1)}}^{z^{(t-1)}+ c( \mbf{x}_t )} \phi(u) du$.  Then, if $c( \mbf{x}_t ) = 0$, which corresponds to the case when $\mbf{x}_t = \mbf{0}$, we have the following:
\begin{align*}
    f_t(\mbf{x}_t) + \lVert \mbf{x}_t - \mbf{x}_{t-1} \rVert_{\ell 1 (\mbf{w})} - \int_{z^{(t-1)}}^{z^{(t-1)}+ c( \mbf{x}_t )}  \phi(u) du = 0 + \lVert - \mbf{x}_{t-1} \rVert_{\ell 1 (\mbf{w})} - 0 \leq \beta c( \mbf{x}_{t-1} )
\end{align*}

This gives that for any time step where $c( \mbf{x}_t ) = 0$, we have the following inequality:
\begin{align}
    & f_t(\mbf{x}_t) + \lVert \mbf{x}_t - \mbf{x}_{t-1} \rVert_{\ell 1 (\mbf{w})} \le \beta c( \mbf{x}_{t-1} ), \forall t \in [T] : c( \mbf{x}_t ) = 0.
\end{align}

\noindent And thus, since any time step where $c( \mbf{x}_t ) > 0$ implies $f_t(\mbf{x}_t) + \lVert \mbf{x}_t - \mbf{x}_{t-1} \rVert_{\ell 1 (\mbf{w})} < \int_{z^{(t-1)}}^{z^{(t-1)}+ c( \mbf{x}_t )} \phi(u) du$, we have the following inequality for all time steps (\textit{i.e., an upper bound on the excess cost not accounted for in the pseudo-cost threshold function or compulsory trade})
\begin{align}
    & f_t(\mbf{x}_t) + \lVert \mbf{x}_t - \mbf{x}_{t-1} \rVert_{\ell 1 (\mbf{w})} - \int_{z^{(t-1)}}^{z^{(t-1)}+ c( \mbf{x}_t )} \phi(u) du \leq \beta c( \mbf{x}_{t-1} ), \forall t \in [T].
\end{align}

Thus, we have
\begin{align}
   \beta z^{(j)} = \sum_{t \in [j]} \beta c( \mbf{x}_{t-1} ) &\ge \sum_{t \in [j]}  \left[ f_t(\mbf{x}_t) + \lVert \mbf{x}_t - \mbf{x}_{t-1} \rVert_{\ell 1 (\mbf{w})} - \int_{z^{(t-1)}}^{z^{(t-1)}+ c( \mbf{x}_t ) } \phi(u) du \right]\\
   &= \sum_{t \in [j]}  \left[ f_t(\mbf{x}_t) + \lVert \mbf{x}_t - \mbf{x}_{t-1} \rVert_{\ell 1 (\mbf{w})} \right] - \int_{0}^{z^{(j)}} \phi(u) du \\
    &= \ALGone - (1-z^{(j)})U - \int_{0}^{z^{(j)}}\phi(u) du.
\end{align}
 
Combining Lemma~\ref{lem:opt-lb} and Lemma~\ref{lem:alg-ub} gives
\begin{align}
    \frac{\ALGone(\mathcal{I})}{\OPT(\mathcal{I})} \le \frac{\int_{0}^{z^{(j)}} \phi(u) du + \beta z^{(j)} + (1- z^{(j)}) U}{\phi(z^{(j)}) - \beta} \leq \alpha,
\end{align}
where the last inequality holds since for any $z\in[0,1]$ 
\begin{align}
  \int_{0}^{z} \phi(u) du + \beta z + (1- z) U &=  \int_{0}^{z}\left[U - \beta + (U/\alpha - U + 2 \beta) \exp(z/\alpha)\right] + \beta z + (1 - z)U\\
  & = (U - \beta)z + \alpha (U/\alpha - U + 2 \beta) [\exp(z/\alpha) - 1] + \beta z + (1 - z)U\\
  & = \alpha (U/\alpha - U + 2 \beta) [\exp(z/\alpha) - 1] + U\\
  &= \alpha \left[U - 2 \beta + (U/\alpha - U + 2 \beta) \exp(z/\alpha) \right]\\
  &= \alpha [\phi(z) - \beta]. 
\end{align}
Thus, we conclude that \ALGone is $\alpha$-competitive for \CFL.
\end{proof}

\subsection{Proof of \sref{Corollary}{cor:alphaCompMAL}} \label{apx:alphaCompMAL}

In this section, we prove \sref{Corollary}{cor:alphaCompMAL}, which shows that the worst-case competitive ratio of \ALGone for \MAL is again upper bounded by $\alpha$ as defined in \eqref{eq:alpha}.

\begin{proof}[Proof of \sref{Corollary}{cor:alphaCompMAL}]
To show this result, we first prove a result stated in the main body, namely \sref{Lemma}{lem:simplexTransform}, which states the following:
For any \MAL instance on a weighted star metric $(X, d)$, there is a corresponding \CFL instance on $(\mathbb{R}^{n-1}, \lVert \cdot \rVert_{\ell_1 (\mbf{w'})})$ which preserves $f_t^a(\cdot) \ \forall t, c(\cdot) \ \forall a \in X$, and upper bounds $d(a,b) \ \forall (a,b) \in X$.

Before the proof, we note that \citet{Bansal:22} showed online metric allocation on a weighted star metric $(X, d)$ is identical to convex function chasing (with separable cost functions) on the normed vector space $(\Delta_n, \lVert \cdot \rVert_{\ell_1 (\mbf{w})})$, where $\Delta_n$ is the $n$-point simplex in $\mathbb{R}^{n}$ and $\lVert \cdot \rVert_{\ell_1 (\mbf{w})}$ is the weighted $\ell_1$ norm, with weights given by the corresponding edge weight in the underlying star metric as follows:
\begin{align*}
    \lVert \mbf{x} \rVert_{\ell_1 (\mbf{w})} &= \sum_{a \in X} \mbf{w}^a \lvert \mbf{x}^a \rvert.
\end{align*}

\begin{proof}[Proof of \sref{Lemma}{lem:simplexTransform}]
Recall that by assumption, the \MAL instance contains at least one \OFF point denoted by $a' \in X$ in the \MAL instance, where $\mbf{c}^{a'} = 0$.  Without loss of generality, let the \textit{first dimension} in $\Delta_n$ correspond to this \OFF point.

We define a linear map $\Phi : \Delta_n \to \mathbb{R}^{n-1}$, where
$\Phi$ has $n-1$ rows and $n$ columns, and is specified as follows:
\begin{align*}
    \Phi_{i,j} &= \begin{cases}
        1 \text{ if } $j = i+1$\\
        0 \text{ otherwise }
    \end{cases}
\end{align*}

It is straightforward to see that $\Phi \mbf{x} \in \mathbb{R}^{n-1},  \ \forall \mbf{x} \in \Delta_n$.

Recall that a \CFL decision space is the $\ell_1$ ball defined by the long-term constraint function in $\mathbb{R}^{n-1}$.  For any \MAL instance with constraint function $c(\mbf{x}) : \Delta_n \to \mathbb{R}$, we can define a long-term constraint function $c'(\mbf{x}') : \mathbb{R}^{n-1} \to \mathbb{R}$ as follows.  
The \MAL constraint function $c(\mbf{x})$ is defined as $\lVert \cdot \rVert_{\ell_1 (\mbf{c})}$ for some vector $\mbf{c} \in \mathbb{R}^{n-1}$.  Then 
\begin{align*}
    \mbf{c'} &= \Phi \mbf{c}\\
    c'(\mbf{x}') &= \lVert \mbf{x}' \rVert_{\ell_1 (\mbf{c'})} \ \forall \mbf{x}' \in \mathbb{R}^{n-1} 
\end{align*}

Furthermore, for any $z \in [0,1]$, let $\mbf{x} \in \Delta_n : c(\mbf{x}) < 1 - z$.  Then it follows that $\Phi \mbf{x}$ is in $\mathbb{R}^{n-1} : c'(\mbf{x}') < 1 - w$.

Recall that cost functions in the \MAL instance are convex and linearly separable as follows:
\begin{align*}
    f_t(\mbf{x}) = \sum_{a \in X} f_t^a(\mbf{x}^a)
\end{align*}
Next, again letting $\mbf{x} \in \Delta_n$, note that the $i$th term in $\mbf{x}$ is identical to the $(i-1)$th term in $\Phi \mbf{x}$ (excluding the first term in $\mbf{x}$).  Then we can construct cost functions in the \CFL instance as follows:
\begin{align*}
    f_t'(\mbf{x}') = \sum_{i \in [n-1]}  f_t^{i+1}(\mbf{x}^i)
\end{align*}

Under the mapping $\Phi$, note that it is straightforward to show that $f_t(\mbf{x}) = f_t'(\Phi \mbf{x})$ for any $\mbf{x} \in \Delta_n$.

Finally, consider the distances in the \MAL instance's weighted star metric, which can be expressed as a weighted $\ell_1$ norm defined by $\mbf{w}$, where the terms of $\mbf{w}$ correspond to the weighted edges of the star metric.  Recall that $\beta \coloneqq \max_{a', a} \lVert a' - a \rVert_{\ell_1 (\mbf{w})}$, i.e., the maximum distance between the \OFF point and any other point in the weighted star.

Then we define a corresponding distance metric in the \CFL instance, which is an $\ell_1$ norm weighted by $\mbf{w'} \in \mathbb{R}^{n-1}$, which is defined as follows:
\begin{align*}
    \mbf{w'}^i &= \mbf{w}^{i+1} + \mbf{w}^0.
\end{align*}
Note that $\mbf{w}^0$ is the edge weight associated with the \OFF point.  Then for any $(\mbf{x},\mbf{y}) \in \Delta_n$, it is straightforward to show the following:
\begin{align*}
    \lVert \mbf{x} - \mbf{y} \rVert_{\ell_1 (\mbf{w})} \le \lVert \Phi \mbf{x} - \Phi \mbf{y} \rVert_{\ell_1 (\mbf{w'})}
\end{align*}
This follows since for any $(\mbf{x},\mbf{y}) \in \Delta_n$ where $\mbf{x}^0 = 0$ and $\mbf{y}^0 = 0$ (i.e., allocations which do not allocate anything to the \OFF point), $\lVert \Phi \mbf{x} - \Phi \mbf{y} \rVert_{\ell_1 (\mbf{w'})} = \lVert \mbf{x} - \mbf{y} \rVert_{\ell_1 (\mbf{w})} + \lVert \mbf{x} - \mbf{y} \rVert_{\ell_1} \cdot \mbf{w}^0$.  

Conversely, if either $\mbf{x}$ or $\mbf{y}$ have $\mbf{x}^0 > 0$ or $\mbf{y}^0 > 0$, we have  $\lVert \mbf{x} - \mbf{y} \rVert_{\ell_1 (\mbf{w})} \le \lVert \Phi \mbf{x} - \Phi \mbf{y} \rVert_{\ell_1 (\mbf{w'})} \le \lVert \mbf{x} - \mbf{y} \rVert_{\ell_1 (\mbf{w})} + \lVert \mbf{x} - \mbf{y} \rVert_{\ell_1} \cdot \mbf{w}^0$.  Finally, supposing that (without loss of generality) $\mbf{x}$ has $\mbf{x}^0 = 1$, we have that $\lVert \mbf{x} -\mbf{y} \rVert_{\ell_1 (\mbf{w})} = \lVert \Phi \mbf{x} - \Phi \mbf{y} \rVert_{\ell_1 (\mbf{w'})}$.\\  Thus, $\lVert \Phi \mbf{x} - \Phi \mbf{y} \rVert_{\ell_1 (\mbf{w'})}$ upper bounds $\lVert \mbf{x} - \mbf{y} \rVert_{\ell_1 (\mbf{w})}$.  Furthermore, the constructed distance metric preserves $\beta$, i.e. given $(a', a) = \argmax_{a', a} \lVert a' - a \rVert_{\ell_1 (\mbf{w})}$, we have that $\lVert \Phi a' - \Phi a \rVert_{\ell_1 (\mbf{w'})} = \beta$.

Next, we show that the transformation $\Phi$ is bijective.  We define the affine map $\Phi^{-1} : \mathbb{R}^{n-1} \to \Delta_n$ as follows:  $\Phi^{-1}$ has $n$ rows and $n-1$ columns, where the first row is all $-1$, and the bottom $n$ rows are the $n \times n$ identity matrix.  Let $\mbf{b} \in \mathbb{R}^{n-1}$ denote the vector with $\mbf{b}^0 = 1$ and all other terms are zero, i.e., $\mbf{b}_i = 0 \ \forall i \geq 1$.

For any $\mbf{x}' \in \mathbb{R}^{n-1} : c'(\mbf{x}') \leq 1$, it is straightforward to show that $\Phi^{-1}\mbf{x}' + \mbf{b}$ is in $\Delta_n$, since by definition we have that $\sum_{i \in [n+1]} \left( \Phi^{-1}\mbf{x}' + \mbf{b} \right)_i = 1$.  Furthermore, by definition of $c'(\mbf{x}')$, we have that $c(\Phi^{-1}\mbf{x}' + \mbf{b}) = c'(\mbf{x}')$, because the $i$th term (excluding the first term) of $\Phi^{-1}\mbf{x}' + \mbf{b}$ is identical to the $(i-1)$th term of $\mbf{x}'$.  Similarly, by definition of $f_t'$, we have that  $f_t(\Phi^{-1}\mbf{x}' + \mbf{b}) = f_t'(\mbf{x}')$.

Finally, considering the distance metric, we have that for any $(\mbf{x}', \mbf{y}') \in \mathbb{R}^{n-1} : c'(\mbf{x}') \leq 1$:
\begin{align*}
    \lVert (\Phi^{-1}\mbf{x}' + \mbf{b}) - (\Phi^{-1}\mbf{y}' + \mbf{b}) \rVert_{\ell_1 (\mbf{w})} \le \lVert \mbf{x}' - \mbf{y}' \rVert_{\ell_1 (\mbf{w'})}.
\end{align*}
This follows by considering that for any $\mbf{x}'$, $\Phi^{-1}\mbf{x}' + \mbf{b}$ adds a dimension (corresponding to the \OFF point) and sets $\left(\Phi^{-1}\mbf{x}' + \mbf{b}\right) = 1 - \lVert \mbf{x}' \rVert_1$.  Then the distance between any two points which allocate a non-negative fraction to the \OFF point in $\Delta_n$ is $\leq$ the distance in $\mathbb{R}^{n-1}$ by definition of the weight vector $\mbf{w}'$, and the distance between e.g., the allocation fully in the \OFF point ($a'$) and any other allocation is exactly preserved.

Furthermore, note that if $\mbf{w}^0 = 0$ (i.e., the weight of the \OFF state in the weighted star metric is 0), $\Phi$ is a bijective isometry between $(\Delta_n, \lVert \cdot \rVert_{\ell_1 (\mbf{w})})$ and $(\mathbb{R}^{n-1}, \lVert \cdot \rVert_{\ell_1 (\mbf{w'})})$.
\end{proof}

The transformation defined by $\Phi$ in \sref{Lemma}{lem:simplexTransform} allows us to put decisions on the \CFL instance $(\mathbb{R}^{n-1}, \lVert \cdot \rVert_{\ell_1 (\mbf{w'})})$ in one-to-one correspondence with decisions in $(\Delta_n, \lVert \cdot \rVert_{\ell_1 (\mbf{w})})$.  

Below, we formalize this by proving a result stated in the main body (\sref{Proposition}{prop:CFL-MAL}) which states the following: Given an algorithm \ALG for \CFL, any performance bound on \ALG which assumes \OPT does not pay any switching cost will translate to an identical performance bound for \MAL whose parameters depend on the corresponding \CFL instance constructed according to \sref{Lemma}{lem:simplexTransform}.

\begin{proof}[Proof of \sref{Proposition}{prop:CFL-MAL}]
The cost of \ALG on the \CFL instance is an upper bound on the cost of the \ALG 's decisions mapped into the \MAL instance.  This follows since the cost functions are preserved exactly between the two instances, the long-term constraint function is preserved exactly, and the \CFL switching cost is by definition an upper bound on the \MAL switching cost.

If the \CFL performance bound assumes that \OPT does not pay any switching cost (e.g., as in \autoref{thm:alphaCompCFL}), lower bounding the cost of \OPT on the \CFL instance is equivalent to lower bounding the cost of \OPT on the \MAL instance, as the cost functions and constraint functions are preserved exactly.

Thus, we have that any such performance bound for \ALG on the \CFL instance constructed appropriately (as in \sref{Lemma}{lem:simplexTransform}) immediately gives an identical performance bound for the \MAL instance, yielding the result.
\end{proof}

By \sref{Lemma}{lem:simplexTransform}, we have that since \ALGone is $\alpha$-competitive for \CFL (\autoref{thm:alphaCompCFL}), \ALGone is $\alpha$-competitive for any \CFL instance constructed based on a \MAL instance.  Furthermore, by \sref{Proposition}{prop:CFL-MAL}, \ALGone is also $\alpha$-competitive on the underlying \MAL instance, where $\alpha$ is given by \eqref{eq:alpha}.
\end{proof}

\subsection{Proof of \autoref{thm:lowerboundCFL}} \label{apx:lowerboundCFL}

In this section, we prove \autoref{thm:lowerboundCFL}, which shows that $\alpha$ as given by \eqref{eq:alpha} is the best competitive ratio achievable for \CFL.

To show this lower bound, we first define a family of special adversaries, and then show that the competitive ratio for any deterministic algorithm is lower bounded under the instances provided by these adversaries.

Prior work has shown that difficult instances for online search problems with a minimization objective occur when inputs arrive at the algorithm in an decreasing order of cost~\cite{ElYaniv:01, Lorenz:08, SunZeynali:20, Lechowicz:23}.  For \CFL, we additionally consider how an adaptive adversary can essentially force an algorithm to incur a large switching cost in the worst-case.
We now formalize such a family of adversaries $\{ \mathcal{A}_y \}_{y \in [L, U]}$, where $\mathcal{A}_y$ is called a \textit{$y$-adversary}.

\begin{definition}[$y$-adversary for \CFL] \label{dfn:yadversary}

Let $w, m \in \mathbb{Z}$ be sufficiently large, and $\delta := \nicefrac{(U-L)}{w}$. 

Without loss of generality, let $k = \argmax_{i \in [d]} \mbf{w}_i$, where $\mbf{w}$ is the weight vector for $\lVert \cdot \rVert_{\ell_1 (\mbf{w})}$, and let $\beta = \max_{i\in [d]} \mbf{w}_i$.
For $y \in [L, U]$, an adaptive adversary $\mathcal{A}_y$ sequentially presents two types of cost functions $f_t(\cdot)$ to both $\ALG$ and $\OPT$.  

These types of cost functions are $\textbf{Up}(\mbf{x}) = U \mbf{1} \mbf{x}^\intercal$, and $\textbf{Down}^i (\mbf{x}) = \sum_{j \not = k}^d U \mbf{x}^j + (U - i \delta) \mbf{x}^k$.

The adversary sequentially presents cost functions from these two types in an alternating, ``continuously decreasing'' order.  Specifically, they start by presenting cost function $\textbf{Up}(\mbf{x})$, up to $m$ times.  

Then, they present $\textbf{Down}^1 (\mbf{x})$, which has linear cost coefficient $U$ in every direction except direction $k$, which has cost coefficient $(U - 1 \cdot \delta)$. $\textbf{Down}^1 (\mbf{x})$ is presented up to $m$ times.  If $\ALG$ ever ``accepts'' a cost function $\textbf{Down}^1 (\mbf{x})$ (i.e., if $\ALG$ makes a decision $\mbf{x}$ where $c(\mbf{x}) > 0$), the adaptive adversary immediately presents $\textbf{Up}(\mbf{x})$ starting in the next time step until either $\ALG$ moves to the origin (i.e. online decision $\mbf{x} = \mbf{0}$) or $\ALG$'s utilization $z = 1$.

The adversary continues alternating in this manner, presenting $\textbf{Down}^2 (\mbf{x})$ up to $m$ times, followed by $\textbf{Up}(\mbf{x})$ if $\ALG$ accepts anything, followed by $\textbf{Down}^3 (\mbf{x})$ up to $m$ times, and so on.  This continues until the adversary presents $\textbf{Down}^{w_y} (\mbf{x})$, where $y \coloneqq (U - w_y \delta)$, up to $m$ times.  After presenting $\textbf{Down}^{w_y} (\mbf{x})$, $\mathcal{A}_y$ will present $\textbf{Up}(\mbf{x})$ until either $\ALG$ moves to the origin or has utilization $z=1$.  Finally, the adversary presents exactly $m$ cost functions of the form $\sum_{j \not = k}^d U \mbf{x}^j + (y + \varepsilon) \mbf{x}^k$, followed by $m$ cost functions $\textbf{Up}(\mbf{x})$.

The mechanism of this adaptive adversary is designed to present ``good cost functions'' (i.e., $\textbf{Down}^i (\mbf{x})$) in a worst-case decreasing order, interrupted by blocks of ``bad cost functions'' $\textbf{Up}(\mbf{x})$ which force a large switching cost in the worst case.

\end{definition}

\newcommand{\cItem}[1]{\boxed{ \{ g(y) = #1y \} \times m }}

\newcommand{\aItem}[1]{ g(y) = #1y }

$\mathcal{A}_{U}$ is simply a stream of $m$ cost functions $U$, and the final cost functions in any $y$-adversary instance are always $\textbf{Up}(\mbf{x})$. %

\begin{proof}[Proof of \autoref{thm:lowerboundCFL}]

Let $g(y)$ denote a \textit{conversion function} $[L,U] \rightarrow [0,1]$, which fully describes the progress towards the long-term constraint (before the compulsory trade) of a deterministic $\ALG$ playing against adaptive adversary $\mathcal{A}_y$.  Note that for large $w$, the adaptive adversary $\mathcal{A}_{y-\delta}$ is equivalent to first playing $\mathcal{A}_y$ (besides the last two batches of cost functions), and then processing batches with cost functions $\textbf{Down}^{w_y + 1} (\mbf{x})$ and $\textbf{Up}(\mbf{x})$.  Since $\ALG$ is deterministic and the conversion is unidirectional (irrevocable), we must have that $g(y - \delta) \geq g(y)$, i.e. $g(y)$ is non-increasing in $[L, U]$.  Intuitively, the entire capacity should be satisfied if the minimum possible price is observed, i.e $g(L) = 1$.

Note that for $\varepsilon \to 0$, the optimal solution for adversary $\mathcal{A}_y$ is $\OPT(\mathcal{A}_y) = y + 2\beta/m$, and for $m$ sufficiently large, $\OPT(\mathcal{A}_y) \to y$.

Due to the adaptive nature of each $y$-adversary, any deterministic $\ALG$ incurs a switching cost proportional to $g(y)$, which gives the amount of utilization obtained by $\ALG$ before the end of $\mathcal{A}_y$'s sequence.  

Whenever $\ALG$ accepts some cost function with coefficient $U-i \delta$ in direction $k$, the adversary presents $\textbf{Up}(\mbf{x})$ starting in the next time step.  Any $\ALG$ which does not switch away immediately obtains a competitive ratio strictly worse than an algorithm which does switch away (if an algorithm accepts $c$ fraction of a good price and switches away immediately, the switching cost it will pay is $2\beta c$.  An algorithm may continue accepting $c$ fraction of coefficient $U$ in the subsequent time steps, but a sequence exists where this decision will take up too much utilization to recover when better cost functions are presented later.  In the extreme case, if an algorithm continues accepting $c$ fraction of these $U$ coefficients, it might fill its utilization and then $\OPT$ can accept a cost function which is arbitrarily better).

Since accepting any price by a factor of $c$ incurs a switching cost of $2 \beta c$, the switching cost paid by $\ALG$ on adversary $\mathcal{A}_y$ is $2\beta g(y)$.  We assume that $\ALG$ is notified of the compulsory trade, and does not incur a significant switching cost during the final batch.

Then the total cost incurred by an $\alpha^\star$-competitive online algorithm $\ALG$ on adversary $\mathcal{A}_y$ is $\ALG(\mathcal{A}_y) = g(\nicefrac{U}{\alpha^\star}) \nicefrac{U}{\alpha^\star} - \int^y_{\nicefrac{U}{\alpha^\star}} u d g(u) + 2\beta g(y) + (1 - g(y))U$, where $udg(u)$ is the cost of buying $dg(u)$ utilization at cost coefficient $u$, the last term is from the compulsory trade, and the second to last term is the switching cost incurred by $\ALG$.  Note that any deterministic $\ALG$ which makes conversions when the price is larger than $\nicefrac{U}{\alpha^\star}$ can be strictly improved by restricting conversions to prices~$\leq \nicefrac{U}{\alpha^\star}$.

For any $\alpha^\star$-competitive online algorithm, the corresponding conversion function $g(\cdot)$ must satisfy $\ALG(\mathcal{A}_y) \leq \alpha^\star \OPT(\mathcal{A}_y) = \alpha^\star y, \forall y \in [L, U]$.  This gives a necessary condition which the conversion function must satisfy as follows:
\[
\ALG(\mathcal{A}_y) = g(\nicefrac{U}{\alpha^\star}) \nicefrac{U}{\alpha^\star} - \int^y_{\nicefrac{U}{\alpha^\star}} u d g(u) + 2\beta g(y) + (1 - g(y))U \leq \alpha^\star y , \quad \forall y \in [L, U].
\]
By integral by parts, the above implies that the conversion function must satisfy $g(y) \geq \frac{U - \alpha^\star y}{U - y - 2\beta} - \frac{1}{U - y - 2\beta} \int_{\nicefrac{U}{\alpha^\star}}^y g(u) du$.  By Gr\"{o}nwall's Inequality \citep[Theorem 1, p. 356]{Mitrinovic:91}, we have that
\begin{align*}
g(y) & \geq \frac{U - \alpha^\star y}{U - y - 2\beta} - \frac{1}{U - y - 2\beta} \int_{\nicefrac{U}{\alpha^\star}}^y \frac{U - \alpha^\star u}{U - u - 2\beta} \cdot \exp\left( \int_u^y \frac{1}{U - r - 2\beta} dr \right) du \\
& \geq \frac{U - \alpha^\star y}{U - y - 2\beta} - \int_{\nicefrac{U}{\alpha^\star}}^y \frac{U - \alpha^\star u}{(U - u - 2\beta)^2} du \\
& \geq \frac{U - \alpha^\star y}{U - y - 2\beta} - \left[ \frac{U\alpha^\star - U - 2\beta \alpha^\star}{u + 2\beta - U} - \alpha^\star \ln \left( u + 2\beta -U \right) \right]_{\nicefrac{U}{\alpha^\star}}^y \\
& \geq \alpha^\star \ln \left( y + 2\beta -U \right) - \alpha^\star \ln \left( \nicefrac{U}{\alpha^\star} + 2\beta -U \right), \quad \forall y \in [L, U].
\end{align*}

$g(L) = 1$ by the problem definition -- we can combine this with the above constraint to give the following condition for an $\alpha^\star$-competitive online algorithm:
\[
\alpha^\star \ln \left( L + 2\beta -U \right) - \alpha^\star \ln \left( \nicefrac{U}{\alpha^\star} + 2\beta -U \right) \leq g(L) = 1.
\]
The optimal $\alpha^\star$ is obtained when the above inequality is binding, so solving for the value of $\alpha^\star$ which solves $\alpha^\star \ln \left( L + 2\beta -U \right) - \alpha^\star \ln \left( \nicefrac{U}{\alpha^\star} + 2\beta -U \right) = 1$ yields that the best competitive ratio for any $\ALG$ solving \CFL is $\alpha^\star \geq \left[ W \left( \frac{e^{\nicefrac{2\beta}{U}} ( \nicefrac{L}{U} + \nicefrac{2 \beta}{U} - 1) }{e} \right) - \frac{2\beta}{U} + 1 \right]^{-1}$.
\end{proof}

\subsection{Proof of \sref{Corollary}{cor:lowerboundMAL}} \label{apx:lowerboundMAL}

In this section, we prove \sref{Corollary}{cor:lowerboundMAL}, which shows that $\alpha$ as given by \eqref{eq:alpha} is the best competitive ratio achievable for \MAL.

To show this lower bound, we build off of the family of adversaries in \sref{Definition}{dfn:yadversary}, which are designed to force an algorithm to incur a large switching cost while satisfying the long-term constraint.  In \sref{Definition}{dfn:yadversaryMAL} we define this family of adversarial instances tailored for \MAL.

\begin{definition}[$y$-adversary for \MAL] \label{dfn:yadversaryMAL}
Let $w, m \in \mathbb{Z}$ be sufficiently large, and $\delta := \nicefrac{(U-L)}{w}$. 

Recall that $\mbf{w}$ denotes the vector of edge weights for each point in the weighted star metric $X$, and the \OFF point is defined (without loss of generality) as the point $a' \in X$ where $\mbf{c}^{a'} = 0$ and $f_t^{a'}(\mbf{x}^a) = 0 \forall t \in [T], \forall \mbf{x}^a \in [0,1]$.  We will assume that $\mbf{c}^{a} = 1 \ \forall a \in X : a \not = a'$.

Then we set $\mbf{w}^{a'} = 0$, i.e., the \OFF point is connected to the interior vertex of the weighted star with an edge of weight $0$.  Without loss of generality, we let $k = \argmax_{a \in [n]} \mbf{w}^a$ denote the largest edge weight of any other (non-\OFF) point in the metric. By definition, recall that $\beta = \mbf{w}^k$.

For $y \in [L, U]$, an adaptive adversary $\mathcal{A}_y$ sequentially presents two different sets of cost functions $f_t^a(\cdot)$ at each point in the metric space. 

These sets of cost functions are $\textbf{Up} = \{ f^a(x) = U \mbf{x}^a \ \ \forall a \in X \setminus \{ a' \} \}$, and $\textbf{Down}^i = \{ f^k(\mbf{x}^k) = (U - i \delta) \mbf{x}^k \} \cap \{ f^a(\mbf{x}^a) = U \mbf{x}^a \ \ \forall a \in X \setminus \{ a', k \} \}$.  Note that the adversary only ever presents cost functions with a coefficient $< U$ at the point $k$ which corresponds to the largest edge weight.

The adversary sequentially presents either of these two sets of cost functions in an alternating, ``continuously decreasing'' order.  Specifically, they start by presenting $\textbf{Up}$, up to $m$ times.  

Then, they present $\textbf{Down}$, which has cost coefficient $U$ in every point except point $k$, which has cost coefficient $(U - 1 \cdot \delta)$. $\textbf{Down}^1$ is presented up to $m$ times.  If $\ALG$ ever ``accepts'' a cost function in $\textbf{Down}^1$ (i.e., if $\ALG$ makes a decision $x$ where $c(x) > 0$), the adaptive adversary immediately presents $\textbf{Up}$ starting in the next time step until either $\ALG$ moves entirely to the \OFF point (i.e. online decision $x^{a'} = 1$) or $\ALG$'s utilization $z = 1$.

The adversary continues alternating in this manner, presenting $\textbf{Down}^2$ up to $m$ times, followed by $\textbf{Up}$ if $\ALG$ accepts anything, followed by $\textbf{Down}^3$ up to $m$ times, and so on.  This continues until the adversary presents $\textbf{Down}^{w_y}$, where $y = (U - w_y \delta)$, up to $m$ times.  After presenting $\textbf{Down}^{w_y}$, $\mathcal{A}_y$ will present $\textbf{Up}(x)$ until either $\ALG$ moves to the \OFF point or has utilization $z=1$.  Finally, the adversary presents the set of cost functions $\{ f^k(\mbf{x}^k) = (y + \varepsilon) \mbf{x}^k \} \cap \{ f^a(\mbf{x}^a) = U \mbf{x}^a \ \ \forall a \in X \setminus \{ a', k \} \}$ $m$ times, followed by $\textbf{Up}$ $m$ times.

The mechanism of this adaptive adversary is designed to present ``good cost functions'' (i.e., $\textbf{Down}^i$) in a worst-case decreasing order, interrupted by blocks of ``bad cost functions'' $\textbf{Up}$ which force a large switching cost in the worst case.

\end{definition}

As in \autoref{thm:lowerboundCFL}, $\mathcal{A}_{U}$ is simply a stream of $m$ $\textbf{Up}$ sets of cost functions, and the final cost functions in any $y$-adversary instance are always $\textbf{Up}$. 

\begin{proof}[Proof of \sref{Corollary}{cor:lowerboundMAL}]
As previously, we let $g(y)$ denote a \textit{conversion function} $[L,U] \rightarrow [0,1]$, which fully describes the progress towards the long-term constraint (before the compulsory trade) of a deterministic $\ALG$ playing against adaptive adversary $\mathcal{A}_y$.  %
Since $\ALG$ is deterministic and the conversion is unidirectional (irrevocable), %
$g(y)$ is non-increasing in $[L, U]$.  Intuitively, the entire long-term constraint should be satisfied if the minimum possible price is observed, i.e $g(L) = 1$.
For $\varepsilon \to 0$, the optimal solution for adversary $\mathcal{A}_y$ is $\OPT(\mathcal{A}_y) = y + 2\beta/m$, and for $m$ sufficiently large, $\OPT(\mathcal{A}_y) \to y$.

As in \autoref{thm:lowerboundCFL}, the adaptive nature of each $y$-adversary forces any deterministic $\ALG$ to incur a switching cost of $2\beta g(y)$ on adversary $\mathcal{A}_y$, and we assume that $\ALG$ does not incur a significant switching cost during the final batch (i.e., during the compulsory trade).

Then the total cost incurred by an $\alpha^\star$-competitive online algorithm $\ALG$ on adversary $\mathcal{A}_y$ is $\ALG(\mathcal{A}_y) = g(\nicefrac{U}{\alpha^\star}) \nicefrac{U}{\alpha^\star} - \int^y_{\nicefrac{U}{\alpha^\star}} u d g(u) + 2\beta g(y) + (1 - g(y))U$, where $udg(u)$ is the cost of buying $dg(u)$ utilization at cost coefficient $u$, the last term is from the compulsory trade, and the second to last term is the switching cost incurred by $\ALG$.  Note that this expression for the cost is exactly as defined in \autoref{thm:lowerboundCFL}.

Thus by \autoref{thm:lowerboundCFL}, for any $\alpha^\star$-competitive online algorithm, the conversion function $g(\cdot)$ must satisfy $\ALG(\mathcal{A}_y) \leq \alpha^\star \OPT(\mathcal{A}_y) = \alpha^\star y, \forall y \in [L, U]$.  Via integral by parts and Gr\"{o}nwall's Inequality \citep[Theorem 1, p. 356]{Mitrinovic:91}, we have the following condition on $g(y)$: 
\begin{align*}
g(y) & \geq \alpha^\star \ln \left( y + 2\beta -U \right) - \alpha^\star \ln \left( \nicefrac{U}{\alpha^\star} + 2\beta -U \right), \quad \forall y \in [L, U].
\end{align*}

$g(L) = 1$ by the problem definition -- combining this with the previous condition gives the following condition for an $\alpha^\star$-competitive online algorithm:
\[
\alpha^\star \ln \left( L + 2\beta -U \right) - \alpha^\star \ln \left( \nicefrac{U}{\alpha^\star} + 2\beta -U \right) \leq g(L) = 1.
\]
As in \autoref{thm:lowerboundCFL}, the optimal $\alpha^\star$ is obtained when the above inequality is binding, yielding that the best competitive ratio for any $\ALG$ solving \MAL is $\alpha^\star \geq \left[ W \left( \frac{e^{\nicefrac{2\beta}{U}} ( \nicefrac{L}{U} + \nicefrac{2 \beta}{U} - 1) }{e} \right) - \frac{2\beta}{U} + 1 \right]^{-1}$.
\end{proof}

\section{Proofs for \autoref{sec:clip} (Learning-Augmentation)} \label{apx:clip}
\subsection{Proof of \sref{Lemma}{lem:baseline}} \label{apx:baseline}

In this section, we prove \sref{Lemma}{lem:baseline}, which shows that the baseline fixed-ratio combination algorithm (\BL) is $(1+\epsilon)$-consistent and $\left( \frac{\nicefrac{(U + 2 \beta)}{L} (\alpha - 1 - \epsilon) + \alpha \epsilon}{( \alpha - 1) } \right)$-robust for \CFL, given any $\epsilon \in [0, \alpha-1]$ and where $\alpha$ is as defined in \eqref{eq:alpha}.  Recall that \sref{Lemma}{lem:baseline} specifies \ALGone as the ``robust algorithm'' to use for the following analysis.

\begin{proof}[Proof of \sref{Lemma}{lem:baseline}]
Under the assumption that \ADV satisfies the long-term constraint, (i.e., that $\sum_{t=1}^{T} c(\mbf{a}_t) \ge 1$), we first observe that the online solution of \BL must also satisfy the long-term constraint.  

Under the assumptions of \CFL, note that $c(\mbf{x})$ is linear (i.e., a weighted $\ell_1$ norm with weight vector $\mbf{c}$).  By definition, denoting the decisions of \ALGone by $\tilde{\mbf{x}}_t$, we know that $\sum_{t=1}^{T} c(\tilde{\mbf{x}}_t) \ge 1$.

Thus, we have the following:
\begin{align*}
\sum_{t=1}^{T} c(\mbf{x}_t) = \sum_{t=1}^{T} c\left( \lambda \mbf{a}_t + (1 - \lambda) \tilde{\mbf{x}}_t \right) = \lambda \sum_{t=1}^{T} c ( \mbf{a}_t ) + (1 - \lambda) \sum_{t=1}^{T} c \left( \tilde{\mbf{x}}_t \right) \geq \lambda + (1 - \lambda) = 1.
\end{align*}
Let $\mathcal{I} \in \Omega$ be an arbitrary valid \CFL sequence.
We denote the \textit{hitting} and \textit{switching} costs of the robust advice by $\ALGone_{\text{hitting}}$ and $\ALGone_{\text{switch}}$, respectively.  Likewise, the hitting and switching cost of the black-box advice \ADV is denoted by $\ADV_{\text{hitting}}$ and $\ADV_{\text{switch}}$.

The total cost of \BL is upper bounded by the following:
\begin{align*}
    \BL(\mathcal{I}) &= \sum_{t=1}^{T} f_t(\mbf{x}_t) + \sum_{t=1}^{T+1} \lVert \mbf{x}_t - \mbf{x}_{t-1} \rVert_{\ell_1 (\mbf{w})}, \\
    &= \sum_{t=1}^{T} f_t \left( \lambda \mbf{a}_t + (1 - \lambda) \tilde{\mbf{x}}_t \right) + \sum_{t=1}^{T+1} \lVert \lambda \mbf{a}_t + (1 - \lambda) \tilde{\mbf{x}}_t - \lambda \mbf{a}_{t-1} - (1 - \lambda) \tilde{\mbf{x}}_{t-1}  \rVert_{\ell_1 (\mbf{w})}, \\
    &\leq \lambda \sum_{t=1}^{T} f_t(\mbf{a}_t) + (1 - \lambda) \sum_{t=1}^{T} f_t( \tilde{\mbf{x}}_t )  + \sum_{t=1}^{T+1} \lVert \lambda \mbf{a}_t - \lambda \mbf{a}_{t-1} \rVert_{\ell_1 (\mbf{w})} + \sum_{t=1}^{T+1} \lVert (1 - \lambda) \tilde{\mbf{x}}_t - (1 - \lambda) \tilde{\mbf{x}}_{t-1} \rVert_{\ell_1 (\mbf{w})},  \\
    &\leq \lambda \ADV_{\text{hitting}}(\mathcal{I}) + (1 - \lambda) \ALGone_{\text{hitting}}(\mathcal{I}) + \lambda \sum_{t=1}^{T+1} \lVert \mbf{a}_t - \mbf{a}_{t-1} \rVert_{\ell_1 (\mbf{w})} + (1 - \lambda) \sum_{t=1}^{T+1} \lVert \tilde{\mbf{x}}_t - \tilde{\mbf{x}}_{t-1} \rVert_{\ell_1 (\mbf{w})},  \\
    &\leq \lambda \ADV_{\text{hitting}}(\mathcal{I}) + (1 - \lambda) \ALGone_{\text{hitting}}(\mathcal{I}) + \lambda \ADV_{\text{switch}}(\mathcal{I}) + (1 - \lambda) \ALGone_{\text{switch}}(\mathcal{I}),  \\
    &\leq \lambda \ADV(\mathcal{I}) + (1 - \lambda) \ALGone(\mathcal{I}).
\end{align*}

Since $\ALGone \leq \alpha \cdot \OPT \leq \alpha \cdot \ADV$, this gives the following:
\begin{align}
\BL(\mathcal{I}) & \leq \lambda \ADV(\mathcal{I}) + (1 - \lambda) \alpha \ADV(\mathcal{I}),\\
\BL(\mathcal{I}) & \leq (\lambda + (1 - \lambda) \alpha) \cdot \ADV(\mathcal{I})\\ \BL(\mathcal{I}) &\leq (1 + \epsilon) \cdot \ADV(\mathcal{I}). \label{eq:min-const}
\end{align}

Furthermore, since $\ADV \leq U + 2\beta \leq \frac{\OPT}{\nicefrac{L}{(U+2\beta)}}$, we have:
\begin{align}
\BL(\mathcal{I}) & \leq \lambda \frac{\OPT(\mathcal{I})}{\nicefrac{L}{(U+2\beta)}} + (1 - \lambda) \alpha \OPT(\mathcal{I}),\\
\BL(\mathcal{I}) & \leq \left[ \frac{\lambda (U + 2\beta)}{L} + (1 - \lambda) \alpha \right] \cdot \OPT(\mathcal{I}),\\ 
\BL(\mathcal{I}) &\leq \left( \frac{\nicefrac{(U + 2 \beta)}{L} (\alpha - 1 - \epsilon) + \alpha \epsilon}{( \alpha - 1) } \right) \cdot \OPT(\mathcal{I}). \label{eq:min-rob}
\end{align}

By combining \eqref{eq:min-const} and \eqref{eq:min-rob}, we conclude that \BL is $(1+\epsilon)$-consistent with respect to black-box advice \ADV, and $\left( \frac{\nicefrac{(U + 2 \beta)}{L} (\alpha - 1 - \epsilon) + \alpha \epsilon}{( \alpha - 1) } \right)$-robust.
\end{proof}

\subsection{Proof of \autoref{thm:constrobCLIP}} \label{apx:constrobCLIP}

In this section, we prove \autoref{thm:constrobCLIP}, which shows that \CLIP is $(1+\epsilon)$-consistent and $\gamma^\epsilon$-robust for \CFL, where $\gamma^\epsilon$ is defined as the solution to the following (as in \eqref{eq:gamma}):
\[
\gamma^\epsilon = \epsilon + \frac{U}{L} - \frac{\gamma^\epsilon}{L} (U-L) \ln \left( \frac{U - L - 2\beta}{U - \nicefrac{U}{\gamma^\epsilon} - 2\beta} \right).
\]
\begin{proof}[Proof of \autoref{thm:constrobCLIP}]
We show the above result by separately considering consistency (the competitive ratio when advice is correct) and robustness (the competitive ratio when advice is not correct) in turn.

Recall that the black-box advice \ADV is denoted by a decision $\mbf{a}_t$ at each time $t$. %
Throughout the following proof, we use shorthand notation $\CLIP_t$ to denote the cost of $\CLIP$ up to time $t$, and $\ADV_t$ to denote the cost of $\ADV$ up to time $t$.
We start with the following lemma to prove consistency.

\begin{lemma} \label{lem:clipconst}
    \CLIP is $(1+\epsilon)$-consistent.
\end{lemma}
\begin{proof}
First, we note that the constrained optimization enforces that the possible cost so far plus a compulsory term is always within $(1 + \epsilon)$ of the advice.  Formally, if time step $j \in [T]$ denotes the time step marking the start of the compulsory trade, we have that the constraint given by \eqref{eq:const-constraint} holds for every time step $t \in [j]$.

Thus, to show $(1+\epsilon)$ consistency, we must resolve the cost during the \textit{compulsory trade} and show that the final cumulative cost of \CLIP is upper bounded by $(1+\epsilon) \ADV$.

Let $\mathcal{I} \in \Omega$ be an arbitrary valid \CFL sequence.
If the compulsory trade begins at time step $j < T$, both $\CLIP$ and $\ADV$ must greedily fill their remaining utilization during the last $m$ time steps $[j, T]$.  This is assumed to be feasible, and the switching cost is assumed to be negligible as long as $m$ is sufficiently large.

Let $(1 - z^{(j-1)})$ denote the remaining long-term constraint that must be satisfied by $\CLIP$ at the final time step, and let $(1 - A^{(j-1)})$ denote the remaining long-term constraint to be satisfied by $\ADV$.

We consider the following two cases, which correspond to the cases where $\CLIP$ has \textit{under}- and \textit{over}- provisioned with respect to \ADV, respectively.

\paragraph{Case 1: $\CLIP(\mathcal{I})$ has ``underprovisioned'' ($(1 - z^{(j-1)}) > (1 - A^{(j-1)})$).} 

In this case, \CLIP must satisfy \textit{more} of the long-term constraint during the compulsory trade compared to \ADV.

From the previous time step, we know that the following constraint holds: $\CLIP_{j-1} + \lVert \mbf{x}_{j-1} - \mbf{a}_{j-1} \rVert_{\ell_1 (\mbf{w})} + \lVert \mbf{a}_{j-1} \rVert_{\ell_1 (\mbf{w})} + (1 - A^{(j-1)} ) L + (A^{(j-1)} - z^{(j-1)})U \le (1+ \epsilon) \left[ \ADV_{j-1} + \lVert \mbf{a}_{j-1} \rVert_{\ell_1 (\mbf{w})} + (1 - A^{(j-1)} ) L \right]$.

Let $\{ \mbf{x}_t \}_{t \in [j, T]}$ and $\{ \mbf{a}_t \}_{t \in [j, T]}$ denote the decisions made by $\CLIP$ and $\ADV$ during the compulsory trade, respectively.  By definition, we have that $\sum_{t=j}^T c(\mbf{x}_t) = (1 - z^{(j-1)})$ and $\sum_{t=j}^T c(\mbf{a}_t)  = (1 - A^{(j-1)})$.

Considering $\{ f_t(\cdot) \}_{t \in [j, T]}$, we know that by definition $\sum_{t=j}^T f_t(\mbf{a}_t) \geq L \sum_{t=j}^T c(\mbf{a}_t)$, and by convex assumptions on the cost functions, $\sum_{t=j}^T f_t(\mbf{x}_t) \leq \sum_{t=j}^T f_t(\mbf{a}_t) + U (\sum_{t=j}^T c(\mbf{x}_t)- \sum_{t=j}^T c(\mbf{a}_t))$.

Note that the worst case for $\CLIP$ occurs when $\sum_{t=j}^T f_t(\mbf{a}_t) = L \sum_{t=j}^T c(\mbf{a}_t)$, as $\ADV$ is able to satisfy the rest of the long-term constraint at the best possible price.  

By the constraint in the previous time step, we have the following:
\begin{align*}
\CLIP_{j-1} + \lVert \mbf{a}_{j-1} \rVert_{\ell_1 (\mbf{w})} + (1 - A^{(j-1)} ) L & + (A^{(j-1)} - z^{(j-1)})U& \\
& \leq (1+\epsilon) [ \ADV_{j-1} + \lVert \mbf{a}_{j-1} \rVert_{\ell_1 (\mbf{w})} + (1 - A^{(j-1)} ) L  ], \\
\CLIP_{j-1} + \lVert \mbf{a}_{j-1} \rVert_{\ell_1 (\mbf{w})} + L  \sum_{t=j}^T c(\mbf{a}_t) & + U \left( \sum_{t=j}^T c(\mbf{x}_t)- \sum_{t=j}^T c(\mbf{a}_t) \right)  \\
& \leq (1+\epsilon) \left[ \ADV_{j-1} + \lVert \mbf{a}_{j-1} \rVert_{\ell_1 (\mbf{w})} + L \sum_{t=j}^T c(\mbf{a}_t) \right],  \\
\CLIP(\mathcal{I}) &\leq (1+\epsilon) \left[ \ADV(\mathcal{I}) \right].
\end{align*}

\paragraph{Case 2: $\CLIP(\mathcal{I})$ has ``overprovisioned'' ($(1 - z^{(j-1)}) \leq (1 - A^{(j-1)})$).}

In this case, \CLIP must satisfy \textit{less} of the long-term constraint during the compulsory trade compared to \ADV.

From the previous time step, we know that the following constraint holds: $\CLIP_{j-1} + \lVert \mbf{x}_{j-1} - \mbf{a}_{j-1} \rVert_{\ell_1 (\mbf{w})} + \lVert \mbf{a}_{j-1} \rVert_{\ell_1 (\mbf{w})} + (1 - A^{(j-1)} ) L + (A^{(j-1)} - z^{(j-1)})U \le (1+ \epsilon) \left[ \ADV_{j-1} + \lVert \mbf{a}_{j-1} \rVert_{\ell_1 (\mbf{w})} + (1 - A^{(j-1)} ) L \right]$.

Let $\{ \mbf{x}_t \}_{t \in [j, T]}$ and $\{ \mbf{a}_t \}_{t \in [j, T]}$ denote the decisions made by $\CLIP$ and $\ADV$ during the compulsory trade, respectively. By definition, we have that $\sum_{t=j}^T c(\mbf{x}_t) = (1 - z^{(j-1)})$ and $\sum_{t=j}^T c(\mbf{a}_t)  = (1 - A^{(j-1)})$.

Considering $\{ f_t(\cdot) \}_{t \in [j, T]}$, we know that by definition, $\sum_{t=j}^T f_t(\mbf{x}_t) \geq L \sum_{t=j}^T c(\mbf{x}_t)$, and $\sum_{t=j}^T f_t(\mbf{a}_t) \geq L \sum_{t=j}^T c(\mbf{a}_t)$.
By convexity, because $\sum_{t=j}^T c(\mbf{x}_t) \leq \sum_{t=j}^T c(\mbf{a}_t)$, $\sum_{t=j}^T f_t(\mbf{x}_t) \leq \sum_{t=j}^T f_t(\mbf{a}_t)$.

By the constraint in the previous time step, we have:
\begin{align*}
\frac{\CLIP_{j-1} + \lVert \mbf{x}_{j-1} - \mbf{a}_{j-1} \rVert_{\ell_1 (\mbf{w})} + \lVert \mbf{a}_{j-1} \rVert_{\ell_1 (\mbf{w})} + (1 - z^{(j-1)}) L}{\ADV_{j-1} + \lVert \mbf{a}_{j-1} \rVert_{\ell_1 (\mbf{w})} + (1 - A^{(j-1)} ) L } = \hspace{10em} &\\
\frac{\CLIP_{j-1} + \lVert \mbf{x}_{j-1} - \mbf{a}_{j-1} \rVert_{\ell_1 (\mbf{w})} + \lVert \mbf{a}_{j-1} \rVert_{\ell_1 (\mbf{w})} + L \sum_{t=j}^T c(\mbf{x}_t)}{\ADV_{j-1} +\lVert \mbf{a}_{j-1} \rVert_{\ell_1 (\mbf{w})} + L \sum_{t=j}^T c(\mbf{a}_t)} &\leq (1 + \epsilon).
\end{align*}

Let $y = \sum_{t=j}^T f_t(\mbf{x}_t) - L \sum_{t=j}^T c(\mbf{x}_t)$, and let $y' = \sum_{t=j}^T f_t(\mbf{a}_t) - L \sum_{t=j}^T c(\mbf{a}_t)$.  By definition, $y \geq 0$ and $y' \geq 0$.

Note that $\CLIP_{j-1} + \lVert \mbf{x}_{j-1} - \mbf{a}_{j-1} \rVert_{\ell_1 (\mbf{w})} + \lVert \mbf{a}_{j-1} \rVert_{\ell_1 (\mbf{w})} + (1 - z^{(j-1)}) L + y \geq \CLIP(\mathcal{I})$ and $\ADV_{j-1} + \lVert \mbf{a}_{j-1} \rVert_{\ell_1 (\mbf{w})} + L \sum_{t=j}^T c(\mbf{a}_t) + y' = \ADV(\mathcal{I})$.

Furthermore, by definition and convexity of the cost functions $f_t( \cdot )$, we have that $y \leq y'$.

Combined with the constraint from the previous time step, we have the following bound:
\begin{align*}
\frac{\CLIP(\mathcal{I})}{\ADV(\mathcal{I})} \leq \frac{\CLIP_{j-1} + \lVert \mbf{x}_{j-1} - \mbf{a}_{j-1} \rVert_{\ell_1 (\mbf{w})} + \lVert \mbf{a}_{j-1} \rVert_{\ell_1 (\mbf{w})} + (1 - z^{(j-1)}) L + y}{\ADV_{j-1} + \lVert \mbf{a}_{j-1} \rVert_{\ell_1 (\mbf{w})} + (1 - A^{(j-1)} ) L + y'} & \\
\leq \frac{\CLIP_{j-1} + \lVert \mbf{a}_{j-1} \rVert_{\ell_1 (\mbf{w})} + L \sum_{t=j}^T c(\mbf{x}_t)}{\ADV_{j-1} + \lVert \mbf{a}_{j-1} \rVert_{\ell_1 (\mbf{w})} + L \sum_{t=j}^T c(\mbf{a}_t)}  &\leq (1 + \epsilon).
\end{align*}

Thus, by combining the bounds in each of the above two cases, the result follows, and we conclude that $\CLIP$ is $(1+\epsilon)$-consistent with accurate advice.
\end{proof}

Having proved the consistency of \CLIP, we proceed to show robustness in the next lemma.

\begin{lemma} \label{lem:cliprob}
    \CLIP is $\gamma^\epsilon$-robust, where $\gamma^\epsilon$ is as defined in \eqref{eq:gamma}.
\end{lemma}
\begin{proof}
Let $\epsilon \in (0, \alpha-1]$ be the target consistency (recalling that \CLIP is $(1+\epsilon)$ consistent), and let $\mathcal{I} \in \Omega$ denote an arbitrary valid \CFL sequence.

To prove the robustness of \CLIP, we consider two ``bad cases'' for the advice $\ADV(\mathcal{I})$, and show that in the worst-case, \CLIP's competitive ratio is bounded by $\gamma^\epsilon$.

\paragraph{Case 1: $\ADV(\mathcal{I})$ is ``inactive''.}

Consider the case where \ADV accepts nothing during the main sequence and instead satisfies the entire long-term constraint in the final time step.  In the worst-case, this gives that $\ADV(\mathcal{I}) = U + 2\beta$.

Based on the consistency constraint (and using the fact that \CLIP will always be ``overprocuring'' w.r.t. \ADV throughout the main sequence), we can derive an upper bound on the amount that \CLIP is allowed to accept from the robust pseudo-cost minimization.  Recall the following constraint:
\begin{align*}
\CLIP_{t-1} + f_t(\mbf{x}_t) + \lVert \mbf{x}_t - \mbf{x}_{t-1} \rVert_{\ell_1 (\mbf{w})} + \lVert \mbf{x}_t - \mbf{a}_t \rVert_{\ell_1 (\mbf{w})} + \lVert \mbf{a}_t \rVert_{\ell_1 (\mbf{w})}& + (1 - z^{(t-1)} - c( \mbf{x}_t) )L\\
&\le (1+ \epsilon) \left[ \ADV_t + \lVert \mbf{a}_t \rVert_{\ell_1 (\mbf{w})} + (1 - A^{(t)} ) L \right].
\end{align*}

\begin{proposition}
    $z_{\text{PCM}}$ is an upper bound on the amount that \CLIP can accept from the pseudo-cost minimization without violating $(1+\epsilon)$ consistency, and is defined as:
\[
z_{\text{PCM}} = \gamma^\epsilon \ln \left[ \frac{U - L - 2\beta}{U - \nicefrac{U}{\gamma^\epsilon} - 2\beta} \right]
\]
\end{proposition}
\begin{proof}
Consider an arbitrary time step $t$.  When \CLIP is \textit{not} allowed to accept anything more from the robust pseudo-cost minimization, we have that $c(\mbf{x}_t)$ is restricted to be $0$ (recall that $\mbf{a}_t = \mbf{0}$ for any time steps before $T$, because the advice is assumed to be inactive).

By definition, since any cost functions accepted in $\CLIP_{t-1}$ can be attributed to the robust pseudo-cost minimization, we have the following in the worst-case:
\begin{align*}
\CLIP_{t-1} = \int_0^{z^{(t-1)}} \phi^\epsilon (u) du + \beta z^{(t-1)}.
\end{align*}

Combining the above with the left-hand side of the consistency constraint, we have the following by observing that $\mbf{x}_t = \mbf{0}$ and $\mbf{a}_t = \mbf{0}$, and the switching cost to ``ramp-up'' is absorbed into the pseudo-cost $\phi$:
\begin{align*}
\CLIP_{t-1} + (1 - z^{(t-1)})L &= \int_0^{z^{(t-1)}} \phi^\epsilon (u) du + \beta z^{(t-1)} + (1 - z^{(t-1)})L.
\end{align*}

As stated, let $z^{(t-1)} = z_{\text{PCM}}$.
Then by properties of the pseudo-cost,  
\begin{align*}
\CLIP_{t-1} + (1 - z_{\text{PCM}})L &= \int_0^{z_{\text{PCM}}} \phi (u) du + \beta z_{\text{PCM}} + (1 - z_{\text{PCM}})U + (1 - z_{\text{PCM}})L - (1 - z_{\text{PCM}})U, \\
&= \gamma^\epsilon \left[ \phi^\epsilon (z_{\text{PCM}}) - \beta \right] + (1 - z_{\text{PCM}})L - (1 - z_{\text{PCM}})U, \\
&= \gamma^\epsilon L + \left(L - U \right) \left(1 - \gamma^\epsilon \ln \left[ \frac{U - L - 2\beta}{U - \nicefrac{U}{\gamma^\epsilon} - 2\beta} \right] \right), \\
&= \gamma^\epsilon L + L - U - \left(L-U\right) \gamma^\epsilon \ln \left[ \frac{U - L - 2\beta}{U - \nicefrac{U}{\gamma^\epsilon} - 2\beta} \right].
\end{align*}

Substituting for the definition of $\gamma^\epsilon$, we obtain:
\begin{align*}
\CLIP_{t-1} + (1 - z_{\text{PCM}})L &= \gamma^\epsilon L + L - U - \left(L-U\right) \gamma^\epsilon \ln \left[ \frac{U - L - 2\beta}{U - \nicefrac{U}{\gamma^\epsilon} - 2\beta} \right], \\
&= \left[ \epsilon L + U - \gamma^\epsilon (U-L) \ln \left[ \frac{U - L - 2\beta}{U - \nicefrac{U}{\gamma^\epsilon} - 2\beta} \right] \right] + L - U + \left(U-L\right) \gamma^\epsilon \ln \left[ \frac{U - L - 2\beta}{U - \nicefrac{U}{\gamma^\epsilon} - 2\beta} \right], \\
&= \epsilon L + L = (1+\epsilon) L.
\end{align*}

As $(1+\epsilon) L $ is exactly the right-hand side of the consistency constraint (i.e., $(1+\epsilon) \left[ \ADV_t + \lVert \mbf{a}_t \rVert_{\ell_1 (\mbf{w})} + (1-A_t) L \right] = (1+\epsilon) L$), this completes the proposition.
\end{proof}

If \CLIP is constrained to use at most $z_{\text{PCM}}$ of its utilization to be robust, the remaining $(1-z_{\text{PCM}})$ utilization must be used for the compulsory trade and/or to follow $\ADV$.  Thus, we have the following worst-case competitive ratio for \CLIP, specifically for Case 2:
\begin{align*}
\frac{\CLIP(\mathcal{I})}{\OPT(\mathcal{I})} &\le \frac{ \int_0^{z_{\text{PCM}}} \phi^\epsilon(u) du + \beta z_{\text{PCM}} + (1 - z_{\text{PCM}})U}{L}
\end{align*}

By the definition of $\phi^\epsilon(p)$, we have the following:
\begin{align*}
\frac{\CLIP(\mathcal{I})}{\OPT(\mathcal{I})} &\le \frac{ \int_0^{z_{\text{PCM}}} \phi^\epsilon(u) du + \beta z_{\text{PCM}} + (1 - z_{\text{PCM}})U}{L}\\
&\le \frac{ \gamma^\epsilon \left[ \phi^\epsilon(z_{\text{PCM}}) - \beta \right]}{L} \le \frac{ \gamma^\epsilon \left[ L + \beta - \beta \right]}{L} \ \leq \ \gamma^\epsilon.
\end{align*}

\paragraph{Case 2: $\ADV(\mathcal{I})$ is ``overactive''.}

We now consider the case where \ADV accepts bad cost functions which it ``should not'' accept (i.e. $\ADV(\mathcal{I}) \gg \OPT(\mathcal{I})$).
Let $\ADV(\mathcal{I}) = v \gg \OPT_T$ (i.e. the final total hitting and switching cost of \ADV is $v$ for some $v \in [L, U+2\beta]$, and this is much greater than the optimal solution).

This is without loss of generality, since we can assume that $v$ is the ``best cost function'' accepted by \ADV and the consistency ratio changes strictly in favor of \ADV.
Based on the consistency constraint, we can derive a lower bound on the amount that \CLIP \textit{must} accept from \ADV in order to stay $(1+\epsilon)$-consistent.

To do this, we consider the following sub-cases:

\noindent $\bullet$ \textbf{Sub-case 2.1}: Let $v \geq \frac{U+\beta}{1+\epsilon}$.

In this sub-case, \CLIP can fully ignore the advice, because the following consistency constraint is never binding (note that $\ADV_t \geq \frac{U+\beta}{1+\epsilon} A^{(t)}$):
\begin{align*}
 \CLIP_{t-1} + f_t(\mbf{x}_t) + \lVert \mbf{x}_t - \mbf{x}_{t-1} \rVert_{\ell_1 (\mbf{w})} + \lVert \mbf{x}_t - \mbf{a}_{t} \rVert_{\ell_1 (\mbf{w})}& + \lVert \mbf{a}_t \rVert_{\ell_1 (\mbf{w})} + (1 - A^{(t)} ) L + (A^{(t)} - z^{(t-1)} - c(\mbf{x}_t))U\\
 &\le (1+ \epsilon) \left[ \ADV_t + \lVert \mbf{a}_t \rVert_{\ell_1 (\mbf{w})} + (1 - A^{(t)} ) L \right], \\
(1 - A^{(t)} ) L + (A^{(t)})U + \lVert \mbf{a}_t \rVert_{\ell_1 (\mbf{w})}  &\le (1+ \epsilon) \left[ \ADV_t + \lVert \mbf{a}_t \rVert_{\ell_1 (\mbf{w})} + (1 - A^{(t)} ) L\right],  \\
 (1 - A^{(t)} ) L + U A^{(t)} + \beta A^{(t)} &\le (1+ \epsilon) \left[ \frac{U+\beta}{1+\epsilon} A^{(t)} + (1 - A^{(t)} ) L \right]
\end{align*}

\noindent $\bullet$ \textbf{Sub-case 2.2}: Let $v \in ( L, \ \frac{U+\beta}{1+\epsilon} )$.

To remain $(1 + \epsilon)$ consistent, \CLIP must accept some of these ``bad cost functions'' denoted by $v$ in the worst-case.
We would like to derive a lower bound $z_\ADV$, such that $z_\ADV$ describes the minimum amount that \CLIP must accept from \ADV in order to always satisfy the $(1+\epsilon)$ consistency constraint.

Based on the consistency constraint, we have the following:
\begin{align*}
\CLIP_{t-1} + f_t(\mbf{x}_t) + \lVert \mbf{x}_t - \mbf{x}_{t-1} \rVert_{\ell_1 (\mbf{w})} + \lVert \mbf{x}_t - \mbf{a}_{t} \rVert_{\ell_1 (\mbf{w})}& + \lVert \mbf{a}_t \rVert_{\ell_1 (\mbf{w})} + (1 - A^{(t)} ) L + (A^{(t)} - z^{(t-1)} - c(\mbf{x}_t))U\\
&\le (1+ \epsilon) \left[ \ADV_t + \lVert \mbf{a}_t \rVert_{\ell_1 (\mbf{w})} + (1 - A^{(t)} ) L \right]. \\
\end{align*}

We let $f_t(\mbf{x}_t) + \lVert \mbf{x}_t - \mbf{x}_{t-1} \rVert_{\ell_1 (\mbf{w})} + \lVert \mbf{x}_t - \mbf{a}_t \rVert_{\ell_1 (\mbf{w})} + \lVert \mbf{a}_t \rVert_{\ell_1 (\mbf{w})} \leq v c(\mbf{x}_t)$ for any $\mbf{x}_t : c(\mbf{x}_t) < c(\mbf{a}_t)$, which holds by convexity of the cost functions $f_t(\cdot)$ and a prevailing condition that $c( \mbf{x}_t ) \leq c( \mbf{a}_t )$ for the ``bad cost functions'' accepted by \ADV. Note that $v - U$ is negative (by the condition of Sub-case 2.2):

\begin{align*}
\CLIP_{t-1} + v c(\mbf{x}_t)& + L - L A^{(t)} + U A^{(t)} - U z^{(t-1)} - U c(\mbf{x}_t) \le (1+ \epsilon) \left[ vA^{(t-1)} + vc(\mbf{a}_t) + L - L A^{(t)}  \right],\\
v c(\mbf{x}_t) - U c(\mbf{x}_t) &\le (1+ \epsilon) \left[ vA^{(t-1)} + vc(\mbf{a}_t) + L - L A^{(t)}  \right] - \CLIP_{t-1} - L + L A^{(t)} - U A^{(t)} + U z^{(t-1)}, \\
v c(\mbf{x}_t) - U c(\mbf{x}_t) &\le vA^{(t)} - U A^{(t)} - \CLIP_{t-1} + U z^{(t-1)} + \epsilon \left[ vA^{(t-1)} + vc(\mbf{a}_t) + L - L A^{(t)}  \right], \\
c(\mbf{x}_t) &\ge \frac{vA^{(t)} - U A^{(t)} - \CLIP_{t-1} + U z^{(t-1)} + \epsilon \left[ vA^{(t)} + L - L A^{(t)}  \right]}{v-U}.
\end{align*}

In the event that $A^{(t-1)} = 0$ (i.e. nothing has been accepted so far by either \ADV or \CLIP), we have the following:
\begin{align*}
c(\mbf{x}_t) &\ge \frac{vc(\mbf{a}_t) - U c(\mbf{a}_t) + \epsilon \left[ vc(\mbf{a}_t) + L - L c(\mbf{a}_t)  \right]}{v-U}, \\
c(\mbf{x}_t) &\ge c(\mbf{a}_t) - \frac{\epsilon \left[ vc(\mbf{a}_t) + L - L c(\mbf{a}_t)  \right]}{U-v}.
\end{align*}

Through a recursive definition, we can show that for any $A^{(t)}$, given that \CLIP has accepted $z^{(t-1)}$ of \ADV's suggested prices so far, it must set $\mbf{x}_t$ such that:
\begin{align*}
z^{(t)} &\ge z^{(t-1)} + c(\mbf{a}_t) - \frac{\epsilon \left[ vc(\mbf{a}_t) + L - L c(\mbf{a}_t)  \right]}{U-v}.
\end{align*}

Continuing the assumption that $v$ is constant, if \CLIP has accepted $z^{(t-1)}$ thus far, we have the following if we assume that the acceptance up to this point happened in a single previous time step $m$:
\begin{align*}
c(\mbf{x}_t) &\ge A^{(t)} + \frac{U c(\mbf{x}_m) - \CLIP_{t-1} + \epsilon \left[ vA^{(t)} + L - L A^{(t)}  \right]}{v-U}, \\
c(\mbf{x}_t) &\ge c(\mbf{a}_t) + c(\mbf{a}_m) + \frac{U c(\mbf{x}_m) - v c(\mbf{x}_m) + \epsilon \left[ v(c(\mbf{a}_t) + c(\mbf{a}_m)) + L - L (c(\mbf{a}_t) + c(\mbf{a}_m))  \right]}{v-U}, \\
c(\mbf{x}_t) &\ge c(\mbf{a}_t) + c(\mbf{a}_m) - \mbf{x}_m + \frac{\epsilon \left[ v(c(\mbf{a}_t) + c(\mbf{a}_m)) + L - L (c(\mbf{a}_t) + c(\mbf{a}_m))  \right]}{v-U}, \\
c(\mbf{x}_t) + c(\mbf{x}_m) &\ge c(\mbf{a}_t) + c(\mbf{a}_m) + \frac{\epsilon \left[ v(c(\mbf{a}_t) + c(\mbf{a}_m)) + L - L (c(\mbf{a}_t) + c(\mbf{a}_m))  \right]}{v-U}, \\
z^{(t)} &\ge A^{(t)} + \frac{\epsilon \left[ vA^{(t)} + L - L A^{(t)}  \right]}{v-U}.
\end{align*}

This gives intuition into the desired $z_\ADV$ bound.  The above describes and motivates that the \textit{aggregate} acceptance by \CLIP at any given time step $t$ must satisfy a lower bound.  Consider that the worst case for Sub-case 2.2 occurs when all of the $v$ prices accepted by \ADV arrive first, before any prices which would be considered by the pseudo-cost minimization.  Then let $A^{(t)} = 1$ for some arbitrary time step $t$, and we have the following lower bound on $z_\ADV$:
\begin{align*}
z_\ADV &\geq 1 - \frac{ v\epsilon }{U - v}.
\end{align*}

If \CLIP is forced to use $z_\ADV$ of its utilization to be $(1+\epsilon)$ consistent against \ADV, that leaves at most $(1-z_\ADV)$ utilization for robustness.  

We define $z' = \min ( 1 - z_\ADV, z_{\text{PCM}})$ and consider the following two cases.  

\noindent $\bullet$ \textbf{Sub-case 2.2.1}: if $z' = z_{\text{PCM}}$, the worst-case competitive ratio is bounded by the following.  Note that if $z' = z_{\text{PCM}}$, the amount of utilization that \CLIP can use to ``be robust'' is exactly the same as in \textbf{Case 1}:
\begin{align*}
\frac{\CLIP(\mathcal{I})}{\OPT(\mathcal{I})} &\leq \frac{\int_0^{z_{\text{PCM}}} \phi(u) du + \beta z_{\text{PCM}} + (1 -z_{\ADV} - z_{\text{PCM}})U + z_{\ADV} v}{L}, \\
&\leq \frac{\int_0^{z_{\text{PCM}}} \phi(u) du + \beta z_{\text{PCM}} + (1 - z_{\text{PCM}})U}{L} \leq \gamma^\epsilon.
\end{align*}

\noindent $\bullet$ \textbf{Sub-case 2.2.2}: if $z' = 1 - z_{\ADV}$, the worst-case competitive ratio is bounded by the following.  Note that \CLIP \textit{cannot} use $z_{\text{PCM}}$ of its utilization for robustness, so the following bound assumes that the cost functions accepted by \CLIP are bounded by the \textit{worst $(1 - z_{\ADV})$ fraction} of the pseudo-cost threshold function $\phi^\epsilon$ (which follows since $\phi^\epsilon$ is non-decreasing on $z \in [0,1]$):
\begin{align*}
\frac{\CLIP(\mathcal{I})}{\OPT(\mathcal{I})} \leq \frac{\int_0^{1 - z_{\ADV}} \phi(u) du + \beta (1 - z_{\ADV}) + z_{\ADV} v}{L}.
\end{align*}

Note that if $z' = 1 - z_\ADV$, we know that $1 - z_\ADV < z_{\text{PCM}}$, which further gives the following by definition of $z_\ADV$:
\begin{align*}
    1-z_{\text{PCM}} &< 1 - \frac{v\epsilon}{U-v}, \\
    v\epsilon &< (U-v) z_{\text{PCM}}, \\
    v &< \frac{U}{(1+ \frac{\epsilon}{z_{\text{PCM}}})}.
\end{align*}

\noindent By plugging $v$ back into the definition of $z_\ADV$, we have that $z_\ADV v \leq \left( \frac{(1-z_{\text{PCM}})U}{1 + \frac{\varepsilon}{z_{\text{PCM}}}} \right)$, giving the following:
\begin{align*}
\frac{\CLIP(\mathcal{I})}{\OPT(\mathcal{I})} &\leq \frac{\int_0^{1 - z_{\ADV}} \phi(u) du + \beta (1 - z_{\ADV}) + \left( \frac{(1-z_{\text{PCM}})U}{1 + \frac{\epsilon}{z_{\text{PCM}}}} \right) }{L}, \\
&\leq \frac{\int_0^{z_{\text{PCM}}} \phi(u) du + \beta z_{\text{PCM}} + (1 - z_{\text{PCM}})U}{L} \leq \gamma^\epsilon.
\end{align*}

Thus, by combining the bounds in each of the above two cases, the result follows, and we conclude that $\CLIP$ is $\gamma^\epsilon$-robust for any advice \ADV.
\end{proof}

Having proven \sref{Lemma}{lem:clipconst} (consistency) and \sref{Lemma}{lem:cliprob} (robustness), the statement of \autoref{thm:constrobCLIP} follows -- \CLIP is $(1+\epsilon)$-consistent and $\gamma^\epsilon$-robust given any advice for \CFL.
\end{proof}

\subsection{Proof of \sref{Corollary}{cor:constrobCLIPMAL}} \label{apx:constrobCLIPMAL}

In this section, we prove \sref{Corollary}{cor:constrobCLIPMAL}, which shows that \CLIP is $(1+\epsilon)$-consistent and $\gamma^\epsilon$-robust for \MAL, where $\gamma^\epsilon$ is defined in \eqref{eq:gamma}.

\begin{proof}[Proof of \sref{Corollary}{cor:constrobCLIPMAL}]
We show the above result by separately considering consistency (the competitive ratio when advice is correct) and robustness (the competitive ratio when advice is not correct), relying on the proof of \autoref{thm:constrobCLIP}.

\noindent\textbf{Consistency.}  By definition, \MAL on a weighted star metric is identical to an instance of \CFLno on $(\Delta_n, \lVert \cdot \rVert_{\ell_1 (\mbf{w'})})$, where $\Delta_n$ is the $n$-point simplex in $\mathbb{R}^{n}$ and $\lVert \cdot \rVert_{\ell_1 (\mbf{w'})}$ is the weighted $\ell_1$ norm, with weights $\mbf{w'}$ given by the corresponding edge weight in the underlying star metric.

Observe that the consistency proof given in \sref{Lemma}{lem:clipconst} holds when the consistency constraint at each time step is defined as follows:
\begin{align}
\scriptsize
\begin{split}
&\CLIP_{t-1} + f_t(\mbf{x}) + \lVert \mbf{x} - \mbf{x}_{t-1} \rVert_{\ell_1 (\mbf{w'})} + \lVert \mbf{x} - \mbf{a}_{t} \rVert_{\ell_1 (\mbf{w'})} + \lVert \mbf{a}_{t} \rVert_{\ell_1 (\mbf{w'})} + (1 - z^{(t-1)} - c(\mbf{x}))L + \max( (A^{(t)} - z^{(t-1)} - c(\mbf{x})), \ 0)(U-L)\\ 
& \hspace{43.5em} \leq (1+\epsilon) [\ADV_{t} + \lVert \mbf{a}_t \rVert_{\ell_1 (\mbf{w'})} + (1 - A^{(t)} ) L],
\end{split}
\end{align}
where $\mbf{x}$ and $\mbf{a}$ denote decisions by \CLIP and \ADV (respectively) supported on $\Delta_n$.
Thus, since the consistency proof in \sref{Lemma}{lem:clipconst} exactly holds under the \CFL vector space corresponding to \MAL, we conclude that \CLIP is $(1+\epsilon)$-consistent for \MAL.

\smallskip

\noindent\textbf{Robustness.}  First, we note that the robustness proof given in \sref{Lemma}{lem:cliprob} assumes \OPT does not pay any switching cost.  This implies that the proof of \sref{Lemma}{lem:cliprob} meets the conditions of \sref{Proposition}{prop:CFL-MAL}, which states that any performance bound for an arbitrary \ALG solving \CFL which assumes \OPT pays no switching cost translates to an identical bound for \MAL, where the problem's parameters can be recovered by constructing a corresponding \CFL instance according to \sref{Lemma}{lem:simplexTransform}.

Thus, by \sref{Proposition}{prop:CFL-MAL}, we conclude that \CLIP is $\gamma^\epsilon$-robust for \MAL, where $\gamma^\epsilon$ is defined in \eqref{eq:gamma}.

\smallskip

By combining the two results, the statement of \sref{Corollary}{cor:constrobCLIPMAL} follows -- \CLIP is $(1+\epsilon)$-consistent and $\gamma^\epsilon$-robust given any advice \ADV for \MAL.
\end{proof}

\subsection{Proof of \autoref{thm:optimalconstrobCFL}} \label{apx:optimalconstrobCFL}

In this section, we prove \autoref{thm:optimalconstrobCFL}, which shows that any $(1+\epsilon)$-consistent algorithm for \CFL is at least $\gamma^\epsilon$-robust, where $\gamma^\epsilon$ is as defined in \eqref{eq:gamma}.

\begin{proof}[Proof of \autoref{thm:optimalconstrobCFL}]
    To show this result, we leverage the same special family of $y$-adversaries for \CFL defined in Definition~\ref{dfn:yadversary}, where $y \in [L,U]$.  Recall that $k = \argmax_{i \in [d]} \mbf{w}_i$, where $\mbf{w}$ is the weight vector for $\lVert \cdot \rVert_{\ell_1 (\mbf{w})}$.

    As in the proof of \autoref{thm:lowerboundCFL}, we note that for adversary $\mathcal{A}_y$, the optimal offline solution is $\OPT(\mathcal{A}_y) = y + \nicefrac{2\beta}{m}$, and that as $m$ grows large, $\OPT(\mathcal{A}_y) \to y$.
    
    Against these adversaries, we consider two types of advice -- the first is \textit{bad} advice, which sets $\mbf{a}_t = \mbf{0}$ for all time steps $t < T$ (i.e., before the compulsory trade), incurring a final cost of $U + 2\beta$.  

    On the other hand, \textit{good} advice sets $\mbf{a}_t = \mbf{0}$ for all time steps up to the first time step when $y$ is revealed, at which point it sets $\mbf{a}_t^k = \nicefrac{1}{m}$ to achieve final cost $\ADV(\mathcal{A}_y) = \OPT(\mathcal{A}_y) = y + \nicefrac{2\beta}{m}$.

    We let $g(y)$ denote a \textit{robust conversion function} $[L,U] \rightarrow [0,1]$, which fully quantifies the actions of a learning-augmented algorithm $\LALG$ playing against adaptive adversary $\mathcal{A}_y$, where $g(y)$ gives the progress towards the long-term constraint under the instance $\mathcal{A}_y$ before (either) the compulsory trade or the black-box advice sets $\mbf{a}_t^k > 0$.  Note that for large $w$, the adaptive adversary $\mathcal{A}_{y-\delta}$ is equivalent to first playing $\mathcal{A}_y$ (besides the last two batches of cost functions), and then processing batches with cost functions $\textbf{Down}^{w_y+1}(\mbf{x})$ and $\textbf{Up}(\mbf{x})$.  Since $\LALG$ is deterministic and the conversion is unidirectional (irrevocable), we must have that $g(y - \delta) \geq g(y)$, i.e. $g(y)$ is non-increasing in $[L, U]$.  

    As in the proof of \autoref{thm:lowerboundCFL}, the adaptive nature of each $y$-adversary forces any algorithm to incur a switching cost proportional to $g(y)$, specifically denoted by $2\beta g(y)$. 

    For any $\gamma$-robust online algorithm $\LALG$ given any arbitrary black-box advice, the following must hold:
    \[
    \LALG(\mathcal{A}_y) \leq \gamma \OPT(\mathcal{A}_y) = \gamma y, \ \forall y \in [L, U].
    \]

    The cost of $\LALG$ with conversion function $g$ on an instance $\mathcal{A}_y$ is $\LALG(\mathcal{A}_y) = g(\nicefrac{U}{\gamma}) \nicefrac{U}{\gamma} - \int^y_{\nicefrac{U}{\gamma}} u d g(u) + 2\beta g(y) + (1 - g(y))U$, where $udg(u)$ is the cost of buying $dg(u)$ utilization at price $u$, the last term is from the compulsory trade, and the second to last term is the switching cost incurred by $\LALG$. 

    This implies that $g(y)$ must satisfy the following:
    \[
    g(\nicefrac{U}{\gamma}) \nicefrac{U}{\gamma} - \int^y_{\nicefrac{U}{\gamma}} u d g(u) + 2\beta g(y) + (1 - g(y))U \leq \gamma y, \ \forall y \in [L, U].
    \]

    By integral by parts, the above implies that the conversion function must satisfy $g(y) \geq \frac{U - \gamma y}{U - y - 2\beta} - \frac{1}{U - y - 2\beta} \int_{\nicefrac{U}{\gamma}}^y g(u) du$.  By Gr\"{o}nwall's Inequality \cite{Mitrinovic:91}[Theorem 1, p. 356], we have that
    \begin{align}
    g(y) & \geq \frac{U - \gamma y}{U - y - 2\beta} - \frac{1}{U - y - 2\beta} \int_{\nicefrac{U}{\gamma}}^y \frac{U - \gamma u}{U - u - 2\beta} \cdot \exp\left( \int_u^y \frac{1}{U - r - 2\beta} dr \right) du \\
    & \geq \frac{U - \gamma y}{U - y - 2\beta} - \int_{\nicefrac{U}{\gamma}}^y \frac{U - \gamma u}{(U - u - 2\beta)^2} du \\
    & \geq \frac{U - \gamma y}{U - y - 2\beta} - \left[ \frac{U\gamma - U - 2\beta \gamma}{u + 2\beta - U} - \gamma \ln \left( u + 2\beta -U \right) \right]_{\nicefrac{U}{\gamma}}^y \\
    & \geq \gamma \ln \left( y + 2\beta -U \right) - \gamma \ln \left( \nicefrac{U}{\gamma} + 2\beta -U \right), \quad \forall y \in [L, U]. \label{eq:rob-bound-gron-min}
    \end{align}

    In addition, to simultaneously be $\eta$-consistent when the advice \textit{is} correct, $\LALG$ must satisfy $\LALG(\mathcal{A}_L) \leq \eta \OPT(\mathcal{A}_L) = \eta L$.  If the advice is correct (and $m$ is sufficiently large), we assume that $\LALG$ pays no switching cost to satisfy the long-term constraint at the best cost functions $L$.  It must still pay for switching incurred by the robust algorithm (recall that $\OPT$ pays no switching cost).
    \begin{align}
    \int_{\nicefrac{U}{\gamma}}^L g(u) du + 2\beta g(L) &\leq \eta L - L. \label{eq:const-bound-gron-min}
    \end{align}

    By combining equations \eqref{eq:rob-bound-gron-min} and \eqref{eq:const-bound-gron-min}, the conversion function $g(y)$ of any $\gamma$-robust and $\eta$-consistent online algorithm must satisfy the following:
    \begin{align}
    \gamma \int_{\nicefrac{U}{\gamma}}^L \ln \left( \frac{u + 2\beta -U}{ \nicefrac{U}{\gamma} + 2\beta -U } \right)  du + 2\beta \left[ \gamma \ln \left( \frac{u + 2\beta -U}{ \nicefrac{U}{\gamma} + 2\beta -U } \right)  \right] &\leq \eta L - L.
    \end{align}
    When all inequalities are binding, this equivalently gives that 
    \begin{align}
    \eta \geq \gamma +1-\frac{U}{L} + \frac{\gamma(U-L)}{L}  \ln \left( \frac{U - L - 2\beta}{U - \nicefrac{U}{\gamma^\epsilon} - 2\beta} \right). \label{eq:eta-gamma}
    \end{align}
    We define $\eta$ such that $\eta \coloneqq (1 + \epsilon)$.  By substituting for $\eta$ into \eqref{eq:eta-gamma}, we recover the definition of $\gamma^\epsilon$ as given by \eqref{eq:gamma}, which subsequently completes the proof.
    Thus, we conclude that any $(1+\epsilon)$-consistent algorithm for \CFL is at least $\gamma^\epsilon$-robust.
\end{proof}

\subsection{Proof of \sref{Corollary}{cor:optimalconstrobMAL}} \label{apx:optimalconstrobMAL}

In this section, we prove \sref{Corollary}{cor:optimalconstrobMAL}, which shows that any $(1+\epsilon)$-consistent algorithm for \MAL is at least $\gamma^\epsilon$-robust, where $\gamma^\epsilon$ is as defined in \eqref{eq:gamma}.

\begin{proof}[Proof of \sref{Corollary}{cor:optimalconstrobMAL}]
    To show this result, we leverage the same special family of $y$-adversaries for \CFL defined in Definition~\ref{dfn:yadversaryMAL}, where $y \in [L,U]$.  Recall that $k = \argmax_{a \in [n]} \mbf{w}^a$, deonotes the largest edge weight of any (non-\OFF) point in the metric space, and $\beta = \mbf{w}^k$.

    As in the proof of \autoref{thm:lowerboundCFL}, we note that for adversary $\mathcal{A}_y$, the optimal offline solution is $\OPT(\mathcal{A}_y) = y + \nicefrac{2\beta}{m}$, and that as $m$ grows large, $\OPT(\mathcal{A}_y) \to y$.
    
    Against these adversaries, we consider two types of advice -- the first is \textit{bad} advice, which sets $\mbf{a}^{a'}_t = 1$ (i.e., \ADV stays in the \OFF point) for all time steps $t < T$ (i.e., before the compulsory trade), incurring a final cost of $U + 2\beta$.  

    On the other hand, \textit{good} advice sets $\mbf{a}^{a'}_t = 1$ for all time steps up to the first time step when $y$ is revealed, at which point it sets $\mbf{a}_t^k = \nicefrac{1}{m}$ to achieve final cost $\ADV(\mathcal{A}_y) = \OPT(\mathcal{A}_y) = y + \nicefrac{2\beta}{m}$.

As previously, we let $g(y)$ denote a \textit{robust conversion function} $[L,U] \rightarrow [0,1]$, which fully quantifies the actions of a learning augmented algorithm $\LALG$ playing against adaptive adversary $\mathcal{A}_y$.  %
Since $\LALG$ is deterministic and the conversion is unidirectional (irrevocable), %
$g(y)$ is non-increasing in $[L, U]$.  Intuitively, the entire long-term constraint should be satisfied if the minimum possible price is observed, i.e $g(L) = 1$.

As in \autoref{thm:optimalconstrobCFL}, the adaptive nature of each $y$-adversary forces any deterministic $\ALG$ to incur a switching cost of $2\beta g(y)$ on adversary $\mathcal{A}_y$, and we assume that $\ALG$ does not incur a significant switching cost during the final batch (i.e., during the compulsory trade).

For any $\gamma$-robust \LALG given any arbitrary black-box advice, the following must hold:
\[
\LALG(\mathcal{A}_y) \leq \gamma \OPT(\mathcal{A}_y) = \gamma y, \ \forall y \in [L, U].
\]

The cost of $\LALG$ with conversion function $g$ on an instance $\mathcal{A}_y$ is $\LALG(\mathcal{A}_y) = g(\nicefrac{U}{\gamma}) \nicefrac{U}{\gamma} - \int^y_{\nicefrac{U}{\gamma}} u d g(u) + 2\beta g(y) + (1 - g(y))U$, where $udg(u)$ is the cost of buying $dg(u)$ utilization at price $u$, the last term is from the compulsory trade, and the second to last term is the switching cost incurred by $\LALG$.  Note that this expression for the cost is exactly as defined in \autoref{thm:optimalconstrobCFL}.

Thus by \autoref{thm:optimalconstrobCFL}, for any learning-augmented algorithm \LALG which is simultaneously $\eta$-consistent and $\gamma$-robust, the conversion function $g(\cdot)$ must satisfy the following inequality (via integral by parts and Gr\"{o}nwall's Inequality \citep[Theorem 1, p. 356]{Mitrinovic:91}):
\begin{align}
\gamma \int_{\nicefrac{U}{\gamma}}^L \ln \left( \frac{u + 2\beta -U}{ \nicefrac{U}{\gamma} + 2\beta -U } \right)  du + 2\beta \left[ \gamma \ln \left( \frac{u + 2\beta -U}{ \nicefrac{U}{\gamma} + 2\beta -U } \right)  \right] &\leq \eta L - L.
\end{align}

When all inequalities are binding, this equivalently gives that the optimal $\eta$ and $\gamma$ satisfy:
\begin{align}
\eta \geq \gamma +1-\frac{U}{L} + \frac{\gamma(U-L)}{L}  \ln \left( \frac{U - L - 2\beta}{U - \nicefrac{U}{\gamma^\epsilon} - 2\beta} \right).
\end{align}
We define $\eta$ such that $\eta \coloneqq (1 + \epsilon)$.  By substituting for $\eta$ into \eqref{eq:eta-gamma}, we recover the definition of $\gamma^\epsilon$ as given by \eqref{eq:gamma}, which subsequently completes the proof.
Thus, we conclude that any $(1+\epsilon)$-consistent algorithm for \CFL is at least $\gamma^\epsilon$-robust.
\end{proof}

\end{document}